\documentclass[journal, onecolumn]{IEEEtran}
\IEEEoverridecommandlockouts

\usepackage{bm}
\usepackage{amsmath,amssymb,amsfonts,graphicx,amsthm,mathtools,nicefrac,mathrsfs}
\usepackage{graphicx, subfigure}
\usepackage{float}
\usepackage{times}
\usepackage{color}
\usepackage[table]{xcolor}
\usepackage{booktabs}
\usepackage{multirow}
\usepackage{lscape}
\usepackage{ulem}
\DeclareMathOperator{\E}{\mathbb{E}}
\DeclareMathOperator{\dist}{dist}

\DeclareMathOperator{\cone}{cone}

\renewcommand{\Pr}[2][]{\mathbb{P}_{#1} \left\{ #2 \rule{0mm}{3mm}\right\}}
\newcommand{\ip}[2]{\left\langle#1,#2\right\rangle}

\def \CC {\mathcal{C}}

\def \NN {\mathcal{N}}

\def \SS {\mathcal{S}}
\def \TT {\mathcal{T}}

\def \OO {\mathcal{O}}

\def \B {\mathbb{B}}
\def \R {\mathbb{R}}
\def \S {\mathbb{S}}

\def \va {\bm{a}}
\def \vb {\bm{b}}
\def \vd {\bm{d}}

\def \vg {\bm{g}}
\def \vh {\bm{h}}

\def \vr {\bm{r}}
\def \vs {\bm{s}}
\def \vu {\bm{u}}
\def \vv {\bm{v}}
\def \vw {\bm{w}}
\def \vx {\bm{x}}
\def \vy {\bm{y}}
\def \vz {\bm{z}}

\def \vzero {\bm{0}}

\def \mA {\bm{A}}

\def \mI {\bm{I}}

\def \mG {\bm{G}}

\def \mX {\bm{X}}
\def \mU {\bm{U}}
\def \mV {\bm{V}}
\def \mW {\bm{W}}
\def \mZ {\bm{Z}}
\def \mM {\bm{M}}
\def \mSigma {\bm{\Sigma}}
\def \mPhi {\bm{\Phi}}

\newtheorem{theorem}{Theorem}
\newtheorem{lemma}{Lemma}

\newtheorem{fact}{Fact}

\theoremstyle{definition}

\theoremstyle{remark}
\newtheorem{remark}{Remark}

\makeatletter
\newtheorem*{rep@theorem}{\rep@title}
\newcommand{\newreptheorem}[2]{%
	\newenvironment{rep#1}[1]{%
		\def\rep@title{#2 \ref{##1}}%
		\begin{rep@theorem}}%
		{\end{rep@theorem}}}
\makeatother

\newreptheorem{theorem}{Theorem}
\newreptheorem{lemma}{Lemma}
\newreptheorem{corollary}{Corollary}
\newreptheorem{proposition}{Proposition}
\begin{document}
	
\title{Phase Transitions in Recovery of Structured Signals from Corrupted Measurements}

\author{
	
	\IEEEauthorblockN{Zhongxing~Sun\IEEEauthorrefmark{1}, Wei~Cui\IEEEauthorrefmark{1}, and Yulong~Liu\IEEEauthorrefmark{2}}
	
	\IEEEauthorblockA{\IEEEauthorrefmark{1}School of Information and Electronics, Beijing Institute of Technology, Beijing 100081, China}
	
	\IEEEauthorblockA{\IEEEauthorrefmark{2}School of Physics, Beijing Institute of Technology, Beijing 100081, China}
	
     \thanks{Corresponding author: Yulong Liu. Email: yulongliu@bit.edu.cn.}
}
\maketitle

\begin{abstract}
   This paper is concerned with the problem of recovering a structured signal from a relatively small number of corrupted random measurements. Sharp phase transitions have been numerically observed in practice when different convex programming procedures are used to solve this problem. This paper is devoted to presenting theoretical explanations for these phenomenons by employing some basic tools from Gaussian process theory. Specifically, we identify the precise locations of the phase transitions for both constrained and penalized recovery procedures. Our theoretical results show that these phase transitions are determined by some geometric measures of structure, e.g., the spherical Gaussian width of a tangent cone and  the Gaussian (squared) distance to a scaled subdifferential. By utilizing the established phase transition theory, we further investigate the relationship between these two kinds of recovery procedures, which also reveals an optimal strategy (in the sense of Lagrange theory) for choosing the tradeoff parameter in the penalized recovery procedure. Numerical experiments are provided to verify our theoretical results.
\end{abstract}


\section{Introduction}
This paper studies the problem of recovering a structured signal from a relatively small number of corrupted measurements
\begin{align}\label{model: observe}
\vy = \bm{\Phi}\vx^{\star} +\vv^{\star} + \vz,
\end{align}
where $\bm{\Phi}\in\R^{m\times n}$ is the sensing matrix, $\vx^{\star}\in\R^n$ denotes the structured signal to be estimated, $\vv^{\star}\in\R^m$ stands for the structured corruption, and $\vz \in \R^m$ represents the unstructured observation noise. The objective is to estimate $\vx^{\star}$ and $\vv^{\star}$ from given knowledge of $\vy$ and $\bm{\Phi}$. If $\vv^{\star}$ contains some useful information, then this model \eqref{model: observe} can be regarded as the signal separation (or demixing) problem. In particular, if there is no corruption $(\vv^{\star} = \vzero)$, then the model \eqref{model: observe} reduces to the standard compressed sensing problem.

This problem arises in many practical applications of interest, such as face recognition \cite{wright2009robust}, subspace clustering \cite{elhamifar2009sparse}, sensor network \cite{haupt2008compressed}, latent variable modeling \cite{chandrasekaran2011rank}, principle component analysis \cite{candes2011robust}, source separation \cite{elad2005simultaneous}, and so on. The theoretical aspects of this problem have also been studied under different scenarios in the literature, important examples include sparse signal recovery from sparse corruption \cite{laska2009exact,wright2010dense,li2013compressed,nguyen2013exact,nguyen2013robust,kuppinger2012uncertainty,studer2012recovery,pope2013probabilistic,studer2014stable,su2016data,adcock2018compressed,adcock2019correcting}, low-rank matrix recovery from sparse corruption \cite{chandrasekaran2011rank, candes2011robust, xu2012robust, xu2013outlier,wright2013compressive,chen2013low}, and structured signal recovery from structured corruption \cite{mccoy2014sharp,foygel2014corrupted,Zhang2017On,Chen2017Corrupted,Jinchi2018Stable,Sun2019Recovery,Sun2020Quantized}.

Since this problem is ill-posed in general, tractable recovery is possible when both signal and corruption are suitably structured. Typical examples of structured signal (or corruption) include sparse
vectors and low-rank matrices. Let $f(\cdot)$ and $g(\cdot)$ be suitable proper convex functions which promote structures for signal and corruption respectively. There are three popular convex optimization approaches to reconstruct signal and corruption when different kinds of prior information are available. Specifically, when we have access to the prior knowledge of either signal $f(\vx^{\star})$ or corruption $g(\vv^{\star})$ and the noise level $\delta$ (in terms of the $\ell_2$ norm), it is natural to consider the following constrained convex recovery procedures
\begin{align}\label{Constrained_Optimization_I}
\min_{\vx, \vv} ~f(\vx),\quad\text{s.t.~}&g(\vv) \leq g(\vv^{\star}), ~~\|\vy-\bm{\Phi}\vx-\vv\|_2\leq \delta
\end{align}
and
\begin{align}\label{Constrained_Optimization_II}
\min_{\vx, \vv} ~g(\vv),\quad\text{s.t.~}&f(\vx) \leq f(\vx^{\star}), ~~\|\vy-\bm{\Phi}\vx-\vv\|_2\leq \delta.
\end{align}
When only the noise level $\delta$ is known, it is convenient to employ the partially penalized convex recovery procedure
\begin{align}\label{Partially_Penalized_Optimization}
\min_{\vx,\vv}~f(\vx)+\lambda\cdot g(\vv),\quad\text{s.t.}\quad\|\vy-\bm{\Phi}\vx-\vv\|_2\leq \delta,
\end{align}
where $\lambda > 0$ is a tradeoff parameter. When there is no prior knowledge available, it is practical to use the fully penalized convex recovery procedure
\begin{align}\label{Fully Penalized Optimization}
\min_{\vx,\vv}\frac{1}{2}\|\vy-\bm{\Phi}\vx-\vv\|_2^2+\tau_1\cdot f(\vx)+\tau_2\cdot g(\vv),
\end{align}
where $\tau_1,\tau_2 > 0$ are some tradeoff parameters.

A large number of numerical results in the literature have suggested that phase transitions emerge in all above three recovery procedures (under random measurements), see e.g., \cite{nguyen2013exact,nguyen2013robust,studer2012recovery,studer2014stable,su2016data,adcock2018compressed,mccoy2014sharp,foygel2014corrupted,Zhang2017On,Chen2017Corrupted,Jinchi2018Stable}. Concretely, for a specific recovery procedure, when the number of the measurements exceeds a threshold, this procedure can faithfully reconstruct both signal and corruption with high probability, when the number of the measurements is below the threshold, this procedure fails with high probability. A fundamental question then is:
\begin{center}
\emph{Q1: How to determine the locations of these phase transitions accurately?}
\end{center}
In addition, in partially and fully penalized recovery procedures, the optimization problems also rely on some tradeoff parameters. Another important question is:
\begin{center}
  \emph{Q2: How to choose these tradeoff parameters to achieve the best possible performance?}
\end{center}

\subsection{Model Assumptions and Contributions}
This paper tries to provide answers for the above two questions in the absence of unstructured noise ($\vz=0$). In this scenario, the observation model becomes
\begin{align}\label{corrupted sensing}
\vy=\mPhi\vx^{\star}+\sqrt{m}\vv^{\star}.
\end{align}
Here we assume that $\mPhi$ is a Gaussian sensing matrix with i.i.d. entries ($\mPhi_{ij}\sim \NN(0,1)$), and the factor $\sqrt{m}$ in \eqref{corrupted sensing} makes the columns of $\mPhi$ and $\sqrt{m}\mI_m$ have the same scale, which helps our theoretical results to be more interpretable. Accordingly, the constrained convex recovery procedures become
\begin{align}\label{constrained_procedure1}
\min_{\vx, \vv} ~f(\vx),\quad\text{s.t.~}&\vy=\mPhi\vx+\sqrt{m}\vv,~~g(\vv)\leq g(\vv^{\star})
\end{align}
and
\begin{align}\label{constrained_procedure2}
\min_{\vx, \vv} ~g(\vv),\quad\text{s.t.~}&\vy=\mPhi\vx+\sqrt{m}\vv,~~f(\vx)\leq f(\vx^{\star}),
\end{align}
and partially and fully penalized recovery procedures reduce to
\begin{align}\label{penilized_procedure}
\min_{\vx, \vv} ~f(\vx)+\lambda\cdot g(\vv),\quad\text{s.t.~}& \vy=\mPhi\vx+\sqrt{m}\vv,
\end{align}
where $\lambda > 0$ is a tradeoff parameter. For each recovery procedure, we declare it \emph{succeeds} when its unique solution $(\hat{\vx},\hat{\vv})$ satisfies $\hat{\vx}=\vx^{\star}$ and $\hat{\vv}=\vv^{\star}$, otherwise, it \emph{fails}.

Under the above model settings, the contribution of this paper is twofold:
\begin{itemize}
\item First, we develop a new analytical framework which allows us to establish the phase transition theory of both constrained and penalized recovery procedures in a unified way. Specifically, for constrained recovery procedures \eqref{constrained_procedure1} and \eqref{constrained_procedure2}, our analysis shows that their phase transitions locate at
    $$\mathscr{C}_p = \omega^2(\TT_f(\vx^{\star})\cap\S^{n-1}) + \omega^2(\TT_g(\vv^{\star})\cap\S^{m-1}), $$
    where $\TT_f(\vx^{\star})$ (or $\TT_g(\vv^{\star})$) is the tangent cone induced by $f$ (or $g$) at the true signal $\vx^{\star}$ (or corruption $\vv^{\star}$), $\omega(\TT_f(\vx^{\star})\cap\S^{n-1})$ (or $\omega(\TT_g(\vv^{\star})\cap\S^{m-1})$) is the spherical Gaussian width of this cone, defined in Section \ref{Preliminaries}. For the penalized recovery procedure \eqref{penilized_procedure}, our results indicate that its critical point locates at
    $$ \mathscr{C}_p (\lambda) = \min_{\alpha\leq t\leq \beta}   2 \cdot \zeta\left(\frac{\sqrt{m}}{\lambda t}\partial f(\vx^\star)\right)+ \eta^2\left(\frac{1}{t}\partial g(\vv^\star)\cap \S^{m-1}\right) -1,$$
    where $\partial f(\vx^{\star})$ (or $\partial g(\vv^{\star})$) is the subdifferential of $f$ (or $g$) at the true signal $\vx^{\star}$ (or corruption $\vv^{\star}$),  $\zeta(\SS)$ and $\eta^2(\SS)$ denote the Gaussian distance and the Gaussian squared distance to a set $\SS$ respectively, also defined in Section \ref{Preliminaries},  and $\alpha=\min_{\vb\in\partial g(\vv^\star)}\|\vb\|_2$ and $\beta=\max_{\vb\in\partial g(\vv^\star)}\|\vb\|_2$.

\item Second, we investigate the relationship between these two kinds of recovery procedures by utilizing the established critical points $\mathscr{C}_p$ and $\mathscr{C}_p (\lambda)$, which also reveals an optimal parameter selection strategy for $\lambda$ (in the sense of Lagrange theory):
    $$ \lambda^\star=\arg\min_{\lambda>0}\mathscr{C}_p(\lambda).$$
    More precisely, under mild conditions, if the penalized procedure \eqref{penilized_procedure} is likely to succeed, then the constrained procedures \eqref{constrained_procedure1} and \eqref{constrained_procedure2} succeed with high probability, namely, if $m\geq \mathscr{C}_p (\lambda)$, then we have $m \geq \mathscr{C}_p-1$.  On the contrary, if the constrained procedures \eqref{constrained_procedure1} and \eqref{constrained_procedure2} are likely to succeed, then we can choose the tradeoff parameter $\lambda$ as $\lambda^{\star}$ such that the penalized procedure \eqref{penilized_procedure} succeeds with high probability, namely, if $m\geq \mathscr{C}_p$, then we have $m \geq \mathscr{C}_p(\lambda^{\star})-1$.
\end{itemize}

\subsection{Related Works}
	During the past few decades, there have been abundant works investigating the phase transition phenomenons in random convex optimization problems. Most of these works fit in with the framework of compressed sensing ($\vv^{\star} = \vzero$). In this paper, we focus on the scenario in which the random measurements are contaminated by some structured corruption. We will review the works related to these two aspects in details.

\subsubsection{Related Works in Compressed Sensing}
The works in the context of compressed sensing can be roughly divided into four groups according to their analytical tools.

The early works study phase transitions in the context of sparse signal recovery via \emph{polytope angle calculations} \cite{donoho2005neighborly}. Under Gaussian measurements, Donoho \cite{donoho2006high} analyzes the $\ell_1$-minimization method in the asymptotic regime and establishes an empirically tight lower bound on the number of measurements required for successful recovery. In contrast to \cite{donoho2006high}, Donoho and Tanner \cite{donoho2009counting} also prove the existence of sharp phase transitions in the asymptotic regime when using $\ell_1$-minimization to reconstruct sparse signals from random projections. These results are later extended to other related $\ell_1$-minimization problems. For instance, Donoho and Tanner \cite{donoho2005neighborliness,donoho2010counting2} identify a precise phase transition of the sparse signal recovery problem with an additional nonnegative constraint; Khajehnejad \emph{et al.} \cite{khajehnejad2011analyzing} introduce a nonuniform sparse model and analyze the performance of weighted $\ell_1$-minimization over that model; Xu and Hassibi \cite{xu2011precise} present sharp performance bounds on the number of measurements required for recovering approximately sparse signals from noisy measurements via $\ell_1$-minimization.


 In \cite[Fact 10.1]{amelunxen2014living}, Amelunxen \emph{et al.} explore the relationship between polytope angle theory and \emph{conical integral geometry} in details. Their results have shown that conical integral geometry could go beyond many inherent limitations of polytope angle theory, such as dealing with the nuclear norm regularizer in low-rank matrix recovery problems, non-asymptotic analysis, and establishing phase transition from absolute success to absolute failure. In summary, \cite{amelunxen2014living} provides the first comprehensive analysis that explains phase transition phenomenons in some random convex optimization problems. Other authors further use conical integral geometry to analyze convex optimization problems with random data. For examples, Amelunxen and B{\"u}rgisser \cite{amelunxen2015intrinsic,amelunxen2015probabilistic} apply conical integral geometry to study conic optimization problems; Goldstein \emph{et al.} \cite{goldstein2017gaussian} show that the sequence of conic intrinsic volumes can be approximated by a suitable Gaussian distribution in the high-dimensional limit, which provides more precise probabilities for successful and failed recovery.

	Whereas the above works involve combinatorial geometry, there are some others using  \emph{minimax decision theory} to analyze the phase transition problems. Several papers \cite{donoho2009message,bayati2011dynamics,bayati2011lasso} have observed a close agreement between the asymptotic mean square error (MSE) and the location of phase transition in the linear inverse problems. Donoho \emph{et al.} \cite{donoho2011noise,donoho2013phase} then have shown that the minimax MSE for denoising empirically predicts the locations of phase transitions in both sparse and low-rank recovery problems. Recently, Oymak and Hassibi \cite{oymak2016sharp} prove that the minimax MSE risk in structured signal denoising problems is almost the same as the statistical dimension. Combining with the results in \cite{amelunxen2014living}, their results provide a theoretical explanation for using minimax risk to describe the location of phase transition in regularized linear inverse problems.

	The last line of works study the compressed sensing problem by utilizing some tools from \emph{Gaussian process theory}. The key technique is a sharp comparison inequality for Gaussian processes due to Gordon \cite{gordon1985some}. Rudelson and Vershynin \cite{rudelson2008sparse} first use Gordon's inequality to study the $\ell_1$-minimization problem. Stojnin \cite{stojnic2009various} refines this method and obtains an empirically sharp success recovery condition under Gaussian measurements. Stojnin's calculation is then extended to more general settings. Oymak and Hassibi \cite{oymak2010new} use it to study the nuclear norm minimization problem. Chandrasekaran \emph{et al.} \cite{chandrasekaran2012convex} consider a more general case in which the regularizer can be any convex function. Stojnic \cite{stojnic2013framework,stojnic2013regularly} has also investigated the error behaviors of $\ell_1$-minimization and its variants in random optimization problems, these works have been extended by a series of researches \cite{oymak13The,thrampoulidis2014gaussian} by Oymak, Thrampoulidis and Hassibi. Although the mentioned works provide detailed discussions for using Gaussian process theory to analyze random convex optimization problems, few of them consider the failure case for recovery. In a recent work \cite{oymak2018universality}, Oymak and Tropp demonstrate a universality property for randomized dimension reduction, which also proves the phase transition in the recovery of structured signals from a large class of measurement models.

\subsubsection{Related Works in Corrupted Sensing}
	There are several works in the literature studying the phase transition theory of corrupted sensing problems. For instance, in \cite{amelunxen2014living}, Amelunxen \emph{et al.} study the demixing problem $\vz_0=\vx_0+\mU\vy_0$, where $\mU\in\R^{n\times n}$ is a random orthogonal matrix. They establish sharp phase transition results when the constrained convex program $\min_{\vx, \vy} ~f(\vx),~\text{s.t.~}\vz_0=\vx+\mU\vy,~g(\vy)\leq g(\vy_0)$ is used to solve this demixing problem.  In \cite{oymak2018universality}, Oymak and Tropp consider a more general demixing model $\vy=\mPhi_0\vx_0+\mPhi_1\vx_1$, where $\mPhi_0,\mPhi_1\in\R^{m\times n}$ are two random transformation matrices drawing from a wide class of distributions. They attempt to reconstruct the original signal pair by solving $\min_{\vz_0,~\vz_1}\max\{f_0(\vz_0), f_1(\vz_1)\},~\text{s.t.~}\vy=\mPhi_0\vz_0+\mPhi_1\vz_1$ and establish the related phase transition theory. More related to this work, Foygel and Mackey \cite{foygel2014corrupted} consider the corrupted sensing problem \eqref{model: observe} and analyze both the constrained recovery procedures \eqref{Constrained_Optimization_I} or \eqref{Constrained_Optimization_II}  and the partially penalized recovery procedure \eqref{Partially_Penalized_Optimization}. In each case, they provide sufficient conditions for stable signal recovery from structured corruption with added unstructured noise under Gaussian measurements. Very recently, Chen and Liu \cite{Jinchi2018Stable} develop an extended matrix deviation inequality and use it to analyze all three kinds of convex procedures (\eqref{Constrained_Optimization_I}, \eqref{Constrained_Optimization_II}, \eqref{Partially_Penalized_Optimization}, and \eqref{Fully Penalized Optimization}) in a unified way under sub-Gaussian measurements. In terms of failure case of corrupted sensing, Zhang, Liu, and Lei \cite{Zhang2017On} establish a sharp threshold below which the constrained convex procedures \eqref{constrained_procedure1} and  \eqref{constrained_procedure2} fail to recover both signal and corruption under Gaussian measurements. Together with the work in \cite{foygel2014corrupted}, their results provide a theoretical explanation for the phase transition when the constrained procedures are used to solve corrupted sensing problems.

\subsection{Organization}

The remainder of the paper is organized as follows. We start with reviewing some preliminaries that are necessary for our subsequent analysis in Section \ref{Preliminaries}. Section \ref{Main results} is devoted to presenting the main theoretical results of this paper. In Section \ref{simulations}, we present a series of numerical experiments to verify our theoretical results. We conclude the paper in Section \ref{Conclusion}. All proofs of our main results are included in Appendixes.

\section{Preliminaries}\label{Preliminaries}
In this section, we introduce some notations and facts that underlie our analysis. Throughout the paper, $\S^{n-1}$ and $\B_2^n$  represent the unit sphere and unit ball in $\R^n$ under the $\ell_2$ norm, respectively. 

\subsection{Convex Geometry}
\subsubsection{Subdifferential}
The \emph{subdifferential} of a convex function $f$ at $\vx$ is the set of vectors
\begin{equation*}\label{Definitionofsubdiff}
 \partial f(\vx) = \{\vu\in\R^n: f(\vx+\vd)-f(\vx)\ge\ip{\vu}{\vd}~\text{for all~} \vd\in\R^n\}.
\end{equation*}
If $f$ is convex and $\vx \in \textrm{intdom}(f)$, then $\partial f(\vx)$  is a nonempty, compact, convex set. For any number $t \geq 0$, we denote the scaled subdifferential as $t \cdot \partial f(\vx) = \{t \cdot \vu: \vu \in \partial f(\vx) \}$.

\subsubsection{Cone and Polar Cone}
A subset $\mathcal{C} \subset\R^n $ is called a \emph{cone} if for every $\vx \in \mathcal{C}$ and $t \geq 0$, we have $t\cdot \vx \in \mathcal{C}$. For a cone $\mathcal{C}\subset\R^n$, the \emph{polar cone} of $\mathcal{C}$ is defined as
\begin{equation*}\label{Definition_polar}
\mathcal{C}^\circ  =  \{\vu\in\R^n:\ip{\vu}{\vx}\leq 0~\textrm{for all}~\vx\in\mathcal{C}\}.
\end{equation*}
The polar cone $\mathcal{C}^\circ$ is always closed and convex. A subset $\SS\subset\S^{n-1}$ is called \emph{spherically convex} if $\SS$ is the intersection of a convex cone with the unit sphere.

\subsubsection{Tangent Cone and Normal Cone}
The \emph{tangent cone} of a convex function $f$ at $\vx$ is defined as the set of descent directions of $f$ at $\vx$
\begin{equation*}\label{DefinitionofTangentCone}
\TT_f(\vx) = \{\vu\in\R^n: ~f(\vx+t\cdot\vu)\leq f(\vx)~\textrm{for some}~t > 0\}.
\end{equation*}
The tangent cone of a proper convex function is always convex, but they may not be closed.

The \emph{normal cone} of a convex function $f$ at $\vx$ is the polar of the tangent cone
	\begin{equation*}\label{DefinitionofNormalCone}
	\NN_f(\vx) =\TT_f^\circ(\vx)= \{\vu\in\R^n: \ip{\vu}{\vd}\leq 0~\text{for all~} \vd\in\TT_f(\vx)\}.
	\end{equation*}
	Suppose that $0\notin\partial f(\vx)$, the normal cone can also be written as the cone hull of the subdifferential \cite[Theorem 23.7]{rockafellar1970convex}
	\begin{equation*}
	\NN_f(\vx) = \cone\{\partial f(\vx)\}=\{\vu\in\R^n: ~\vu\in t\cdot\partial f(\vx)~\textrm{for some}~t \geq 0\}.
	\end{equation*}

\subsection{Geometric Measures}

\subsubsection{Gaussian Width}
For any $\SS \subset \R^{n}$, a popular way to quantify the ``size'' of $\SS$ is through its \emph{Gaussian width}
\begin{equation*}\label{Definition_Gaussian_width}
\omega(\SS) := \E \sup_{\vx \in \SS} \langle \vg, \vx \rangle, ~~ \textrm{where} ~~\vg\sim\NN(0,\mI_n).
\end{equation*}

\subsubsection{Gaussian Distance and Gaussian Squared Distance}
Recall that the Euclidean distance to a set $\SS\subset\R^n$ is defined as
\begin{equation*}\label{Definition_dist}
\dist(\vx,\SS) := \inf_{\vy \in \SS} \| \vx-\vy\|_2.
\end{equation*}
We define the \emph{Gaussian distance} to a set $\SS \subset \R^{n}$ as
\begin{equation*}\label{Definition_Gaussian_distance}
\zeta(\SS) := \E \dist(\vg,\SS) = \E \inf_{\vy \in \SS} \| \vg-\vy\|_2, ~~ \textrm{where} ~~\vg\sim\NN(0,\mI_n).
\end{equation*}
Similarly, the \emph{Gaussian squared distance} to a set $\SS \subset \R^{n}$ is defined as
\begin{equation*}\label{Definition_Gaussian_distance_Squared}
\eta^2(\SS) := \E \dist^2(\vg,\SS) = \E \inf_{\vy \in \SS} \| \vg-\vy\|_2^2, ~~ \textrm{where} ~~\vg\sim\NN(0,\mI_n).
\end{equation*}
These two quantities are closely related \footnote{The lower and upper bounds follow from Fact \ref{Variance of Lipschitz} (in Appendix \ref{auxiliaryresults}) and Jensen's inequality, respectively.}
\begin{equation}\label{Relationship}
  \sqrt{\eta^2(\SS)-1} \leq \zeta(\SS) \leq \sqrt{\eta^2(\SS)}.
\end{equation}

\subsection{Tools from Gaussian Analysis}
Our analysis makes heavy use of two well-known results in Gaussian analysis. The first one is a comparison principle for Gaussian processes due to Gordon \cite[Theorem 1.1]{gordon1985some}. This result provides a convenient way to bound the probability of an event from below by that of another one. It is worth noting that the original lemma can be naturally extended from discrete index sets to compact index sets, see e.g., \cite[Lemma C.1]{oymak13The}.
\begin{fact}[Gordon's Lemma] \cite[Theorem 1.1]{gordon1985some}
	\label{Grodon lem}
	Let $\{X_{ij}\}$, $\{Y_{ij}\}$, $1\leq i\leq n,~1\leq j\leq m$, be two centered Gaussian processes. If $X_{ij}$ and $Y_{ij}$ satisfy the following inequalities:
	\begin{align*}
	&\textrm{1.}~\E[X_{ij}^2]=\E[Y_{ij}^2];\\
	&\textrm{2.}~\E[X_{ij}X_{ik}]\leq\E[Y_{ij}Y_{ik}];\\
	&\textrm{3.}~\E[X_{ij}X_{lk}]\geq\E[Y_{ij}Y_{lk}]~\textrm{for all}~i\neq l;
	\end{align*}
	Then, we have
	\begin{align*}
	\Pr{\cap_i\cup_j [X_{ij}\geq\tau_{ij}]}\geq\Pr{\cap_i\cup_j [Y_{ij}\geq\tau_{ij}]}
	\end{align*}
	for all choices of $\tau_{ij}\in\R$.
\end{fact}

The second one is the Gaussian concentration inequality which allows us to establish tail bounds for different kinds of Gaussian Lipschitz functions. Recall that a function $f:\R^n \rightarrow \R$ is $L$-Lipschitz with
respect to the Euclidean norm $\|\cdot\|_2$ if
$$|f(\vx) - f(\vy)| \leq L \|\vx - \vy\|_2   ~~~~\textrm{for all}~~ \vx, \vy \in \R^n.  $$
Then the Gaussian concentration inequality reads as
\begin{fact}[Gaussian concentration inequality]\cite[Theorem 1.7.6]{bogachev1998gaussian}
	\label{gaussian concentration}
	Let $\vx\sim \NN(0,\mI_n)$, and let $f:~\R^n\to\R$  be $L$-Lipschitz with respect to the Euclidean metric. Then for any $\epsilon \geq 0$, we have
	\begin{equation*}
	\Pr{ f(\vx)  \geq \E f(\vx)+\epsilon} \leq \exp\left(\frac{-\epsilon^2}{2L^2}\right),
	\end{equation*}
	and
	\begin{equation*}
	\Pr{ f(\vx) \leq \E f(\vx)-\epsilon} \leq \exp\left(\frac{-\epsilon^2}{2L^2}\right).
	\end{equation*}
\end{fact}

\section{Main Results}
\label{Main results}
In this section, we will present our main results. Section \ref{ph for constrained} is devoted to analyzing the phase transition of constrained recovery procedures \eqref{constrained_procedure1} and \eqref{constrained_procedure2}. The phase transition of penalized recovery procedure \eqref{penilized_procedure} will be established in Section \ref{ph for penilized}. Section \ref{relation between con and pen} explores the relationship between these two kinds of recovery procedures and illustrates how to choose the optimal tradeoff parameter $\lambda$. The proofs are included in Appendixes.

\subsection{Phase Transition of the Constrained Recovery Procedures}\label{ph for constrained}
We start with analyzing the phase transition of constrained recovery procedures \eqref{constrained_procedure1} and \eqref{constrained_procedure2}. Recall that a recovery procedure \emph{succeeds} if it has a unique optimal solution which coincides with the true value; otherwise it \emph{fails}. First of all, it is necessary to specify some analytic conditions under which the constrained procedures \eqref{constrained_procedure1} and \eqref{constrained_procedure2} succeed or fail to recover the original signal and corruption. To this end, we have following lemma.

\begin{lemma}[Sufficient conditions for successful and failed recovery]
	\label{lem: success and failure_con}
	Suppose $\TT_f(\vx^{\star})$ and $\TT_g(\vv^{\star})$ are nonempty and closed. If
	\begin{align}\label{s_condition_con}
	\min_{(\va,\vb)\in \left(\TT_f(\vx^{\star})\times\TT_g(\vv^{\star})\right)\cap\S^{n+m-1}} \|\mPhi\va+\sqrt{m}\vb\|_2>0,
	\end{align}
    then the constrained procedures \eqref{constrained_procedure1} and \eqref{constrained_procedure2} succeed. If	
    \begin{align}\label{f_condition_con}
	\min_{(\va,\vb)\in \left(\TT_f(\vx^{\star})\times\TT_g(\vv^{\star})\right)\cap\S^{n+m-1}} \|\mPhi\va+\sqrt{m}\vb\|_2=0.
	\end{align}
    then the constrained procedures \eqref{constrained_procedure1} and \eqref{constrained_procedure2} fail. Furthermore, a sufficient condition for \eqref{f_condition_con} to hold is
	\begin{align}\label{polar_f_con}
	\min_{\vr\in\S^{m-1}}\min_{\vs\in(\TT_f(\vx^{\star})\times\TT_g(\vv^{\star}))^\circ} \|\vs-\mA^T\vr\|_2>0,
	\end{align}
	where $(\TT_f(\vx^{\star})\times\TT_g(\vv^{\star}))^\circ$ denotes the polar cone of $\TT_f(\vx^{\star})\times\TT_g(\vv^{\star})$, $\mA=[\mPhi, \sqrt{m}\mI_m]$, and $\mI_m$ is the $m$-dimensional identity matrix.
\end{lemma}

Armed with this lemma, our first theorem shows that the phase transition of constrained recovery procedures \eqref{constrained_procedure1} and \eqref{constrained_procedure2} occurs around the sum of squares of spherical Gaussian widths of $\TT_f(\vx^{\star})$ and $\TT_g(\vv^{\star})$. This result ensures that the recovery is likely to succeed when the number of measurements exceeds the critical point. On the contrary, the recovery is likely to fail when the number of measurements is smaller than the critical point.

\begin{theorem}[Phase transition of constrained recovery procedures]
	\label{them: Phase transition_con}
	Consider the corrupted sensing model \eqref{corrupted sensing} with Gaussian measurements. Assume $\TT_f(\vx^{\star})$ and $\TT_g(\vv^{\star})$ are non-empty and closed. Define $\mathscr{C}_p:= {\omega^2\left(\TT_f(\vx^{\star})\cap\S^{n-1}\right)+\omega^2\left(\TT_g(\vv^{\star})\cap\S^{m-1}\right)}$. If the number of measurements satisfies
	\begin{align}\label{NumberofMeasurements_c1}
	\sqrt{m} \geq \sqrt{\mathscr{C}_p} + \sqrt{2} +\epsilon,
	\end{align}
	then the constrained procedures \eqref{constrained_procedure1} and \eqref{constrained_procedure2} succeed with probability at least $1-2\exp\left(\frac{-\epsilon^2}{4}\right)$. If the number of measurements satisfies
	\begin{align}\label{NumberofMeasurements_c2}
	\sqrt{m} \leq \sqrt{\mathscr{C}_p} -\epsilon,
	\end{align}
	then the constrained procedures \eqref{constrained_procedure1} and \eqref{constrained_procedure2} fail with probability at least $1-2\exp\left(\frac{-\epsilon^2}{4}\right)$.
\end{theorem}

\begin{remark}[Relation to existing results]
In \cite[Theorem 1]{foygel2014corrupted}, Foygel and Mackey have shown that when
$$\sqrt{m} \geq \sqrt{\gamma^2\left(\TT_f(\vx^{\star})\cap\B_2^n\right)+\gamma^2\left(\TT_g(\vv^{\star})\cap\B_2^m\right)} + \frac{1}{\sqrt{2}} + \frac{1}{\sqrt{2\pi}} +  \epsilon,$$
the constrained procedures \eqref{constrained_procedure1} and \eqref{constrained_procedure2} succeed with probability at least $1-\exp\left({-\epsilon^2}/{2}\right)$. Here, the \emph{Gaussian squared complexity} of a set $\SS \subset \R^n$ is defined as $\gamma^2(\SS):= \E \left[  \left(\sup_{\vx \in \SS} \langle \vg, \vx \rangle \right)_+^2  \right]$ with $\vg\sim\NN(0,\mI_n)$ and $(a)_+ = \max\{a, 0\}$. On the other hand, the third author and his coauthors \cite[Theorem 1]{Zhang2017On} have demonstrated that when
$$\sqrt{m} \leq \sqrt{\omega^2\left(\TT_f(\vx^{\star})\cap\S^{n-1}\right)+\omega^2\left(\TT_g(\vv^{\star})\cap\S^{m-1}\right)} -\epsilon, $$
the constrained procedures \eqref{constrained_procedure1} and \eqref{constrained_procedure2} fail with probability at least $1-\exp\left({-\epsilon^2}/{2}\right)$. Since $\gamma^2\left(\TT_f(\vx^{\star})\cap\B_2^n\right)$ (or $\gamma^2\left(\TT_g(\vv^{\star})\cap\B_2^m\right)$) is very close to $\omega^2\left(\TT_f(\vx^{\star})\cap\S^{n-1}\right)$ (or $\omega^2\left(\TT_g(\vv^{\star})\cap\S^{m-1}\right)$), the above two results have essentially established the phase transition theory of the constrained procedures \eqref{constrained_procedure1} and \eqref{constrained_procedure2}.

However, in this paper, we have developed a new analytical framework which allows us to unify the results in both success and failure cases in terms of Gaussian width, which makes the phase transition theory of the constrained recovery procedures more natural. More importantly, this framework can be easily applied to establish the phase transition theory of the penalized recovery procedure.
\end{remark}

\begin{remark}[Related works]
  In \cite{amelunxen2014living}, Amelunxen \emph{et al.} consider the following demixing problem
  $$ \vy = \mU \vx^{\star} + \vv^{\star},    $$
  where $\vx^{\star}, \vv^{\star} \in \R^n$ are unknown structured signals and $\mU \in \R^{n\times n}$ is a known orthogonal matrix. They have shown that the phase transition occurs around
  $\delta(\TT_f(\vx^{\star})) + \delta(\TT_g(\vv^{\star}))$ when the constrained recovery procedures are employed to solve this problem. Here, the \emph{statistical dimension} of a convex cone $\CC \in \R^n$ is defined as $ \delta(\CC) := \E \left[  \left(\sup_{\vx \in \CC\cap\B_2^n} \langle \vg, \vx \rangle \right)^2  \right]$ with $\vg\sim\NN(0,\mI_n)$. Although the model assumptions of this demixing problem are different from ours, the results in the two cases are essentially consistent, since we have $\delta(\TT_f(\vx^{\star})) + \delta(\TT_g(\vv^{\star})) \approx \mathscr{C}_p$ (by Fact \ref{relation_w_d} in Appendix \ref{auxiliaryresults}).

  Recently, Oymak and Tropp \cite{oymak2018universality} consider a more general demixing model
  $$ \vy = \mPhi_0 \vx^{\star} + \mPhi_1 \vv^{\star},$$
  where $\vx^{\star}, \vv^{\star} \in \R^n$ are unknown structured signals and $\mPhi_0, \mPhi_1 \in \R^{m\times n}$ are random matrices. They have demonstrated that the critical point of the constrained recovery procedures is nearly located at $\delta(\TT_f(\vx^{\star})) + \delta(\TT_g(\vv^{\star}))$ for a large class of random matrices drawing from some models. Their model assumptions are also different from ours, because $\mPhi_1$ is a deterministic matrix in our case, which makes our analysis different from that of \cite{oymak2018universality}.

\end{remark}

\subsubsection{How to evaluate the critical point $\mathscr{C}_p$?}
Theorem \ref{them: Phase transition_con} has demonstrated that the phase transition of the constrained recovery procedures occurs around
$$
\mathscr{C}_p = \omega^2\left(\TT_f(\vx^{\star})\cap\S^{n-1}\right)+\omega^2\left(\TT_g(\vv^{\star})\cap\S^{m-1}\right).
$$
A natural question then is how to determine the value of this critical point. To this end, it suffices to estimate $\omega^2\left(\TT_f(\vx^{\star})\cap\S^{n-1}\right)$ and $\omega^2\left(\TT_g(\vv^{\star})\cap\S^{m-1}\right)$. It is now well-known that there are some standard recipes to estimate these two quantities, see e.g., \cite{chandrasekaran2012convex, amelunxen2014living,foygel2014corrupted}. Actually, Facts \ref{relation_w_d} and \ref{bound of delta} indicate that $\mathscr{C}_p$ can be accurately approximated by
\begin{equation}\label{cal_ph of con}
\min_{t\geq 0}\eta^2(t\cdot\partial f(\vx^{\star}))+\min_{t\geq 0}\eta^2(t\cdot\partial g(\vv^{\star})).
\end{equation}

To illustrate this result \eqref{cal_ph of con}, we consider two typical examples: sparse signal recovery from sparse corruption and low-rank matrix recovery from sparse corruption. In the first example, we assume $\vx^{\star} \in \R^{n}$ and $\vv^{\star}\in \R^{m}$ are $s$-sparse and $k$-sparse vectors, respectively. Direct calculations (see Appendix \ref{cal}) lead to
\begin{align*}
&\min_{t\geq 0}\eta^2(t\cdot\partial \|\vx^{\star}\|_1)=\min_{t\geq 0}\left\{s(1+t^2)+\frac{2(n-s)}{\sqrt{2\pi}}\left((1+t^2)\int_{t}^{\infty}e^{-x^2/2}dx-te^{-t^2/2}\right)\right\},
\end{align*}
and
\begin{align*}
&\min_{t\geq 0}\eta^2(t\cdot\partial \|\vv^{\star}\|_1)=\min_{t\geq 0}\left\{k(1+t^2)+\frac{2(m-k)}{\sqrt{2\pi}}\left((1+t^2)\int_{t}^{\infty}e^{-x^2/2}dx-te^{-t^2/2}\right)\right\}.
\end{align*}
In the case of low-rank matrix recovery from sparse corruption, suppose $\mX^{\star}\in \R^{n \times n}$ is an $r$-rank matrix, the Gaussian squared distance of the signal in \eqref{cal_ph of con} is given by
\begin{align*}
\min_{t\geq 0}\eta^2(t\cdot\partial\|\mX\|_*)&=\min_{t\geq 0}\left\{r(2n-r+t^2)+\E\sum_{i=1}^{n-r}\textrm{shrink}\left(\sigma_i(\mG_2),t\right)^2\right\},
\end{align*}
where $\mG_2$ is an $(n-r)\times(n-r)$ standard Gaussian matrix, and $\sigma_i(\mG_2)$ is the $i$-th largest singular value of $\mG_2$. The Gaussian squared distance of the corruption can be similarly evaluated as in the first example.
In addition, we should mention that it is also possible to estimate the Gaussian width in $\mathscr{C}_p$ numerically by approximating the expectation in its definition with an empirical average.

\subsection{Phase Transition of the Penalized Recovery Procedure}\label{ph for penilized}
We then study the phase transition theory of the penalized recovery procedure \eqref{penilized_procedure}. Firstly, we also need to establish sufficient conditions under which the penalized recovery procedure \eqref{penilized_procedure} succeeds or fails to recover the original signal and corruption.
\begin{lemma}[Sufficient conditions for successful and failed recovery]
	\label{lem: success and failure}
    The penalized problem \eqref{penilized_procedure} succeeds if
	\begin{align}\label{s_condition}
	\vzero\in \mPhi^T\cdot\partial g(\vv^\star)-\frac{\sqrt{m}}{\lambda}\partial f(\vx^\star).
	\end{align}
	The penalized problem \eqref{penilized_procedure} fails if
	\begin{align}
	\min_{\va\in\partial f(\vx^\star) , \vb\in\partial g(\vv^\star)} \|\mPhi^T\vb-\frac{\sqrt{m}}{\lambda}\va\|_2>0.
	\end{align}
	Define the joint cone $\TT_J=\{t\cdot(\va^T,\vb^T)^T\in\R^n\times\R^m:t\geq 0,~\va\in\partial f(\vx^\star),~\vb\in\partial g(\vv^\star)\}$. Then, a sufficient condition for \eqref{s_condition} to hold is
	\begin{align}\label{polar_s_con}
	\min_{\vr\in\S^{n-1}}\min_{\vs\in\TT_J^\circ} \|\vs-\mM^T\vr\|_2>0,
	\end{align}
	where $\TT_J^\circ$ denotes the polar cone of $\TT_J$, $\mM=[-\frac{\sqrt{m}}{\lambda}\mI_n, \mPhi^T]$, and $\mI_n$ is the $n$-dimensional identity matrix.
\end{lemma}

Our second theorem shows that the critical point of the penalized recovery procedure is nearly located at $\mathscr{C}_p(\lambda)$, which is determined by two Gaussian distances to scaled subdifferentials. This result asserts that the recovery succeeds with high probability when the number of measurements is larger than the critical point. On the other hand, the recovery fails with high probability when the number of measurements is below the critical point. In addition, the critical point $\mathscr{C}_p(\lambda)$ is influenced by the tradeoff parameter $\lambda$.
\begin{theorem}[Phase transition of penalized recovery procedure]
	\label{them: Phase transition}
	Consider the corrupted sensing model \eqref{corrupted sensing} with Gaussian measurements. Suppose that the subdifferential $\partial g(\vv^\star)$ does not contain the origin. Let $\alpha=\min_{\vb\in\partial g(\vv^\star)}\|\vb\|_2$ and $\beta=\max_{\vb\in\partial g(\vv^\star)}\|\vb\|_2$. Define
	\begin{align*}
	\mathscr{C}_p(\lambda):=\min_{\alpha\leq t\leq \beta}   2 \cdot \zeta\left(\frac{\sqrt{m}}{\lambda t}\partial f(\vx^\star)\right)+ \eta^2\left(\frac{1}{t}\partial g(\vv^\star)\cap \S^{m-1}\right) -1.
	\end{align*}
	If the number of measurements satisfies
	\begin{align}\label{NumberofMeasurements1}
	m \geq \mathscr{C}_p(\lambda)+\epsilon,
	\end{align}
	then the penalized problem \eqref{penilized_procedure} succeeds with probability at least $1-2\exp\left(\frac{-\epsilon^2}{16}\right)$. If the number of measurements satisfies
	\begin{align}\label{NumberofMeasurements2}
	m \leq \mathscr{C}_p(\lambda)-\epsilon,
	\end{align}
	then the penalized problem \eqref{penilized_procedure} fails with probability at least $1-2\exp\left(\frac{-\epsilon^2}{16}\right)$.
\end{theorem}

\begin{remark}[Related works]
  In \cite{foygel2014corrupted}, Foygel and Mackey have shown that, under Gaussian measurements, when
  $$\sqrt{m} \geq 2\cdot \sqrt{\eta^2(\lambda_1 \cdot \partial f(\vx^{\star}))} + \sqrt{\eta^2(\lambda_2 \cdot \partial g(\vv^{\star}))} + 3\sqrt{2\pi} + \frac{1}{\sqrt{2}} + \frac{1}{\sqrt{2\pi}}+\epsilon,$$
  the penalized problem \eqref{penilized_procedure} succeeds with probability at least $1-\exp(-\epsilon^2/2)$. Here $\lambda = \lambda_2/\lambda_1$. Very recently, Chen and Liu \cite{Jinchi2018Stable} have illustrated that, under sub-Gaussian measurements, when
  $$ m \geq CK^4 \left[ \eta^2(\lambda_1 \cdot \partial f(\vx^{\star}))+ \eta^2(\lambda_2 \cdot \partial g(\vv^{\star})) \right],$$
  the penalized problem \eqref{penilized_procedure} succeeds with high probability. Here $C$ is an absolute constant and $K$ is the upper bound for the sub-Gaussian norm of rows of the sensing matrix. These two sufficient conditions for successful recovery are demonstrated to be unsharp by their numerical experiments.  To the best of our knowledge, the present results (Theorem \ref{them: Phase transition}) first establish the complete phase transition theory of the penalized recovery procedure \eqref{penilized_procedure}, which closes an important open problem in the literature, see e.g., \cite{mccoy2014sharp}, \cite{foygel2014corrupted}, and \cite{Jinchi2018Stable}.

\end{remark}

\subsubsection{How to evaluate the critical point $\mathscr{C}_p(\lambda)$?}
Theorem \ref{them: Phase transition} has suggested that the phase transition of the penalized recovery procedure occurs around
$$\mathscr{C}_p(\lambda)=\min_{\alpha\leq t\leq \beta}   2 \cdot \zeta\left(\frac{\sqrt{m}}{\lambda t}\partial f(\vx^\star)\right)+ \eta^2\left(\frac{1}{t}\partial g(\vv^\star)\cap \S^{m-1}\right) -1.$$
The next important question is how to calculate $\mathscr{C}_p(\lambda)$ accurately. To this end, we have the following lemma.
\begin{lemma}\label{calculation_Delta_lambda}
	The quantity $\mathscr{C}_p(\lambda)$ can be bounded as
	\begin{align*}
	&\min_{\alpha\leq t\leq \beta} \left[2\cdot\sqrt{\eta^2\left(\frac{\sqrt{m}}{\lambda t}\partial f(\vx^\star)\right)-1}-2\cdot\omega\left(\frac{1}{t}\partial g(\vv^\star)\cap \S^{m-1}\right)+m\right]\leq\mathscr{C}_p(\lambda)\\
	&\hspace{150pt}\leq\min_{\alpha\leq t\leq \beta} \left[2\cdot\sqrt{\eta^2\left(\frac{\sqrt{m}}{\lambda t}\partial f(\vx^\star)\right)}-2\cdot\omega\left(\frac{1}{t}\partial g(\vv^\star)\cap \S^{m-1}\right)+m\right].
	\end{align*}
\end{lemma}
\begin{proof}
	Note that
	\begin{align}\label{bound of delta_lbd}
	\mathscr{C}_p(\lambda)&=\min_{\alpha\leq t\leq \beta} \E\left[2\cdot\dist\left(\vg,\frac{\sqrt{m}}{\lambda t}\partial f(\vx^\star)\right)+\dist^2\left(\vh,\frac{1}{t}\partial g(\vv^\star)\cap \S^{m-1}\right)-1\right] \notag \\
	&=\min_{\alpha\leq t\leq \beta} \E\left[2\cdot\dist\left(\vg,\frac{\sqrt{m}}{\lambda t}\partial f(\vx^\star)\right)+\min_{\vb\in\frac{1}{t}\partial g(\vv^\star)\cap \S^{m-1}}\|\vh-\vb\|_2^2-1\right] \notag\\
	&=\min_{\alpha\leq t\leq \beta} \E\left[2\cdot\dist\left(\vg,\frac{\sqrt{m}}{\lambda t}\partial f(\vx^\star)\right)-2\max_{\vb\in\frac{1}{t}\partial g(\vv^\star)\cap \S^{m-1}}\ip{\vh}{\vb}+m\right].
	\end{align}
	It follows from \eqref{Relationship} that
	\begin{align*}
	\sqrt{\eta^2\left(\frac{\sqrt{m}}{\lambda t}\partial f(\vx^\star)\right)-1}&\leq \zeta\left(\frac{\sqrt{m}}{\lambda t}\partial f(\vx^\star)\right) \leq \sqrt{\eta^2\left(\frac{\sqrt{m}}{\lambda t}\partial f(\vx^\star)\right)}.
	\end{align*}
	Substituting the above bound into \eqref{bound of delta_lbd} completes the proof.
\end{proof}

Lemma \ref{calculation_Delta_lambda} demonstrates that $\mathscr{C}_p(\lambda)$ can be accurately estimated by
\begin{equation}\label{cal_ph of pen}
\min_{\alpha\leq t\leq \beta} \left[2\cdot\sqrt{\eta^2\left(\frac{\sqrt{m}}{\lambda t}\partial f(\vx^\star)\right)}-2\cdot\omega\left(\frac{1}{t}\partial g(\vv^\star)\cap \S^{m-1}\right)+m\right].
\end{equation}
Thus it is sufficient to estimate $\eta^2\left(\frac{\sqrt{m}}{\lambda t}\partial f(\vx^\star)\right)$ and $\omega\left(\frac{1}{t}\partial g(\vv^\star)\cap \S^{m-1}\right)$. There are also some standard methods to estimate these two quantities, see e.g., \cite{chandrasekaran2012convex, amelunxen2014living, foygel2014corrupted}. To illustrate this result \eqref{cal_ph of pen}, we also consider two typical examples: sparse signal recovery from sparse corruption and low-rank matrix recovery from sparse corruption. In the first example, we assume the signal $\vx^{\star} \in \R^{n}$ and the corruption $\vv^{\star} \in \R^{m}$ are $s$-sparse and $k$-sparse vectors, then we can obtain (see Appendix \ref{cal} for details)
\begin{align*}
\eta^2\left(\frac{\sqrt{m}}{\lambda t}\partial \|\vx^{\star}\|_1\right)=s\left(1+\frac{m}{\lambda^2t^2}\right)+\frac{2(n-s)}{\sqrt{2\pi}}\left(\left(1+\frac{m}{\lambda^2t^2}\right)\int_{\frac{\sqrt{m}}{\lambda t}}^{\infty}e^{-x^2/2}dx-\frac{\sqrt{m}}{\lambda t}e^{-\frac{m}{2\lambda^2t^2}}\right)
\end{align*}
and
\begin{align*}
\omega\left(\frac{1}{t}\partial \|\vv^\star\|_1\cap \S^{m-1}\right)=\sqrt{\frac{2}{\pi}(m-k)\left(1-\frac{k}{t^2}\right)}.
\end{align*}
The parameter $t$ in \eqref{cal_ph of pen} takes value from $\alpha=\sqrt{k}$ to $\beta=\sqrt{m}$. Consider the example of low-rank matrix recovery from sparse corruption, the signal  $\mX^{\star}\in \R^{n\times n}$ is an $r$-rank matrix, the first Gaussian squared distance in \eqref{cal_ph of pen} can be calculated as
\begin{align*}
\eta^2\left(\frac{\sqrt{m}}{\lambda t}\partial\|\mX\|_*\right)=r\left(2n-r+\frac{m}{\lambda^2t^2}\right)+\E\sum_{i=1}^{n-r}\textrm{shrink}\left(\sigma_i(\mG_2),\frac{\sqrt{m}}{\lambda t}\right)^2.
\end{align*}
where $\mG_2$ is an $(n-r)\times(n-r)$ standard Gaussian matrix, and $\sigma_i(\mG_2)$ is the $i$-th largest singular value of $\mG_2$.  The calculations of the Gaussian width of the corruption and the range of parameter $t$ are similar to the first example. In addition, it is possible to estimate the Gaussian distance and Gaussian width numerically by approximating the expectations in their definitions with empirical averages.

\begin{figure*}
	\centering
	\includegraphics[width= .6\textwidth]{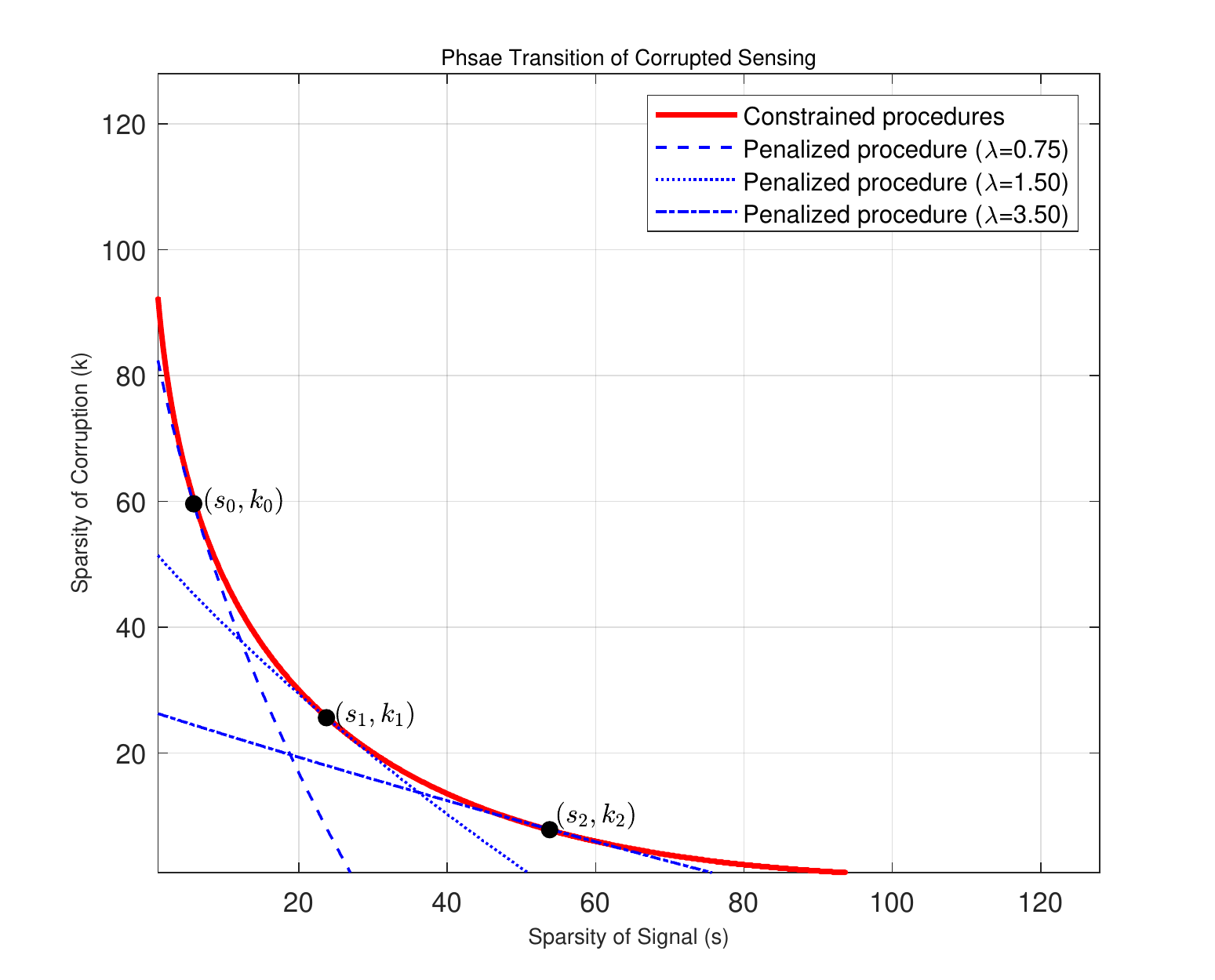}
	\caption{We assume $\vx^\star\in\R^{128}$ is an $s$-sparse vector and $\vv^\star\in\R^{128}$ is a $k$-sparse vector. $f(\cdot)$ and $g(\cdot)$ are set to be the $\ell_1$-norm. Fix the sample size $m=128$. The solid red line corresponds to the phase transition threshold of the constrained procedures: $m=\mathscr{C}_p$, the dashed blue lines correspond to the phase transition thresholds of the penalized recovery procedure with different $\lambda$s: $m=\mathscr{C}_p(\lambda)-\epsilon$. The recovery of both signal and corruption is likely to succeed when the sparsity pair $(s, k)$ lies below the phase thresholds, the recovery is likely to fail when $(s, k)$ lies above the phase thresholds. It is not hard to find that: (I). The successful areas of the penalized recovery procedure with different $\lambda$s are always smaller than that of the constrained methods; (II). Even for the critical points (e.g., $(s, k)=(s_0,k_0), (s_1,k_1), (s_2,k_2)$) that lie on the phase transition threshold of the constrained recovery procedures (which might represent the reconstruction limit of the constrained procedures), we can still choose some corresponding tradeoff parameters (e.g., $\lambda=0.75, 1.50, 3.50$) such that the penalized problem succeeds too (with similar probability). Thus, the successful area of the constrained procedures can be regarded as the union of that of the penalized one with different $\lambda$s.}
	\label{fig:theory}
\end{figure*}

\subsection{Relationship between Constrained and Penalized Recovery Procedures and Optimal Choice of $\lambda$}
	\label{relation between con and pen}
	The theory of Lagrange multipliers \cite[Section 28]{rockafellar1970convex} asserts that solving the constrained recovery procedures is essentially equivalent to solving the penalized problem with a best choice of the tradeoff parameter $\lambda$. More precisely, this equivalence consists of the following two aspects \cite[Appendix A]{mccoy2014sharp}. On the one hand,
\begin{itemize}
  \item [(I).] \emph{Suppose the penalized procedure \eqref{penilized_procedure} succeeds for some value $\lambda>0$. Then the constrained ones \eqref{constrained_procedure1} and \eqref{constrained_procedure2} succeed.}
\end{itemize}
On the other hand, as a partial converse to (I), one has
\begin{itemize}
  \item [(II).] \emph{Suppose that the subdifferentials $\partial f(\vx^\star)$ and $\partial g(\vv^\star)$ do not contain the origin. If the constrained procedures \eqref{constrained_procedure1} and \eqref{constrained_procedure2} succeed, then there exists a parameter $\lambda>0$ such that $(\vx^\star,\vv^\star)$ is an optimal point for the penalized one \eqref{penilized_procedure}.}
\end{itemize}

The above relations indicate that the performance of the constrained procedures can be interpreted as the best possible one for the penalized problem. However, the main difficulty in these results lies in how to select a suitable tradeoff parameter $\lambda$ that leads to this equivalence. Since we have identified the precise phase transitions of both constrained and penalized recovery procedures, it is possible to allow us to explore the relationship between these two kinds of approaches in a quantitative way, which in turn implies an explicit strategy to choose the optimal $\lambda$.

   \begin{theorem}[Relationship between constrained and penalized recovery procedures and optimal choice of $\lambda$]\label{relation_CC_PP}
		Assume that $\TT_f(\vx^{\star})$ and $\TT_g(\vv^{\star})$ are non-empty and closed, and that the subdifferentials $\partial f(\vx^\star)$ and $\partial g(\vv^\star)$ do not contain the origin.
		If $m\geq \mathscr{C}_p(\lambda)$, then we have
		\begin{align*}
		m \geq \mathscr{C}_p-1.
		\end{align*}
		On the other hand, if $m\geq \mathscr{C}_p$, then we can choose the tradeoff parameter as
    \begin{equation}\label{optimal_lambda}
       \lambda^\star=\arg\min_{\lambda>0}\mathscr{C}_p(\lambda),
    \end{equation}
such that \footnote{As shown in the proof of Theorem \ref{relation_CC_PP}, the gap $5$ can be easily reduced by introducing an extra condition $\omega(\TT_f(\vx^\star)\cap\S^{n-1})\geq 4$, namely, if we further let $\omega(\TT_f(\vx^\star)\cap\S^{n-1})\geq 4$, then we have $m\geq \mathscr{C}_p(\lambda^\star)-1$. It is worth noting that this condition is easy to satisfy in practical applications.}
		\begin{align*}
		m\geq \mathscr{C}_p(\lambda^\star)-5.
		\end{align*}
	\end{theorem}

Combining the phase transition results in Theorems \ref{them: Phase transition_con} and \ref{them: Phase transition}, the first part of Theorem \ref{relation_CC_PP} implies that if the penalized procedure \eqref{penilized_procedure} is likely to succeed, then the constrained procedures \eqref{constrained_procedure1} and \eqref{constrained_procedure2} succeed with high probability. Similarly, the second part of Theorem \ref{relation_CC_PP} conveys that if the constrained procedures \eqref{constrained_procedure1} and \eqref{constrained_procedure2} are likely to succeed, then we can choose the tradeoff parameter $\lambda$ as in \eqref{optimal_lambda} such that the penalized procedure \eqref{penilized_procedure} succeeds with high probability. Thus our results provide a quantitative characterization for the relations (I) and (II).

The results in Theorem \ref{relation_CC_PP} also enjoy a geometrical explanation in the phase transition program: The first part implies that the successful area of penalized recovery procedure should be smaller than that of the constrained procedures. The second part indicates that for any point in the successful area of the constrained recovery procedures, we can find at least a $\lambda$ such that this point also belongs to the successful area of the corresponding penalized recovery procedure. In other words, the successful area of the constrained procedures can be regarded as the union of that of the penalized one (with different $\lambda$s). Fig.\ref{fig:theory} illustrates this relationship in the case of sparse signal recovery from sparse corruption.

Moreover, Theorem \ref{relation_CC_PP} has suggested an explicit way to choose the best parameter $\lambda$ predicted by the Lagrange theory, i.e.,
$$
\lambda^\star=\arg\min_{\lambda>0}\mathscr{C}_p(\lambda),
$$
which is equivalent to
	\begin{align}\label{optimal_lbd}
	\lambda^\star=\arg\min_{\lambda>0}\zeta\left(\frac{\sqrt m}{\lambda t^\star} \partial f(\vx^{\star})\right) ~~ \textrm{with} ~~ t^\star=\arg\min_{\alpha\leq t\leq \beta}\eta^2\left(\frac{1}{t} \partial g(\vv^{\star})\cap\S^{m-1}\right).
	\end{align}
We provide some insights for this parameter selection strategy. Recall that Theorem \ref{them: Phase transition} has demonstrated that the penalized procedure succeeds with high probability if the number of measurements exceeds the critical point $\mathscr{C}_p(\lambda)$. Then the strategy \eqref{optimal_lambda} implies that we should pick the $\lambda$ which makes the number of observations required for successful recovery of the penalized procedure as small as possible. Another explanation comes from the relationship between these two kinds of recovery procedures. For a given corrupted sensing problem (with fixed $\vx^{\star}$ and $\vv^{\star}$), the first part of Theorem \ref{relation_CC_PP} indicates that the phase transition threshold of the penalized procedure is always bounded from below by that of constrained ones, it is natural to choose the $\lambda$ such that we can achieve the possibly smallest gap between these two thresholds i.e., $\lambda^\star=\arg\min_{\lambda>0}(\mathscr{C}_p(\lambda)-\mathscr{C}_p)$, which also leads to the strategy \eqref{optimal_lambda}.

	\begin{remark}[Related works]
		In \cite{foygel2014corrupted} and \cite{Jinchi2018Stable}, the authors also provide an explicit way to select the tradeoff parameter $\lambda$. Specifically, their results have shown that $\OO\left( \eta^2(\lambda_1 \cdot \partial f(\vx^{\star}))+ \eta^2(\lambda_2 \cdot \partial g(\vv^{\star})) \right)$ measurements are sufficient to guarantee the success of the penalized procedure. In order to achieve the smallest number of measurements, it is natural to choose the $\lambda$ as follows:
	\begin{equation}\label{Strategy}
		\lambda_1^{\ast}=\arg\min_{\lambda_1>0}\eta^2(\lambda_1 \cdot \partial f(\vx^{\star})),~~\lambda_2^{\ast}=\arg\min_{\lambda_2>0}\eta^2(\lambda_2 \cdot \partial g(\vv^{\star})),~~\textrm{and}~\lambda^{\ast}=\lambda_2^{\ast}/\lambda_1^{\ast}.
	\end{equation}
However, a visible mismatch between the penalized program with the strategy \eqref{Strategy} and the constrained ones has been observed in their numerical experiments. This suggests that the choice \eqref{Strategy} might not be optimal in the sense of the Lagrange theory. As shown in our simulations (Section \ref{simulations}), the empirical performance of the penalized procedure \eqref{penilized_procedure} with our optimal choice of the tradeoff parameter is nearly the same as that of the constrained convex procedures.  Thus, our strategy \eqref{optimal_lambda} solves another significant open problem in \cite{mccoy2014sharp}.
	\end{remark}

\section{Numerical Simulations}\label{simulations}
In this section, we perform a series of numerical experiments to verify our theoretical results. We consider two typical structured signal recovery problems: sparse signal recovery from sparse corruption and low-rank matrix recovery from sparse corruption. In each case, we employ both constrained and penalized recovery procedures to reconstruct the original signal and corruption. Throughout these experiments, the related convex optimization problems are solved by CVX Matlab package \cite{grant2008cvx,grant2008graph}. In addition to the Gaussian measurements, we also consider sub-Gaussian measurements \footnote{In fact, we have tested other distributions of $\mPhi$ such as sparse Rademacher distribution and Student's $t$ distribution, the obtained results are quite similar, so we omit them here.}.


\subsection{Phase Transition of the Constrained Recovery Procedures}
We first consider the empirical behavior of the constrained recovery procedures in the following two structured signal recovery problems.

\begin{figure*}
	\centering
	\subfigure{
		\includegraphics[width=0.48\textwidth]{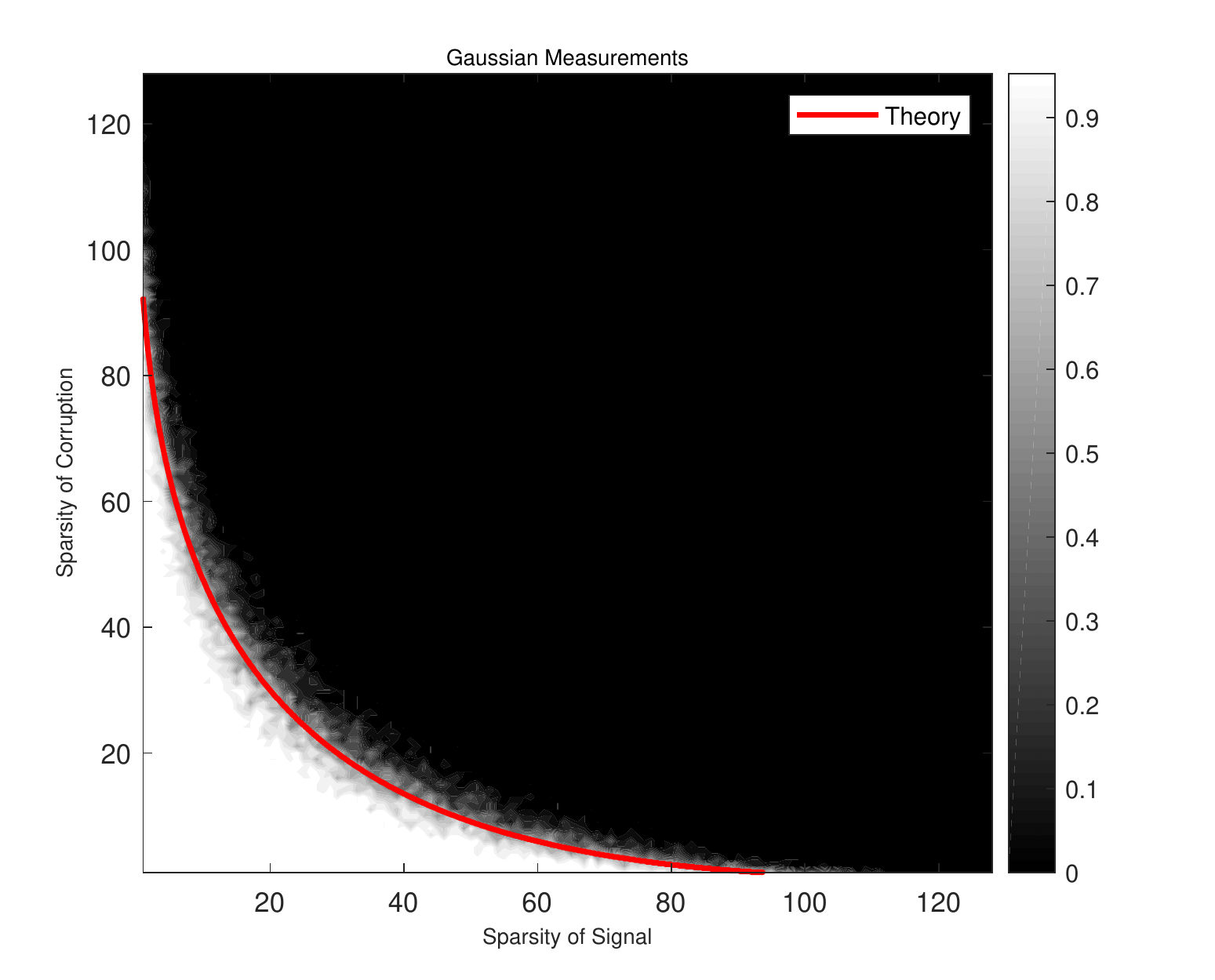}
		\includegraphics[width=0.48\textwidth]{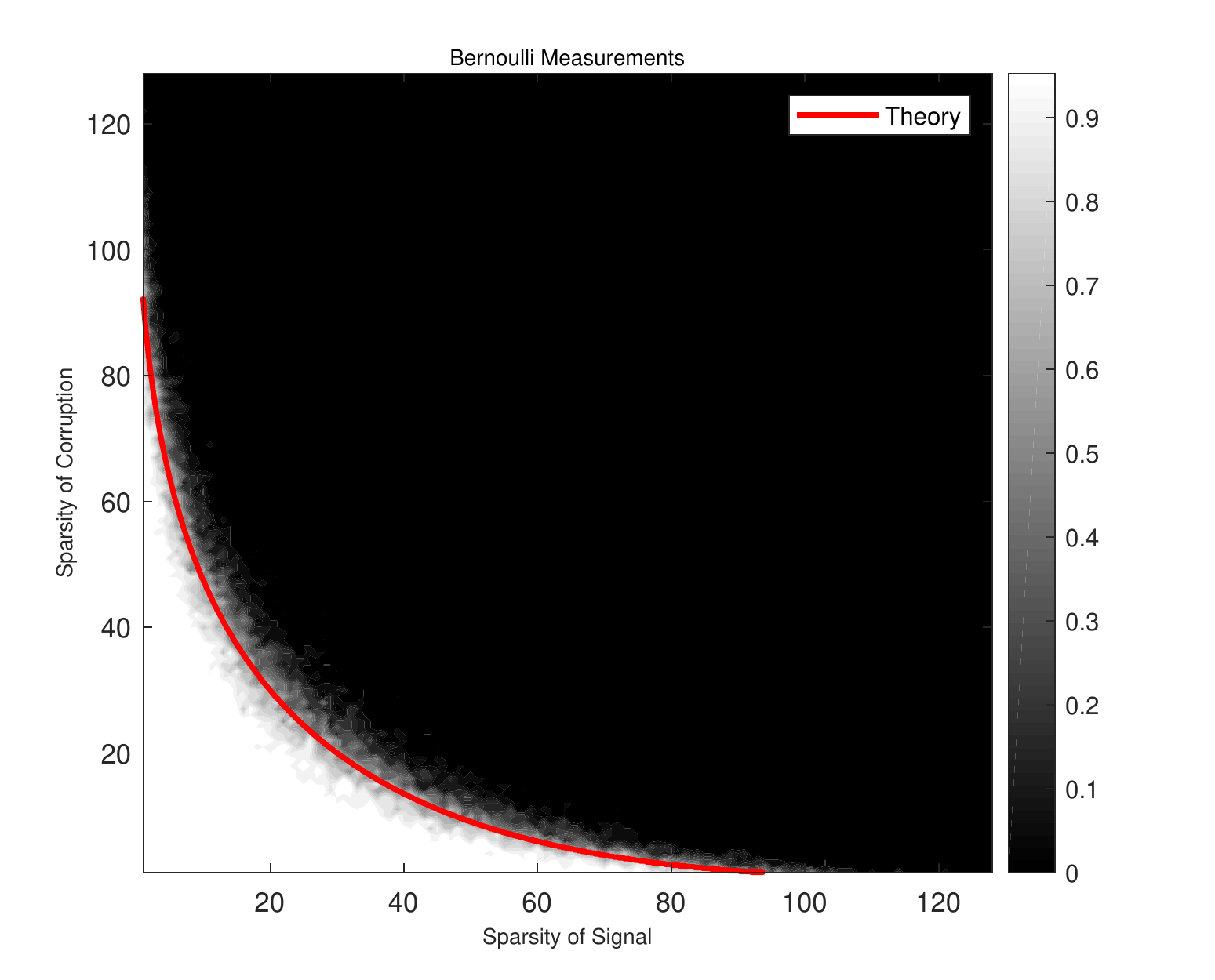}
	}
	\subfigure{
		\includegraphics[width=0.48\textwidth]{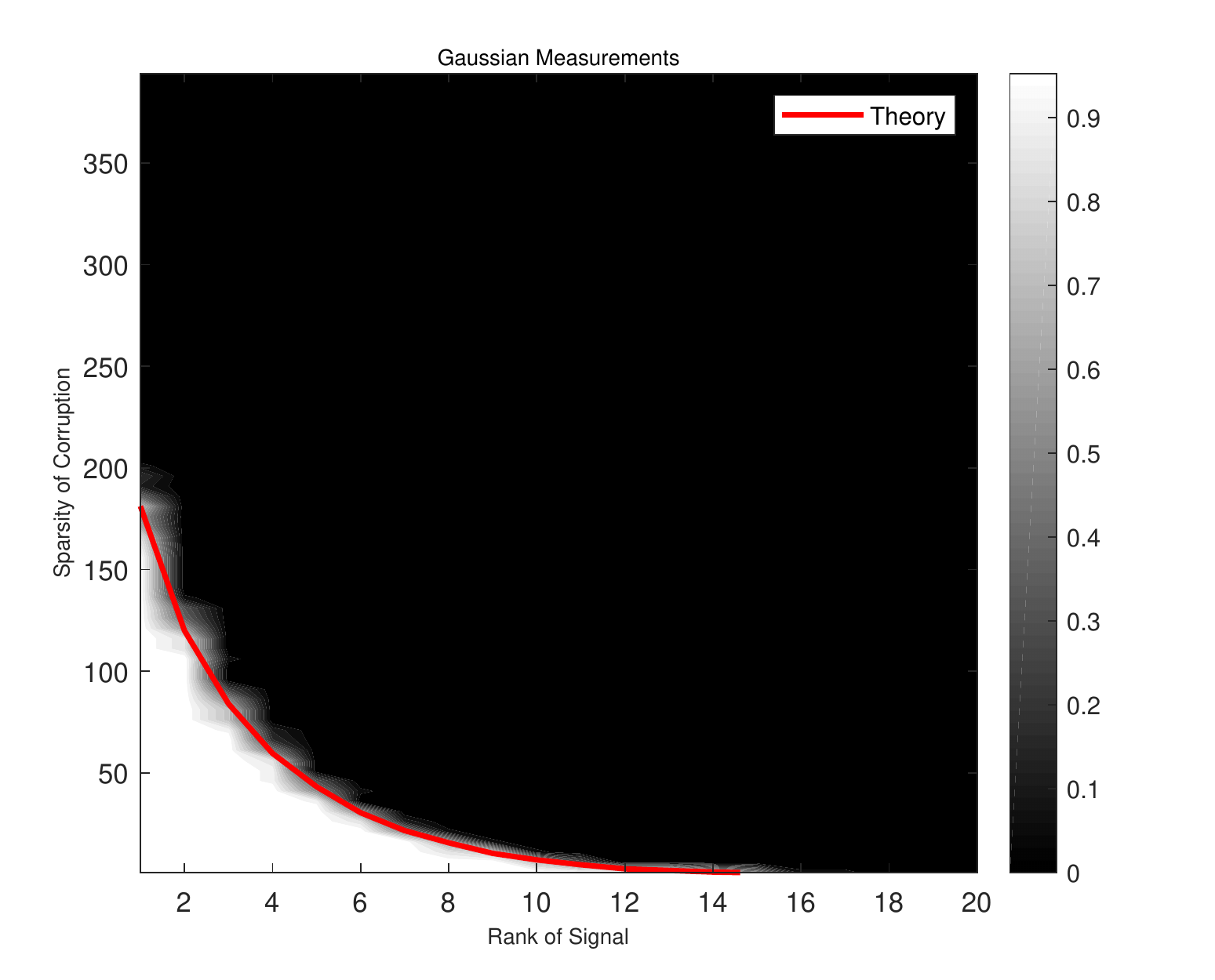}
		\includegraphics[width=0.48\textwidth]{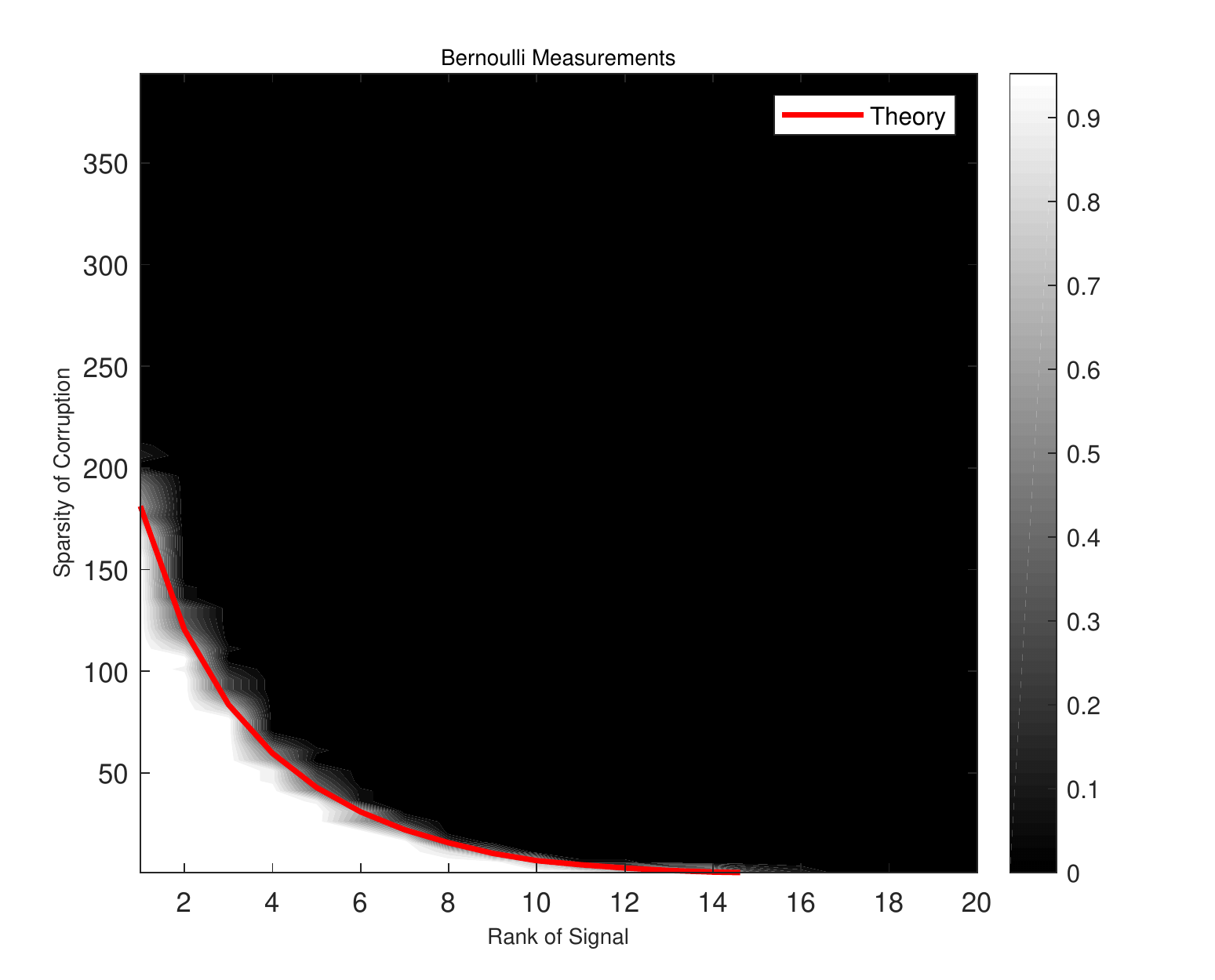}
	}
	\caption{Phase transitions of the constrained recovery procedures under Gaussian and Bernoulli measurements in both sparse signal recovery from sparse corruption and low-rank matrix recovery from sparse corruption. The red curves plot the phase transition thresholds predicted by \eqref{threshold}.}
	\label{fig:Constrained_Procedure}
\end{figure*}

\subsubsection{Sparse Signal Recovery from Sparse Corruption}
In this example, both signal and corruption are sparse, and we use the $\ell_1$-norm to promote their structures, i.e., $f(\vx^\star)=\|\vx^\star\|_1$ and $g(\vv^\star)=\|\vv^\star\|_1$. Suppose the $\ell_1$-norm of the true signal $\|\vx^\star\|_1$ are known beforehand. We fix the sample size and ambient signal dimension $m=n=128$. For each signal sparsity $s=1,2,...,128$ and each corruption sparsity $k=1,2,...,128$. We repeat the following experiments 20 times:
\begin{itemize}
	\item [(1)]
	Generate a signal vector $\vx^\star\in\R^n$ with $s$ non-zero entries and set the other $n-s$ entries to 0. The locations of the non-zero entries are uniformly selected among all possible supports, and nonzero entries are independently sampled from the normal distribution.
	\item [(2)]
	Similarly, generate a corruption vector $\vv^\star\in\R^m$ with $k$ non-zero entries and set the other $m-k$ entries to 0.
	\item [(3)]
	For Gaussian measurements, we draw the sensing matrix $\mPhi\in\R^{m\times n}$ with i.i.d. standard normal entries. For sub-Gaussian measurements, we draw the sensing matrix $\mPhi\in\R^{m\times n}$ with i.i.d. symmetric Bernoulli entries.
	\item [(4)]
	Solve the following constrained optimization problem \eqref{constrained_procedure2}:
	\begin{align*}
	(\hat{\vx}, \hat{\vv})=\arg\min_{\vx, \vv} ~\|\vv\|_1,\quad\text{s.t.~}&\vy=\mPhi\vx+\sqrt{m}\vv,~~\|\vx\|_1\leq \|\vx^{\star}\|_1.
	\end{align*}
	\item [(5)]
	Set $tol=10^{-3}$. Declare success if $\|\hat{\vx}-\vx^{\star}\|_2/\|\vx^{\star}\|_2\leq tol$.
	
\end{itemize}

\begin{figure*}
	\centering
	\subfigure{
		\includegraphics[width=0.48\textwidth]{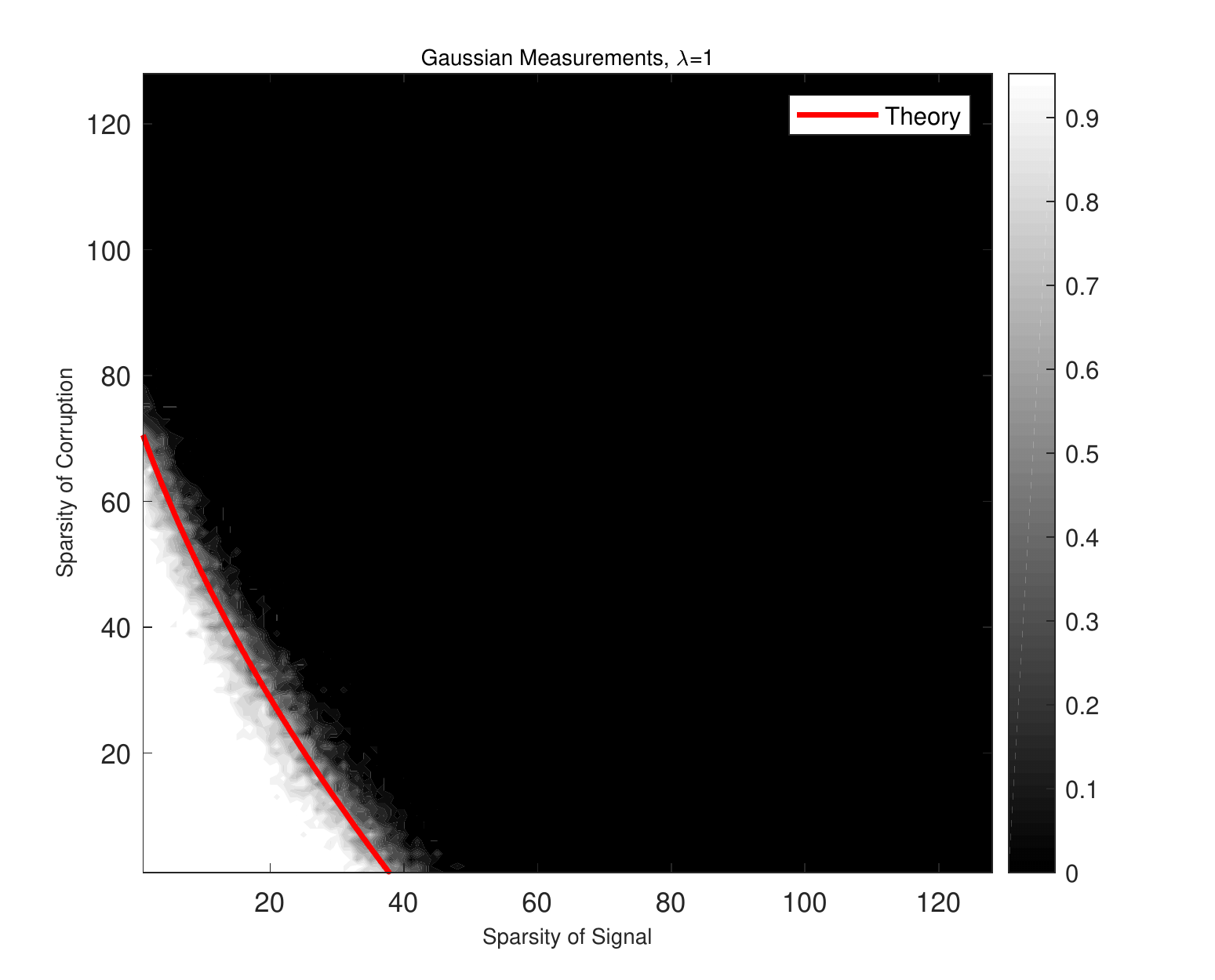}
		\includegraphics[width=0.48\textwidth]{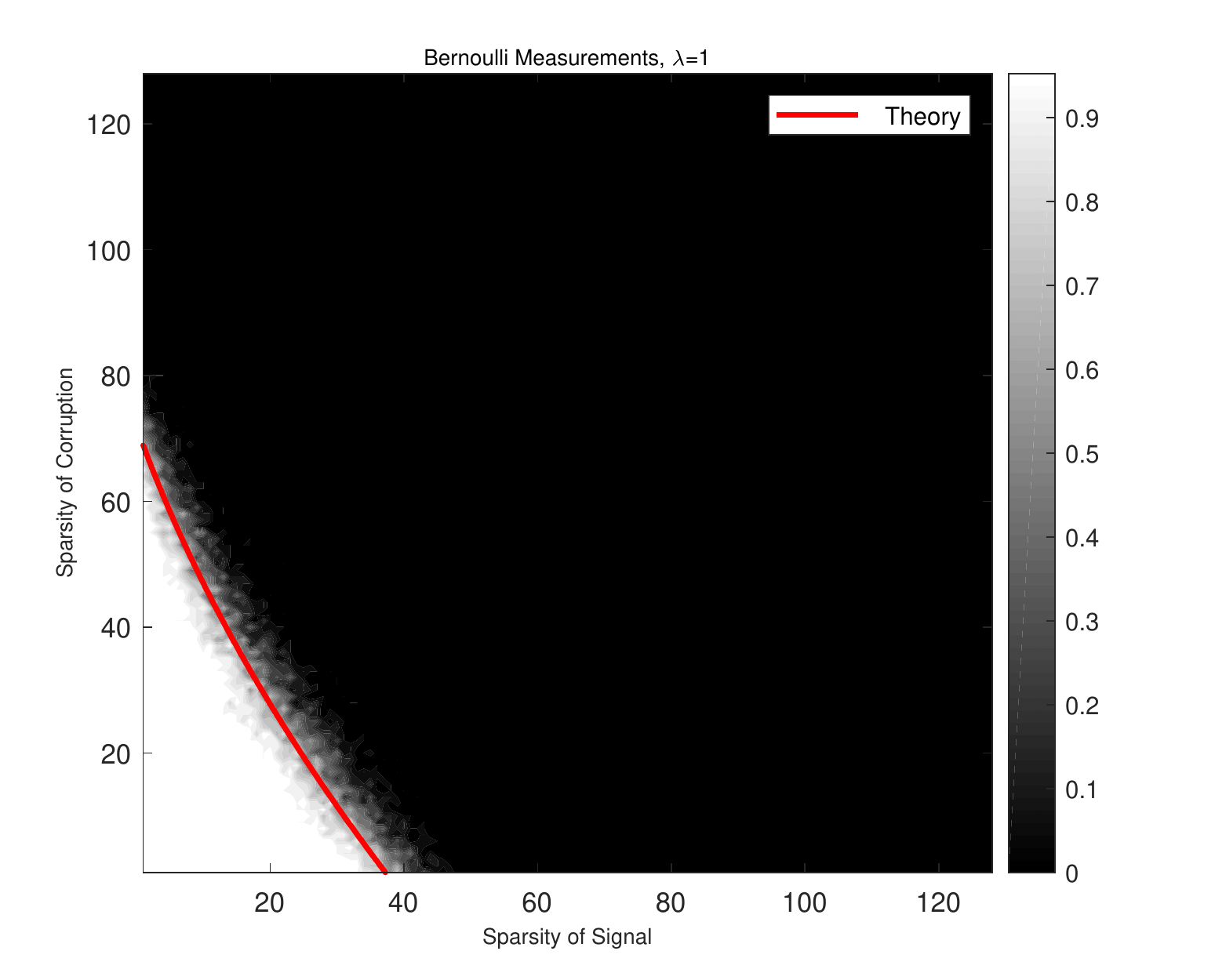}
	}
	\subfigure{
		\includegraphics[width=0.48\textwidth]{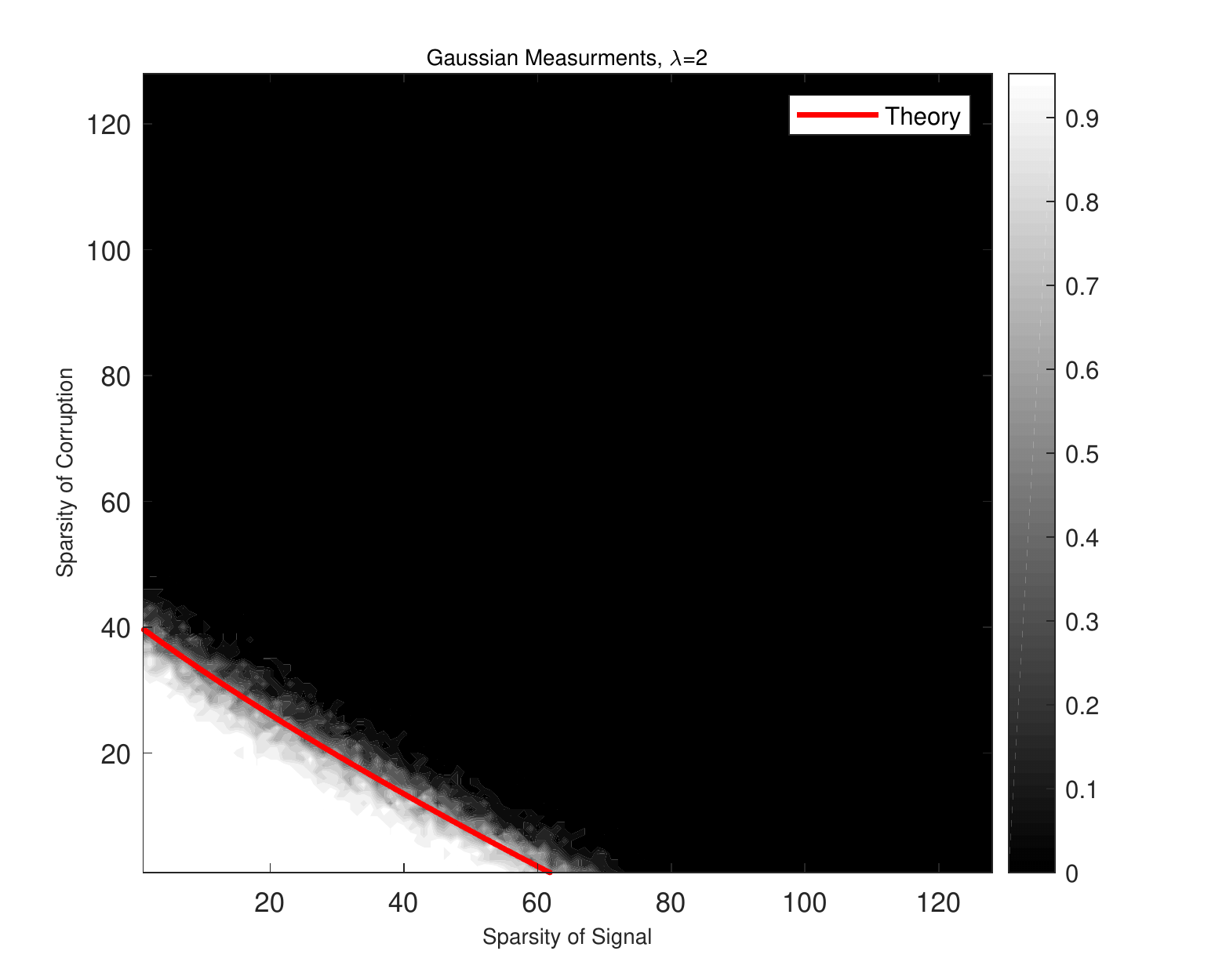}
		\includegraphics[width=0.48\textwidth]{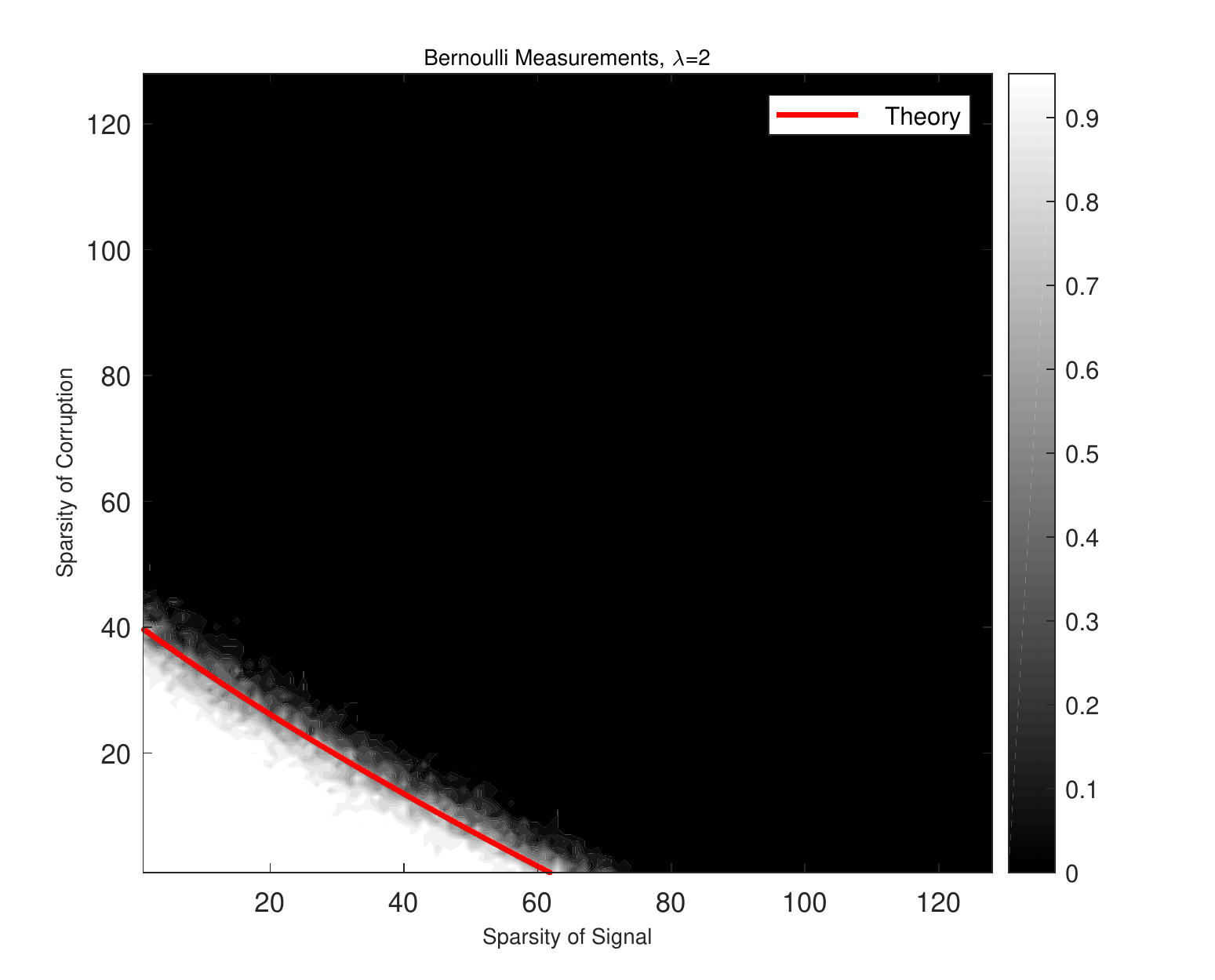}
	}
	\caption{Phase transitions of the penalized recovery procedure with different $\lambda$s under Gaussian and Bernoulli measurements in sparse signal recovery from sparse corruption. The red curves plot the phase transition thresholds predicted by \eqref{threshold2}.}
	\label{fig1:Penalized_Procedure}
\end{figure*}

\subsubsection{Low-rank Matrix Recovery from Sparse Corruption}
In this case, the desired signal is an $r$-rank matrix and the corruption is a $\rho$-sparse vector. We use the nuclear norm $f(\vx^\star)=\|\vx^\star\|_*$ to promote the structure of signal. Suppose the nuclear norm of true signal $\|\vx^\star\|_*$ are known exactly. Let $n=20$ and consider $n\times n$ signal matrices. Set the sample size $m=n^2$. For each rank $r=1,2,...,20$ and each corruption sparsity $\rho=1,6,11,16,...,396$. We repeat the following experiment 20 times:
\begin{itemize}
	\item [(1)]
	Generate an $r$-rank matrix $\mX^\star=\mU_1\mU_2^T$, where $\mU_1$ and $\mU_2$ are independent $n\times r$ matrices with orthonormal columns.
	\item [(2)]
	Generate a corruption vector $\vv^\star\in\R^m$ with $\rho$ non-zero entries and set the other $m-\rho$ entries to 0.
	\item [(3)]
	For Gaussian measurements, we draw the sensing matrix $\mPhi\in\R^{m\times n^2}$ with i.i.d standard normal entries. For sub-Gaussian measurements, we draw the sensing matrix $\mPhi\in\R^{m\times n^2}$ with i.i.d. symmetric Bernoulli entries.
	\item [(4)]
	Solve the following constrained problem \eqref{constrained_procedure2}:
	\begin{align*}
	(\hat{\mX}, \hat{\vv})=\arg\min_{\mX, \vv} ~\|\vv\|_1,\quad\text{s.t.~}&\vy=\mPhi\cdot\textrm{vec}(\mX)+\sqrt{m}\vv,~~\|\mX\|_*\leq \|\mX^{\star}\|_*.
	\end{align*}
	\item [(5)]
	Set $tol=10^{-3}$. Declare success if $\|\hat{\mX}-\mX^{\star}\|_F/\|\mX^{\star}\|_F\leq tol$.
	
\end{itemize}

\begin{figure*}
	\centering
	\subfigure{
		\includegraphics[width=0.48\textwidth]{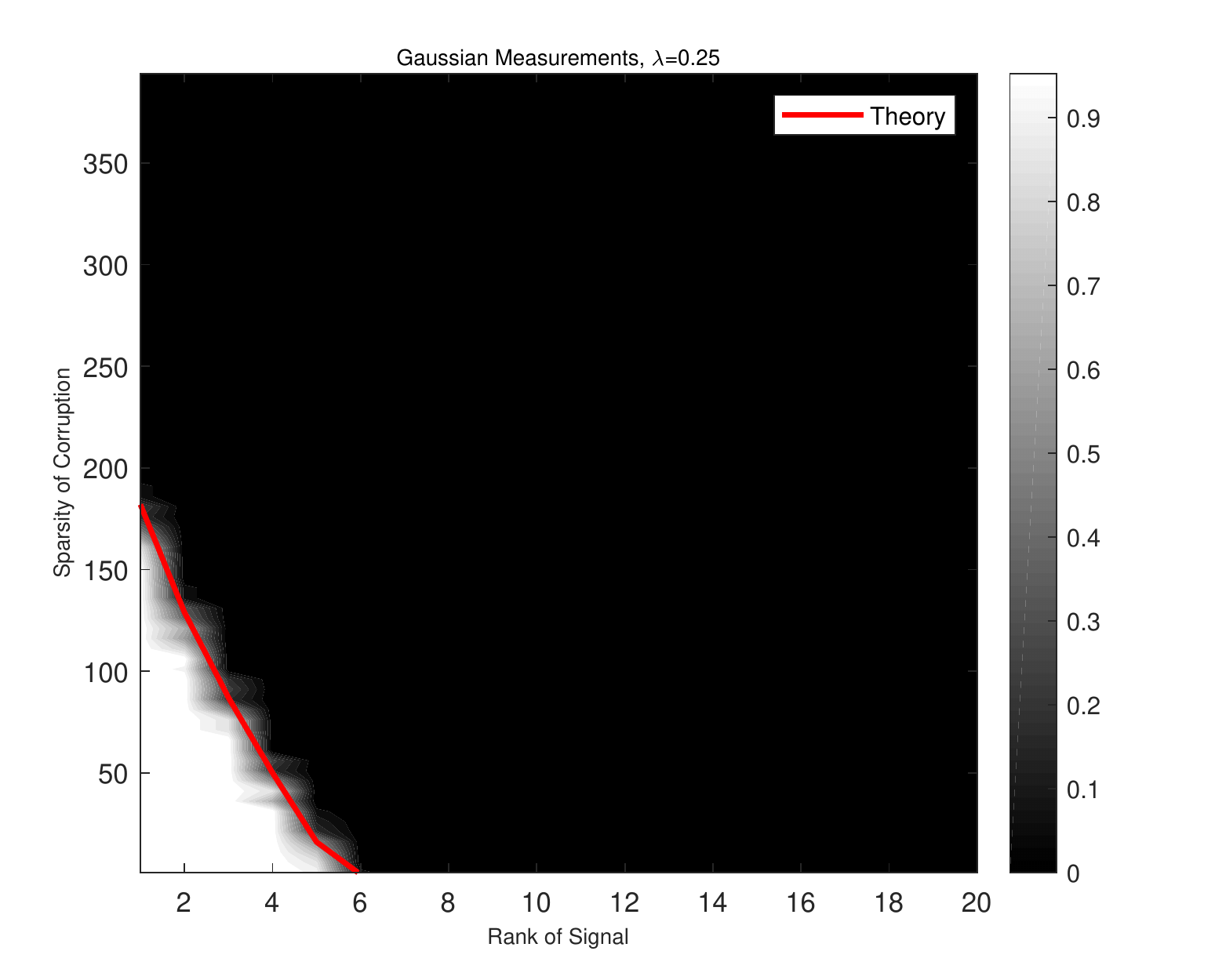}
		\includegraphics[width=0.48\textwidth]{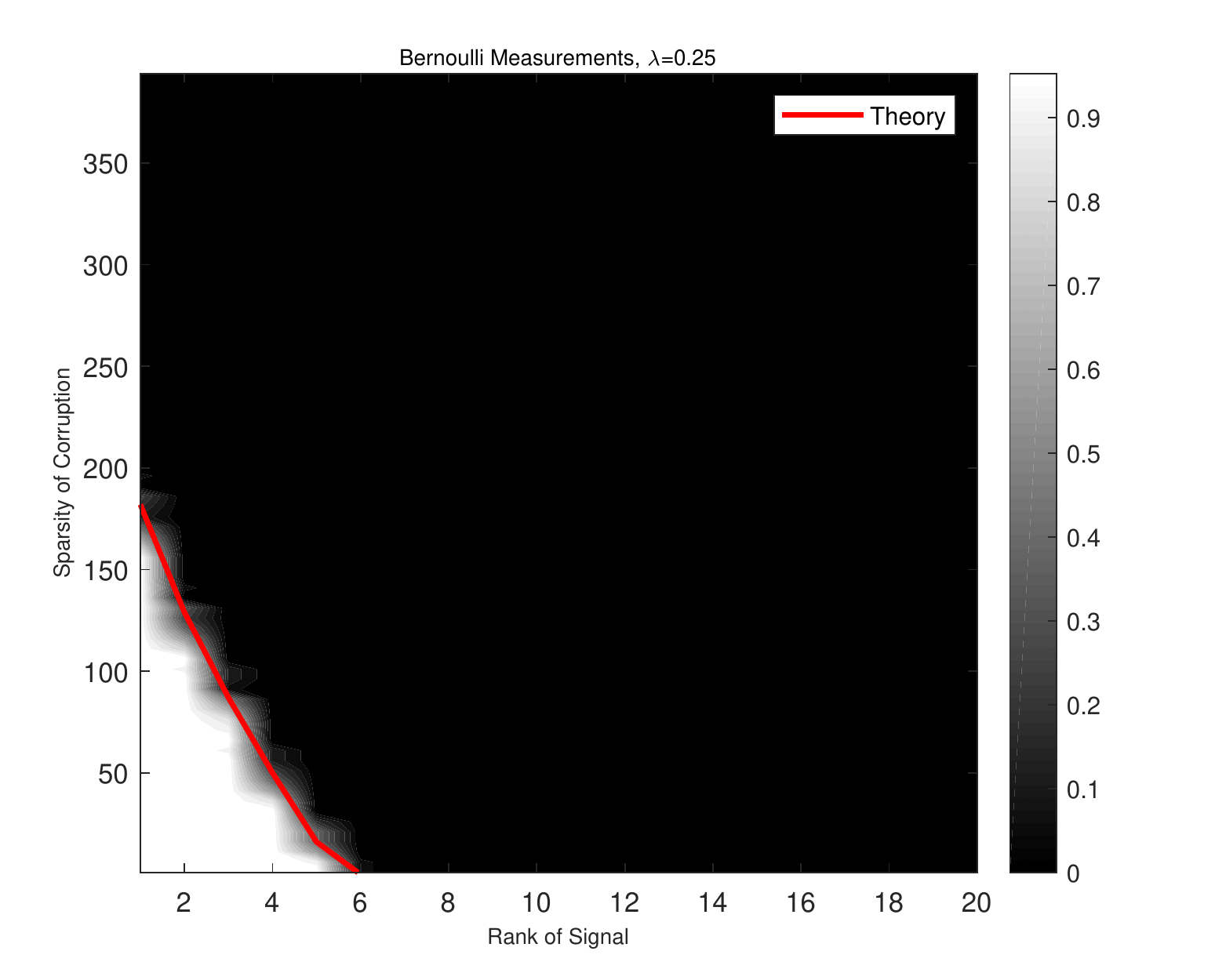}
	}
	\subfigure{
		\includegraphics[width=0.48\textwidth]{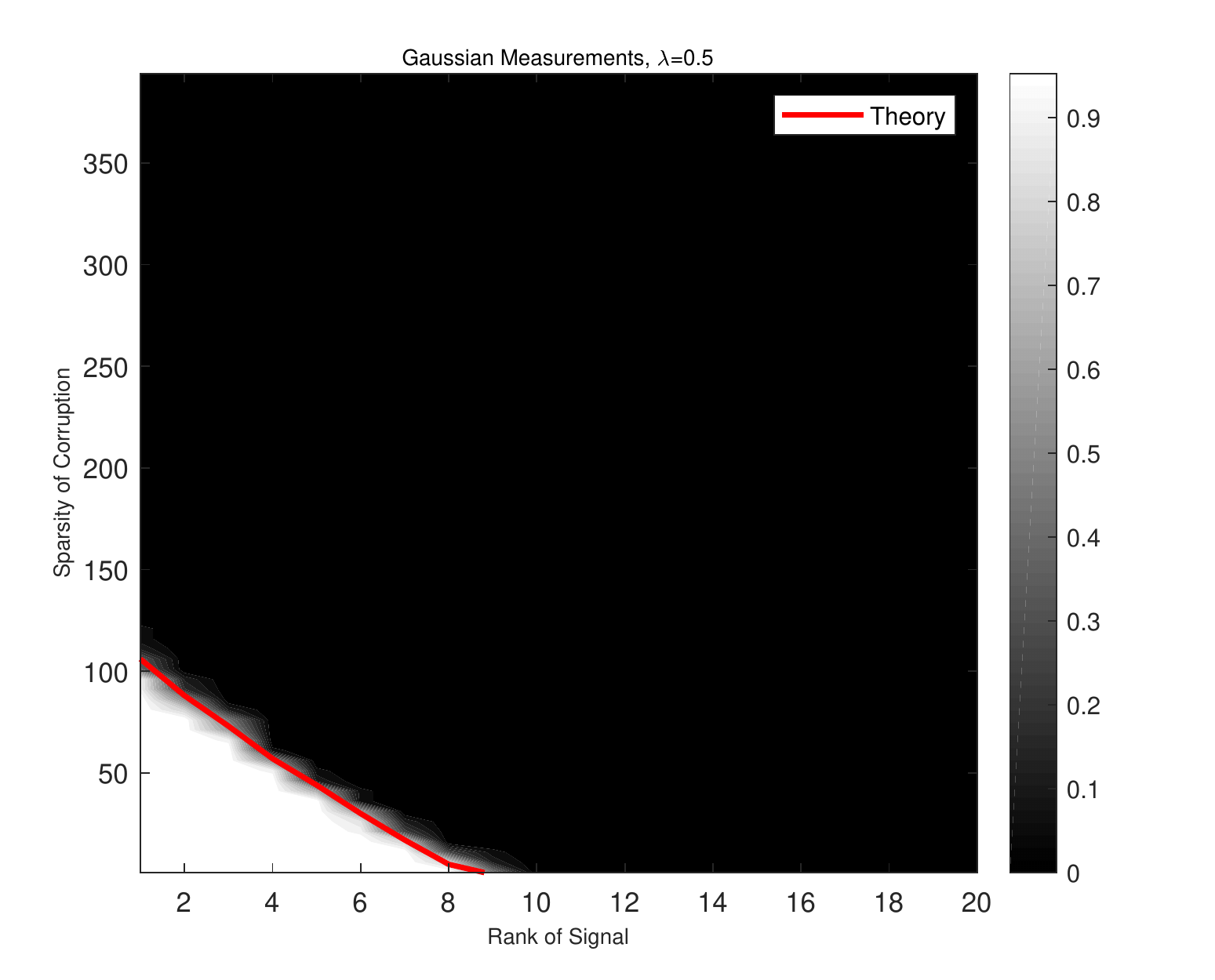}
		\includegraphics[width=0.48\textwidth]{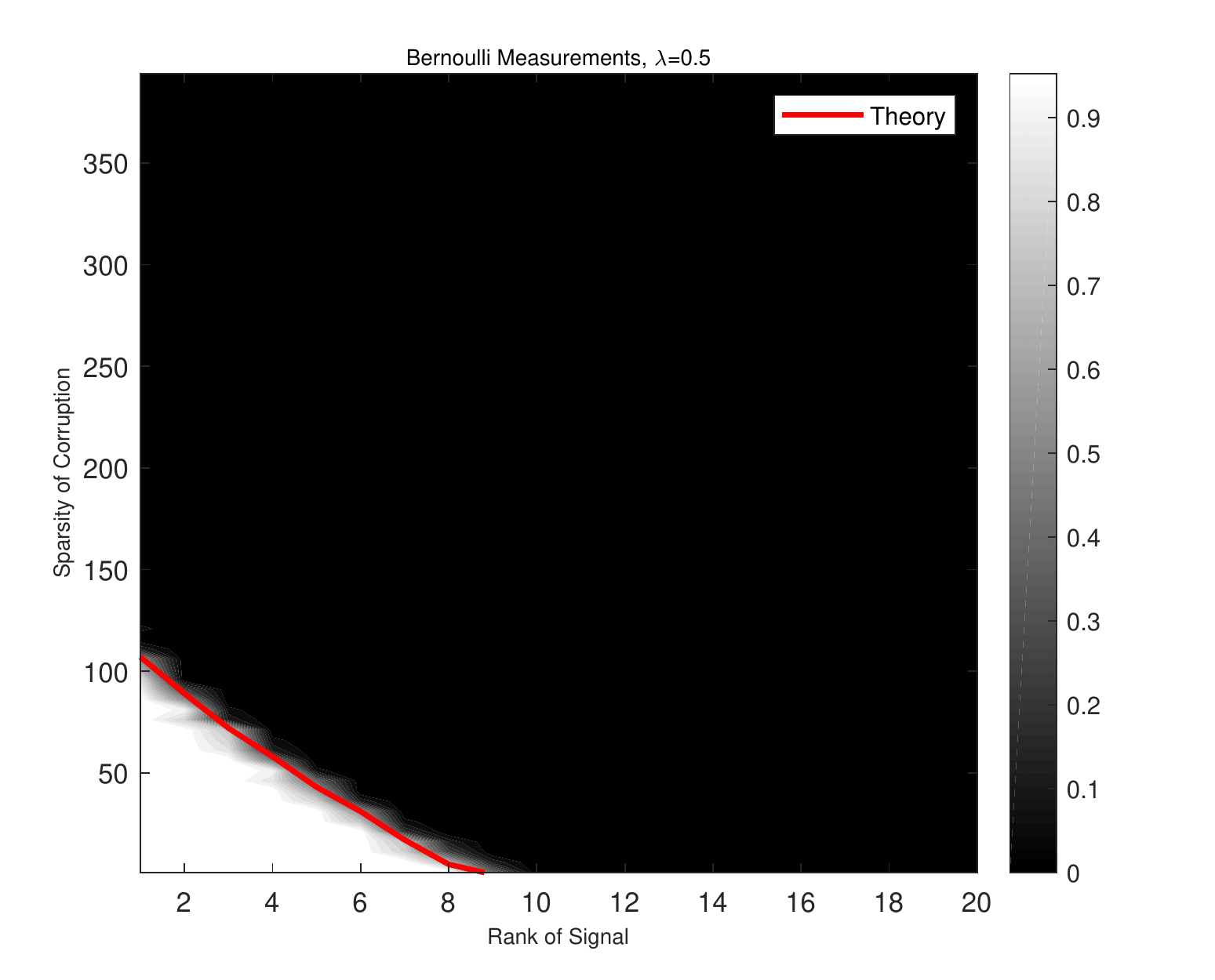}
	}
	\caption{Phase transitions of the penalized recovery procedure with different $\lambda$s under Gaussian and Bernoulli measurements in low-rank matrix recovery from sparse corruption. The red curves plot the phase transition thresholds predicted by \eqref{threshold2}.}
	\label{fig2:Penalized_Procedure}
\end{figure*}
In order to compare the empirical behaviors with theoretical results, we overlay the phase transition curve that predicted in Theorem \ref{them: Phase transition_con}:
\begin{align}\label{threshold}
	\omega^2(\TT_f(\vx^{\star})\cap\S^{n-1})+\omega^2(\TT_g(\vv^{\star})\cap\S^{m-1}).
\end{align}
Fig.\ref{fig:Constrained_Procedure} reports the empirical probability of success for the constrained procedures in these two typical structured signal recovery problems. It is not hard to find that our theoretical predictions sharply align with the empirical phase transitions under both Gaussian and Bernoulli measurements.

\subsection{Phase Transition of the Penalized Recovery Procedure}

We next consider the empirical phase transition of the penalized procedure in these two examples.
\subsubsection{Sparse Signal Recovery from Sparse Corruption}
The experiment settings are almost the same as the constrained case except that we require neither $f(\vx^\star)$ nor $g(\vv^\star)$, and we solve the following penalized procedure instead of the constrained one in step (4):
\begin{align*}
(\hat{\vx}, \hat{\vv})=\arg\min_{\vx, \vv} ~\|\vx\|_1+\lambda\|\vv\|_1,\quad\text{s.t.~}&\vy=\mPhi\vx+\sqrt{m}\vv.
\end{align*}
Here we test two tradeoff parameters: $\lambda=1$ and $\lambda=2$. 
To compare the empirical behaviors with theoretical results, we overlay the phase transition curve that predicted in Theorem \ref{them: Phase transition}:
\begin{equation}\label{threshold2}
\min_{\alpha\leq t\leq \beta}   2 \cdot \zeta\left(\frac{\sqrt{m}}{\lambda t}\partial f(\vx^\star)\right)+ \eta^2\left(\frac{1}{t}\partial g(\vv^\star)\cap \S^{m-1}\right) -1.
\end{equation}

Fig. \ref{fig1:Penalized_Procedure} displays the average empirical probability of success for the penalized problem in sparse signal recovery from sparse corruption. We can see that the theoretical threshold \eqref{threshold2} perfectly predicts the empirical phase transition under different tradeoff parameter $\lambda$s.

\subsubsection{Low-rank Matrix Recovery from Sparse Corruption}
Similarly, the experiment settings are nearly the same as the constrained case except that we recover the original signal and corruption via the following penalized procedure in step (4):
\begin{align*}
(\hat{\mX}, \hat{\vv})=\arg\min_{\mX, \vv} ~\|\mX\|_*+\lambda\|\vv\|_1,\quad\text{s.t.~}&\vy=\mPhi\cdot\textrm{vec}(\mX)+\sqrt{m}\vv.
\end{align*}
The tradeoff parameter is set to be $\lambda=1/4$ or $\lambda=1/2$. To compare our theory with the empirical results, we overlay the theoretical threshold \eqref{threshold2}. 
%
Fig. \ref{fig2:Penalized_Procedure} shows the empirical probability of success for penalized problem in low-rank matrix recovery from sparse corruption. We can find that the theoretical threshold \eqref{threshold2} predicts the empirical phase transition quite well under different tradeoff parameter $\lambda$s.

\begin{figure*}
	\centering
	\subfigure{
		\includegraphics[width=0.48\textwidth]{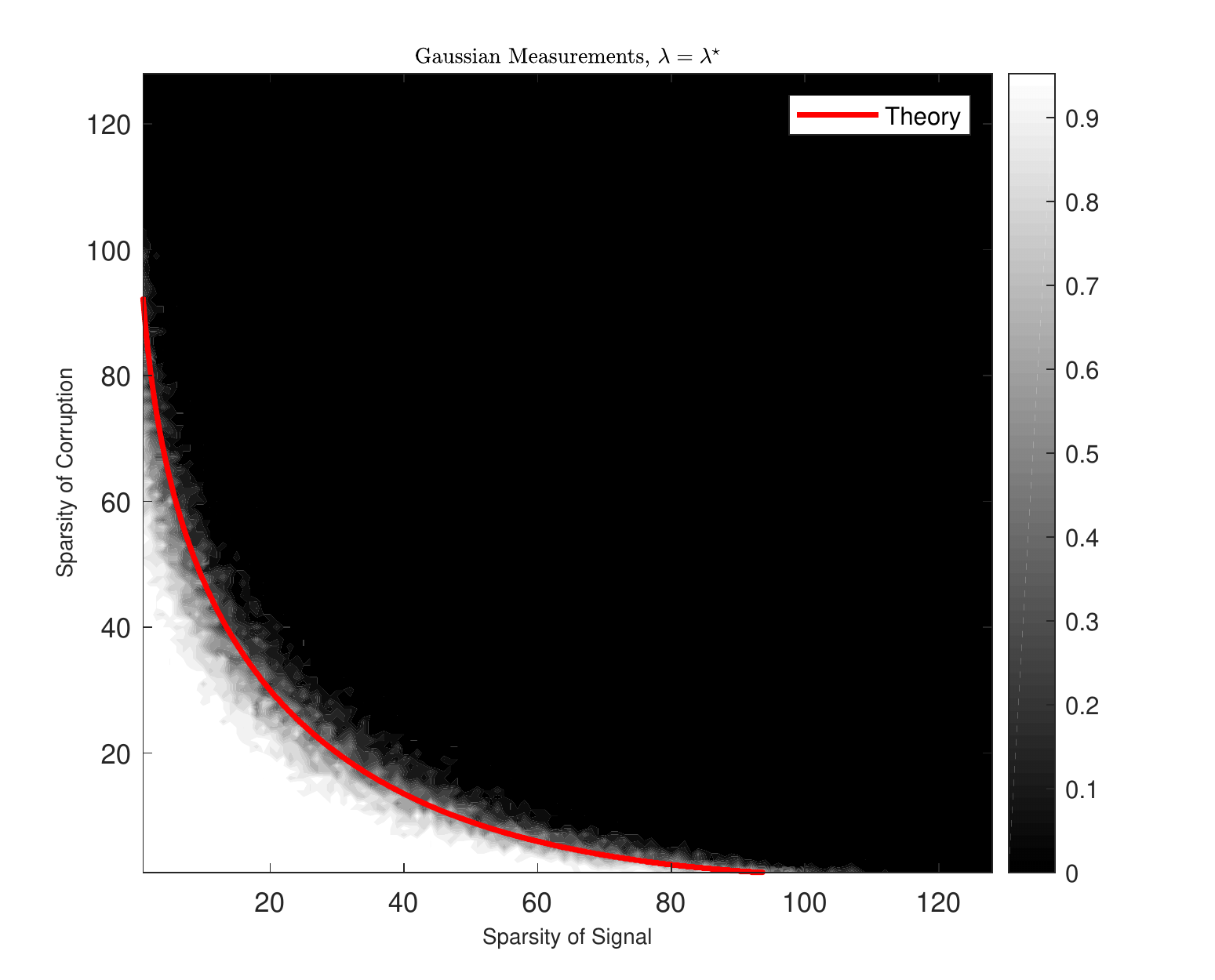}
		\includegraphics[width=0.48\textwidth]{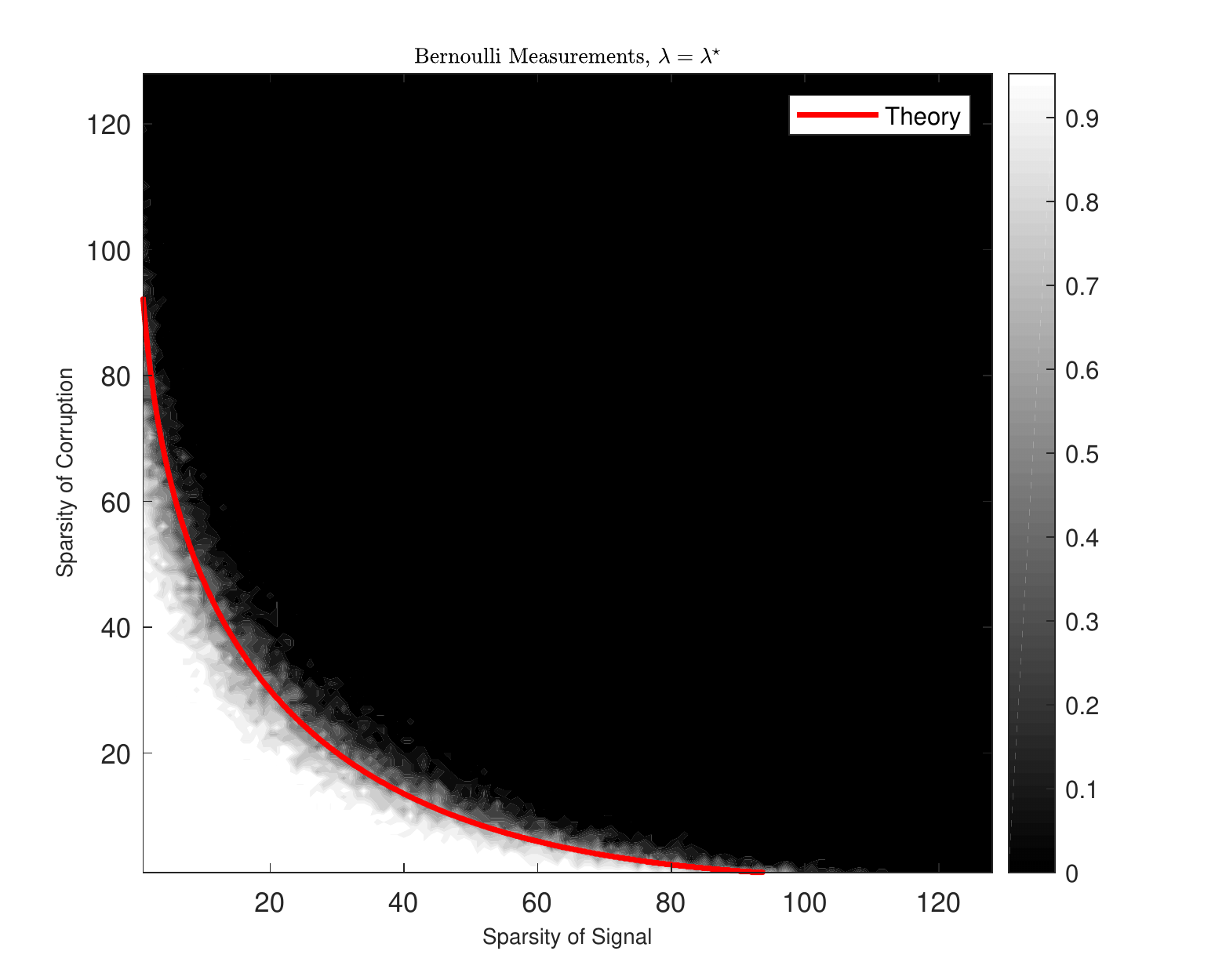}
	}
	\subfigure{
		\includegraphics[width=0.48\textwidth]{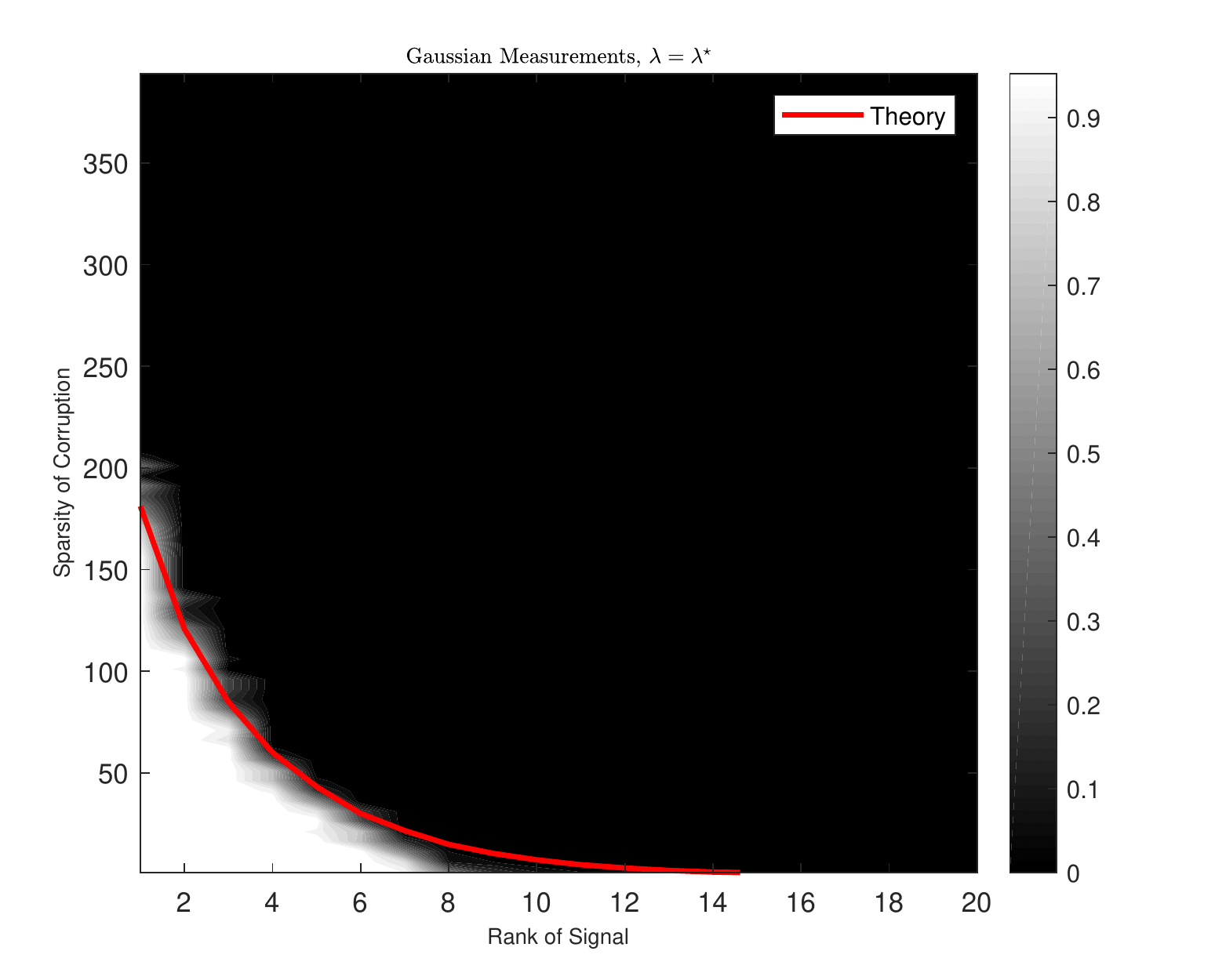}
		\includegraphics[width=0.48\textwidth]{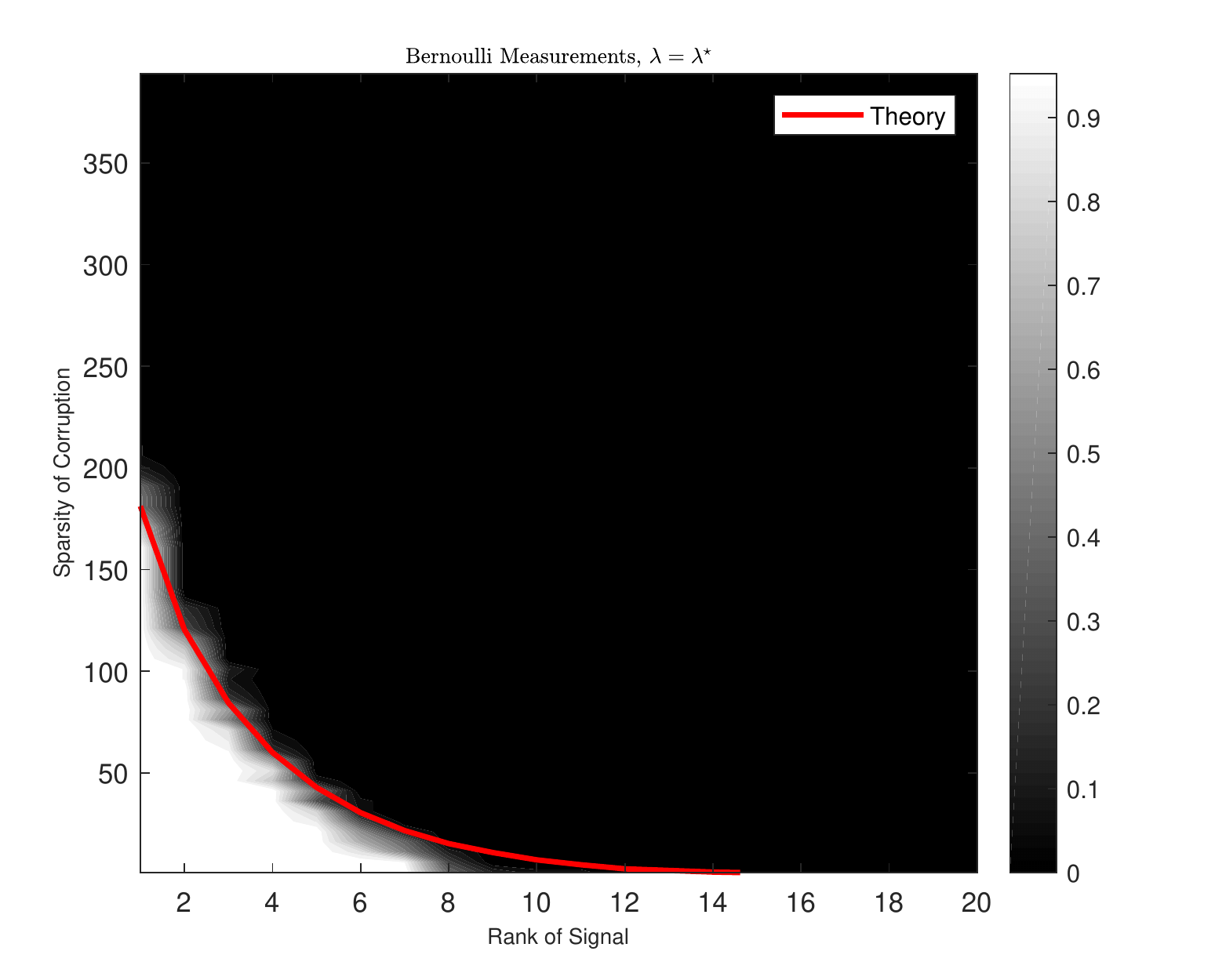}
	}
	\caption{Phase transitions of the penalized recovery procedure with the optimal $\lambda$ under Gaussian and Bernoulli measurements in both sparse signal recovery from sparse corruption and low-rank matrix recovery from sparse corruption. The red curves plot the phase transition thresholds predicted by that of the constrained procedures \eqref{threshold}.}
	\label{fig:optimal_lbd}
\end{figure*}

\subsection{Optimal Choice of the Tradeoff Parameter $\lambda$}
In this section, we explore the empirical phase transition of the penalized procedure with the optimal tradeoff parameter $\lambda$. Similarly, we consider these two typical examples.
\subsubsection{Sparse Signal Recovery from Sparse Corruption}
The experiment settings are the same as the penalized case except that we solve the following penalized procedure in step (4):
\begin{align*}
(\hat{\vx}, \hat{\vv})=\arg\min_{\vx, \vv} ~\|\vx\|_1+\lambda\|\vv\|_1,\quad\text{s.t.~}&\vy=\mPhi\vx+\sqrt{m}\vv
\end{align*}
with the optimal parameter selection strategy $\lambda=\lambda^\star$ as in \eqref{optimal_lambda}.

\subsubsection{Low-rank Matrix Recovery from Sparse Corruption}
We carry out similar experiments as the penalized case except that we solve the following penalized procedure in step (4):
\begin{align*}
(\hat{\mX}, \hat{\vv})=\arg\min_{\mX, \vv} ~\|\mX\|_*+\lambda\|\vv\|_1,\quad\text{s.t.~}&\vy=\mPhi\cdot\textrm{vec}(\mX)+\sqrt{m}\vv.
\end{align*}
The tradeoff parameter is set to $\lambda=\lambda^\star$ according to \eqref{optimal_lambda}.

To compare our theory with the empirical results, we overlay the theoretical threshold \eqref{threshold} of the constrained procedures. Fig.\ref{fig:optimal_lbd} shows the empirical probability of success for the penalized procedure with the optimal parameter $\lambda$ in both examples.  We can find that the theoretical threshold of the constrained problems predicts the empirical phase transition of penalized problems perfectly under both Gaussian and Bernoulli measurements, which indicates that our strategy to choose $\lambda$ is optimal in the sense of the Lagrange theory.

\section{Conclusion and Future Directions}\label{Conclusion}

This paper has developed a unified framework to establish the phase transition theory for both constrained and penalized recovery procedures which are used to solve corrupted sensing problems under different scenarios. The analysis is only based on some well-known results in Gaussian process theory. Our theoretical results have shown that the phase transitions of these two recovery procedures are determined by some geometric measures, e.g., the spherical Gaussian width of a tangent cone, the Gaussian (squared) distance to a scaled subdifferential. We have also explored the relationship between these two procedures from a quantitative perspective, which in turn indicates how to pick the optimal tradeoff parameter in the penalized recovery procedure. The numerical experiments have demonstrated a close agreement between our theoretical results and the empirical phase transitions. For future work, we enlist two promising directions:
	\begin{itemize}
		\item Universality: Under Gaussian measurements, our results provide a thorough explanation for the phase transition phenomenon of corrupted sensing. The Gaussian assumption is critical in the establishment of our main results. However, extensive numerical examples in Section \ref{simulations} have suggested that the phase transition results of corrupted sensing are universal. Thus, an important question is to establish the phase transition theory for corrupted sensing beyond Gaussian measurements.
		\item Noisy phase transition: Throughout the paper, we analyze the phase transition of corrupted sensing in the noiseless setting. It might be interesting to consider the noisy measurements $\vy=\mPhi\vx^{\star}+\vv^{\star}+\vz$, and to provide precise error analysis for different convex recovery procedures. 
			In \cite{donoho2011noise}, Donoho \emph{et al.} have considered the noisy compressed sensing problem $\vy=\mPhi\vx^\star+\vz$ with $\vz\sim\NN(0,\sigma\mI_m)$, and use the penalized $\ell_1$-minimization $\hat{\vx}=\arg\min\{\|\vy-\mPhi\vx\|_2^2/2+\lambda\|\vx\|_1\}$ to recover the original signal. They have shown that the normalized MSE $\frac{\E\|\hat{\vx}-\vx^\star\|_2^2}{\sigma^2}$ is bounded throughout an asymptotic region and is unbounded throughout the complementary region. The phase boundary of the interested region is identical to the previously known phase transition for the noiseless problem. We may expect a non-asymptotic characterization of the normalized MSE for the noisy corrupted sensing problem, which implies a new perspective for the phase transition results in noiseless case.
	\end{itemize}

\appendices
\section{Proofs of Lemma \ref{lem: success and failure_con} and  Theorem \ref{them: Phase transition_con}}
In this appendix, we present a detailed proof for the phase transition result of the constrained recovery procedures. For brevity, we denote $\TT_f(\vx^{\star})$ and $\TT_g(\vv^{\star})$ by $\TT_f$ and $\TT_g$ respectively. Some auxiliary lemma and facts used in the proofs are included in Appendix \ref{auxiliaryresults}.

\subsection{Proof of Lemma \ref{lem: success and failure_con}}
\begin{proof}
It follows from the optimization condition for linear inverse problems \cite[Section 4]{rudelson2008sparse} or \cite[Proposition 2.1]{chandrasekaran2012convex} that $(\hat{\vx},\hat{\vv})=(\vx^{\star},\vv^{\star})$ is the unique optimal solution of \eqref{constrained_procedure1} or \eqref{constrained_procedure2} if and only if $\text{null}\left([\mPhi, \sqrt{m}\mI_m]\right)\cap\left(\TT_f\times\TT_g\right)=\{(\vzero,\vzero)\}$, which is equivalent to $\text{null}\left([\mPhi, \sqrt{m}\mI_m]\right)\cap\left((\TT_f\times\TT_g)\cap\S^{n+m-1}\right)=\emptyset$. Therefore, if $\text{null}\left([\mPhi, \sqrt{m}\mI_m]\right)\cap\left((\TT_f\times\TT_g)\cap\S^{n+m-1}\right)=\emptyset$, i.e.,
	\begin{align*}
	\min_{(\va,\vb)\in \left(\TT_f\times\TT_g\right)\cap\S^{n+m-1}} \|\mPhi\va+\sqrt{m}\vb\|_2>0,
	\end{align*}
	then the constrained procedures \eqref{constrained_procedure1} and \eqref{constrained_procedure2} succeed. If $\text{null}\left([\mPhi, \sqrt{m}\mI_m]\right)\cap\left((\TT_f\times\TT_g)\cap\S^{n+m-1}\right)\neq\emptyset$, i.e.,
	\begin{align}\label{Failcondition}
	\min_{(\va,\vb)\in  \left(\TT_f\times\TT_g\right)\cap\S^{n+m-1}} \|\mPhi\va+\sqrt{m}\vb\|_2=0,
	\end{align}
	then the constrained procedures \eqref{constrained_procedure1} and \eqref{constrained_procedure2} fail.

  Obviously, \eqref{Failcondition} holds if $\vzero \in \mA((\TT_f\times\TT_g)\cap\S^{n+m-1})$. Since $f$ and $g$ are proper convex functions, then $\TT_f$ and $\TT_g$ are convex, and hence $(\TT_f\times\TT_g)\cap\S^{n+m-1}$ is spherically convex. By assumption, $\TT_f$ and $\TT_g$ are nonempty and closed, the desired sufficient condition \eqref{polar_f_con} follows by  directly applying the polarity principle (Fact \ref{polarity argument}).

\end{proof}

\subsection{Proof of Theorem \ref{them: Phase transition_con}}
\begin{proof}
\textbf{Success case:}
	Lemma \ref{lem: success and failure_con} indicates that the constrained procedures \eqref{constrained_procedure1} and \eqref{constrained_procedure2} succeed if
	\begin{align*}
	\min_{(\va,\vb)\in \left(\TT_f\times\TT_g\right)\cap\S^{n+m-1}} \|\mPhi\va+\sqrt{m}\vb\|_2>0.
	\end{align*}
    Our goal then reduces to show that if the number of measurements satisfies \eqref{NumberofMeasurements_c1}, then the above inequality holds with high probability. For clarity, the proof is divided into three steps.
	
\textbf{Step 1: Problem reduction.} We first apply Gordon's Lemma to convert the probability of the targeted event to a surrogate which is convenient to handle. Observe that
	\begin{equation}\label{bound of success cond_con}
	\min_{(\va,\vb)\in \left(\TT_f\times\TT_g\right)\cap\S^{n+m-1}} \|\mPhi\va+\sqrt{m}\vb\|_2
	=\min_{(\va,\vb)\in \left(\TT_f\times\TT_g\right)\cap\S^{n+m-1}}\max_{\vu\in\S^{m-1}}\ip{\mPhi\va}{\vu}+\ip{\sqrt{m}\vb}{\vu}.
	\end{equation}
	For any $(\va,\vb) \in \left(\TT_f\times\TT_g\right)\cap\S^{n+m-1}$ and $\vu\in \S^{m-1}$, define the following two Gaussian processes
	$$
	X_{(\va,\vb),\vu}:=\ip{\mPhi\va}{\vu}+ \|\va\|_2\cdot g~~~
	$$
	and
	$$
	Y_{(\va,\vb),\vu}:=\|\va\|_2\ip{\vh}{\vu}+\ip{\vg}{\va},
	$$
	where $g \sim \NN(0, 1)$, $\vh \sim \NN(\vzero, \mI_m)$, and $\vg \sim \NN(\vzero, \mI_n)$ are independent of each other. It is not hard to check that the above Gaussian processes satisfy the conditions of Gordon's Lemma, i.e.,
	\begin{align*}
	\E X_{(\va,\vb),\vu}^2 &= 2\|\va\|_2^2= \E Y_{(\va,\vb),\vu}^2,\\
	\E [X_{(\va,\vb),\vu} X_{(\va',\vb'),\vu'}]-
	\E [Y_{(\va,\vb),\vu} Y_{(\va',\vb'),\vu'}]	&= \ip{\vu}{\vu'}\ip{\va}{\va'}+\|\va\|_2\|\va'\|_2-\ip{\va}{\va'}-\|\va\|_2\|\va'\|_2\ip{\vu}{\vu'}\\
	&=\left(1-\ip{\vu}{\vu'}\right)\left(\|\va\|_2\|\va'\|_2-\ip{\va}{\va'}\right)\\
	&\geq 0,
	\end{align*}
	where, in the last line, the equality holds when $\va=\va'$. It then follows from Gordon's Lemma (Fact \ref{Grodon lem}) that (by setting $\tau_{(\va,\vb),\vu}=-\ip{\sqrt{m}\vb}{\vu}+0_+$)
	\begin{align*}
	&\Pr{\min_{(\va,\vb)\in \atop \left(\TT_f\times\TT_g\right)\cap\S^{n+m-1}}\max_{\vu\in\S^{m-1}}Y_{(\va,\vb),\vu}\geq \tau_{(\va,\vb),\vu}}=\Pr{\min_{(\va,\vb)\in \atop \left(\TT_f\times\TT_g\right)\cap\S^{n+m-1}}\max_{\vu\in\S^{m-1}}\|\va\|_2\ip{\vh}{\vu}+\ip{\vg}{\va}+\ip{\sqrt{m}\vb}{\vu}> 0}\\
	&\hspace{150pt}\leq\Pr{\min_{(\va,\vb)\in \atop \left(\TT_f\times\TT_g\right)\cap\S^{n+m-1}}\max_{\vu\in\S^{m-1}}X_{(\va,\vb),\vu}\geq \tau_{(\va,\vb),\vu}} \\
	&\hspace{150pt}= \Pr{\min_{(\va,\vb)\in \atop \left(\TT_f\times\TT_g\right)\cap\S^{n+m-1}}\max_{\vu\in\S^{m-1}}\ip{\mPhi\va}{\vu}+\|\va\|_2\cdot g+\ip{\sqrt{m}\vb}{\vu}> 0}\\
	&\hspace{150pt}\leq \frac{1}{2}+\frac{1}{2}\Pr{\min_{(\va,\vb)\in\atop\left(\TT_f\times\TT_g\right)\cap\S^{n+m-1}}\max_{\vu\in\S^{m-1}}\ip{\mPhi\va}{\vu}+\|\va\|_2\cdot g+\ip{\sqrt{m}\vb}{\vu}> 0~\bigg|~g\leq 0}\\
	&\hspace{150pt}\leq \frac{1}{2}+\frac{1}{2}\Pr{\min_{(\va,\vb)\in\atop\left(\TT_f\times\TT_g\right)\cap\S^{n+m-1}}\max_{\vu\in\S^{m-1}}\ip{\mPhi\va}{\vu}+\ip{\sqrt{m}\vb}{\vu}> 0},
	\end{align*}
    where the second inequality is due to the law of total probability and the third inequality holds by noting $-\|\va\|_2\cdot g \geq 0$ when $g \leq 0$.
    Rearranging the above inequality leads to
	\begin{align}\label{p>2q-1}
	&\Pr{\min_{(\va,\vb)\in \atop \left(\TT_f\times\TT_g\right)\cap\S^{n+m-1}}\max_{\vu\in\S^{m-1}}\ip{\mPhi\va}{\vu}+\ip{\sqrt{m}\vb}{\vu}> 0} \notag \\
	&\hspace{135pt}\geq 2\Pr{\underbrace{\min_{(\va,\vb)\in \atop \left(\TT_f\times\TT_g\right)\cap\S^{n+m-1}}\max_{\vu\in\S^{m-1}}\|\va\|_2\ip{\vh}{\vu}+\ip{\vg}{\va}+\ip{\sqrt{m}\vb}{\vu}}_{:=\mathscr{E}_1}> 0}-1.
	\end{align}
	Moreover, $\mathscr{E}_1$ can be rewritten as
	\begin{equation}\label{bound of U1}
	\begin{split}
	\mathscr{E}_1&=\min_{(\va,\vb)\in\atop \left(\TT_f\times\TT_g\right)\cap\S^{n+m-1}}\max_{\vu\in\S^{m-1}}\ip{\vu}{\|\va\|_2\vh+\sqrt{m}\vb}+\ip{\vg}{\va}\\
	&=\min_{(\va,\vb)\in\atop \left(\TT_f\times\TT_g\right)\cap\S^{n+m-1}}\bigg\|\|\va\|_2\vh+\sqrt{m}\vb\bigg\|_2+\ip{\vg}{\va}\\
	&=\min_{t\in[0,1]}\min_{\va'\in\TT_f\cap\S^{n-1} \atop \vb'\in\TT_g\cap\S^{m-1}}\bigg\|t\vh+\sqrt{m(1-t^2)}\vb'\bigg\|_2+t\ip{\vg}{\va'}.
	\end{split}
	\end{equation}
	In the last line, we have let $\|\va\|_2=t,~\|\vb\|_2=\sqrt{1-t^2}$, $\va'={\va}/{\|\va\|_2}, \textrm{and} ~\vb'={\vb}/{\|\vb\|_2}$.

    Define
    \begin{equation*}
	\begin{split}
	U(\vg,\vh,t):=\min_{\va'\in\TT_f\cap\S^{n-1} \atop \vb'\in\TT_g\cap\S^{m-1}}\bigg\|\vh+\sqrt{m(\frac{1}{t^2}-1)}\cdot\vb'\bigg\|_2+\ip{\vg}{\va'}.
	\end{split}
	\end{equation*}
    Let
    \begin{align*}
	t_1\in\arg\min_{t\in[0,1]} \left\{  \min_{\va'\in\TT_f\cap\S^{n-1} \atop \vb'\in\TT_g\cap\S^{m-1}}\bigg\|t\vh+\sqrt{m(1-t^2)}\vb'\bigg\|_2+t\ip{\vg}{\va'} \right\}.
	\end{align*}
    If $t_1 \neq 0$, then we have
	\begin{align*} 
	\Pr{\min_{t\in[0,1]}\min_{\va'\in\TT_f\cap\S^{n-1} \atop \vb'\in\TT_g\cap\S^{m-1}}\bigg\|t\vh+\sqrt{m(1-t^2)}\vb'\bigg\|_2+t\ip{\vg}{\va'} > 0} &=\Pr{t_1\cdot U(\vg,\vh,t_1)>0} \notag\\
	& = \Pr{U(\vg,\vh,t_1)>0} \notag\\
    &\geq \Pr{\min_{t\in(0,1]} U(\vg,\vh,t)>0}.
	\end{align*}
    If $t_1 = 0$, then equation $\mathscr{E}_1=\sqrt{m}$, which implies
	$$
	\Pr{\mathscr{E}_1> 0}=1.
	$$
	Thus we have
	\begin{align} \label{bound of q1}
	& \Pr{\mathscr{E}_1> 0} \geq \Pr{\min_{t\in(0,1]} U(\vg,\vh,t)>0}.
	\end{align}

    Combining \eqref{bound of success cond_con}, \eqref{p>2q-1}, and \eqref{bound of q1} yields
    \begin{equation}\label{ResultofStep2}
	\Pr{\min_{(\va,\vb)\in \atop \left(\TT_f\times\TT_g\right)\cap\S^{n+m-1}} \|\mPhi\va+\sqrt{m}\vb\|_2>0}\geq 2\Pr{\min_{t\in(0,1]}U(\vg,\vh,t) > 0}-1.
	\end{equation}
    Therefore, it is sufficient to establish the lower bound for $\Pr{\min_{t\in(0,1]}U(\vg,\vh,t) > 0}$.

    \textbf{Step 2: Establish the lower bound for $\Pr{\min_{t\in(0,1]}U(\vg,\vh,t) > 0}$.} We then apply the Gaussian concentration inequality to establish the lower bound for $\Pr{\min_{t\in(0,1]}U(\vg,\vh,t) > 0}$.

    Note that $\min_{t\in(0,1]}U(\vg,\vh,t)$ can reformulated as
    \begin{equation*}
	\begin{split}
	\min_{t\in(0,1]}U(\vg,\vh,t)&=\min_{t\in(0,1]}\min_{\va'\in\TT_f\cap\S^{n-1} \atop \vb'\in\TT_g\cap\S^{m-1}}\bigg\|\vh+\sqrt{m(\frac{1}{t^2}-1)}\cdot\vb'\bigg\|_2+\ip{\vg}{\va'}\\
	&=\min_{\va'\in\TT_f\cap\S^{n-1}} \min_{t'\geq 0 \atop \vb'\in\TT_g\cap\S^{m-1}}\|\vh+t'\vb'\|_2+\ip{\vg}{\va'}\\
    &=\min_{\vx\in\TT_g\atop \va'\in\TT_f\cap\S^{n-1}}\|\vh+\vx\|_2+\ip{\vg}{\va'},
	\end{split}
	\end{equation*}
    where, in the second line, we have let $t'=\sqrt{m(\frac{1}{t^2}-1)}\geq 0$. It then follows from Lemma \ref{Lipschitz function} that the function $\min_{t\in(0,1]}U(\vg,\vh,t)$ is a $\sqrt{2}$-Lipschitz function. To apply the Gaussian concentration inequality, it suffices to bound the expectation of $\min_{t\in(0,1]}U(\vg,\vh,t)$. To this end,
	\begin{equation}\label{bound of U2}
	\begin{split}
	\E \min_{t\in(0,1]}U(\vg,\vh,t)&=\E \left\{\min_{\vx\in\TT_g}\|\vh+\vx\|_2-\max_{\va'\in\TT_f\cap\S^{n-1}}\ip{-\vg}{\va'}\right\}\\
    & =\E\left\{\dist(-\vh,\TT_g)-\max_{\va'\in\TT_f\cap\S^{n-1}}\ip{-\vg}{\va'}\right\}\\
    & =\E\dist(\vh,\TT_g)-\omega(\TT_f\cap\S^{n-1})\\
	&\geq\sqrt{\E\dist^2(\vh,\TT_g)-1}-\omega(\TT_f\cap\S^{n-1})\\
	&=\sqrt{m-\E\dist^2(\vh,\TT_g^\circ)-1}-\omega(\TT_f\cap\S^{n-1})\\
	&=\sqrt{m-\E\bigg(\max_{\vx\in\TT_g\cap\B_2^n}\ip{\vh}{\vx}\bigg)^2-1}-\omega(\TT_f\cap\S^{n-1})\\
	&\geq\sqrt{m-\omega^2(\TT_g\cap\S^{m-1})-2}-\omega(\TT_f\cap\S^{n-1})\\
	&\geq\sqrt{m}-\sqrt{\omega^2(\TT_g\cap\S^{m-1})+\omega^2(\TT_f\cap\S^{n-1})+2}\\
	&\geq\epsilon.
	\end{split}
	\end{equation}
    The third line is due to the fact that $-\vh$ (or $-\vg$) and $\vh$ (or $\vg$) share the same distribution.
	The first inequality has used Fact \ref{Variance of Lipschitz}, i.e.,$ \textrm{Var}(\dist(\vh,\TT_g))=\E\dist^2(\vh,\TT_g)-\E^2\dist(\vh,\TT_g)\leq 1.$
	The fifth line holds because of Moreau’s decomposition theorem (Fact \ref{Moreau}). The next two lines follows from Facts \ref{omega & dist} and \ref{relation_w_d}, respectively.  The last two inequalities are due to the measurement condition \eqref{NumberofMeasurements_c1}.

    Now using the Gaussian concentration inequality (Fact \ref{gaussian concentration}) yields
	\begin{align*}
	&\Pr{ \min_{t\in(0,1]}U(\vg,\vh,t) - \E \min_{t\in(0,1]}U(\vg,\vh,t) \leq -\epsilon} \leq \exp\left(\frac{-\epsilon^2}{4}\right),
	\end{align*}
	which in turn implies that
	\begin{equation}\label{pr of U2>0}
	\begin{split}
	\Pr{ \min_{t\in(0,1]}U(\vg,\vh,t) > 0}
	&\geq\Pr{ \min_{t\in(0,1]}U(\vg,\vh,t) >\E \min_{t\in(0,1]}U(\vg,\vh,t)-\epsilon}\\
	&\geq 1-\exp\left(\frac{-\epsilon^2}{4}\right).
	\end{split}
	\end{equation}

\textbf{Step 3: Complete the proof.}

    Combining the results in Steps 1 and 2 (\eqref{ResultofStep2} and \eqref{pr of U2>0}), we have
	\begin{equation*}
	\begin{split}
	\Pr{\min_{(\va,\vb)\in \atop \left(\TT_f\times\TT_g\right)\cap\S^{n+m-1}} \|\mPhi\va+\sqrt{m}\vb\|_2>0}&\geq 2\Pr{\min_{t\in(0,1]}U(\vg,\vh,t) > 0}-1\\
	&\geq 1-2\exp\left(\frac{-\epsilon^2}{4}\right).
	\end{split}
	\end{equation*}
   Therefore, we have established that when $\sqrt{m} \geq\sqrt{\omega^2\left(\TT_f\cap\S^{n-1}\right)+\omega^2\left(\TT_g\cap\S^{m-1}\right)}+\sqrt{2} +\epsilon$, the constrained procedures \eqref{constrained_procedure1} and \eqref{constrained_procedure2} succeed with probability at least $1-2\exp\left(-\epsilon^2/4\right)$.

	\textbf{Failure case:} According to Lemma \ref{lem: success and failure_con}, the constrained procedures \eqref{constrained_procedure1} and \eqref{constrained_procedure2} fail if
	$$
	\min_{\vr\in\S^{m-1}} \min_{ \vs \in (\TT_f\times \TT_g)^{\circ}} \| \vs - \mA^T\vr \|_2 > 0,
	$$
	where $\mA=[\mPhi,\sqrt{m}\mI_m]$. So it suffices to show that if the number of measurements satisfies \eqref{NumberofMeasurements_c2}, then the above inequality holds with high
probability. For clarity, the proof is similarly divided into three steps.

\textbf{Step 1: Problem reduction.} In this step, we also use Gordon's Lemma to convert the probability of the targeted event to another one which is easy to handle.

    Let $\vs = [\vs_1^T, \vs_2^T]^T$ and $\vu = [\vu_1^T, \vu_2^T]^T$. Note first that for any $\vr\in\S^{m-1}$, we have
	\begin{equation*}
	\begin{split}
	\min_{\vs \in (\TT_f\times \TT_g)^{\circ}} \|\vs-\mA^T\vr\|_2&=\min_{\vs \in (\TT_f\times \TT_g)^{\circ}}\left\|\begin{bmatrix} \vs_1-\mPhi^T\vr \\ \vs_2-\sqrt{m}\vr \end{bmatrix}\right\|_2\\
	&=\min_{\vs \in (\TT_f\times \TT_g)^{\circ}}\max_{\vu\in\S^{n+m-1}}\ip{\mPhi^T\vr-\vs_1}{\vu_1}+\ip{\sqrt{m}\vr-\vs_2}{\vu_2}\\
	&=\min_{\vs \in (\TT_f\times \TT_g)^{\circ}}\max_{\vu\in\S^{n+m-1}}\ip{\mPhi^T\vr}{\vu_1}+\ip{\sqrt{m}\vr}{\vu_2}-\ip{\vs}{\vu}\\
	&\geq\max_{\vu\in\S^{n+m-1}}\left[\ip{\mPhi^T\vr}{\vu_1}+\ip{\sqrt{m}\vr}{\vu_2}-\max_{\vs \in (\TT_f\times \TT_g)^{\circ}}\ip{\vs}{\vu}\right]\\
	&=\max_{(\vu_1,\vu_2)\in\atop(\TT_f\times \TT_g)\cap\S^{n+m-1}}\ip{\mPhi^T\vr}{\vu_1}+\ip{\sqrt{m}\vr}{\vu_2}.
	\end{split}
	\end{equation*}
	The inequality is due to the max-min inequality. The last line has used the fact that $\max_{\vs \in (\TT_f\times \TT_g)^{\circ}} \ip{\vs}{\vu} = 0$ when $\vu \in \TT_f\times \TT_g$, otherwise it equals $\infty$. Thus we obtain
	\begin{align}\label{bound for min min norms}
	\min_{\vr\in\S^{m-1}} \min_{ \vs \in (\TT_f\times \TT_g)^{\circ}} \| \vs - \mA^T\vr \|_2 	\ge \min_{\vr \in \S^{m-1}} \max_{(\vu_1,\vu_2)\in\atop(\TT_f\times \TT_g)\cap\S^{n+m-1}}\ip{\mPhi^T\vr}{\vu_1}+\ip{\sqrt{m}\vr}{\vu_2}.
	\end{align}
    We then use Gordon's Lemma to bound the probability of the targeted event from below. To this end, for any $\vr \in \S^{m-1}$ and $(\vu_1,\vu_2) \in (\TT_f\times \TT_g)\cap\S^{n+m-1}$, define the following two Gaussian processes
	$$
	X_{\vr,(\vu_1,\vu_2)}:=\ip{\mPhi^T\vr}{\vu_1}+\|\vu_1\|_2\cdot g~
	$$
	and
	$$
	Y_{\vr,(\vu_1,\vu_2)}:=\ip{\vg}{\vu_1}+\|\vu_1\|_2\ip{\vh}{\vr},
	$$
	where $g\sim \NN(0,1)$, $\vg \sim \NN(\vzero, \mI_n)$, and $\vh \sim \NN(\vzero, \mI_m)$ are independent of each other. It can be easily checked that these two Gaussian processes satisfy the conditions in Gordon's Lemma:
	\begin{align*}
	\E X_{\vr,(\vu_1,\vu_2)}^2 & = 2\| \vu_1\|_2^2 = \E Y_{\vr,(\vu_1,\vu_2)}^2, \\
	\E [X_{\vr,(\vu_1,\vu_2)} X_{\vr',(\vu_1',\vu_2')}]-
	\E [Y_{\vr,(\vu_1,\vu_2)} Y_{\vr',(\vu_1',\vu_2')}]&= \ip{\vr}{\vr'}\ip{\vu_1}{\vu_1'}+\|\vu_1\|_2\|\vu_1'\|_2-\ip{\vu_1}{\vu_1'}-\|\vu_1\|_2\|\vu_1'\|_2\ip{\vr}{\vr'}\\
	&=\left(1-\ip{\vr}{\vr'}\right)\left(\|\vu_1\|_2\|\vu_1'\|_2-\ip{\vu_1}{\vu_1'}\right)\\
	&\geq 0.
	\end{align*}
	Here, in the last line, the equality holds when $\vr=\vr'$. It follows from Gordon's Lemma (Fact \ref{Grodon lem}) that (by setting $\tau_{\vr,(\vu_1,\vu_2)}=-\ip{\sqrt{m}\vr}{\vu_2}+0_+$)
	\begin{align*}
	&\Pr{\min_{\vr \in \S^{m-1}} \max_{(\vu_1,\vu_2)\in\atop(\TT_f\times \TT_g)\cap\S^{n+m-1}}Y_{\vr,(\vu_1,\vu_2)}\geq \tau_{\vr,(\vu_1,\vu_2)}}=\Pr{\min_{\vr \in \S^{m-1}} \max_{(\vu_1,\vu_2)\in\atop(\TT_f\times \TT_g)\cap\S^{n+m-1}}\ip{\vg}{\vu_1}+\|\vu_1\|_2\ip{\vh}{\vr}+\ip{\sqrt{m}\vr}{\vu_2}> 0}\\
	&\hspace{150pt}\leq \Pr{\min_{\vr \in \S^{m-1}} \max_{(\vu_1,\vu_2)\in\atop(\TT_f\times \TT_g)\cap\S^{n+m-1}}X_{\vr,(\vu_1,\vu_2)}\geq \tau_{\vr,(\vu_1,\vu_2)}}\\
	&\hspace{150pt}= \Pr{\min_{\vr \in \S^{m-1}} \max_{(\vu_1,\vu_2)\in\atop(\TT_f\times \TT_g)\cap\S^{n+m-1}}\ip{\mPhi^T\vr}{\vu_1}+\|\vu_1\|_2\cdot g+\ip{\sqrt{m}\vr}{\vu_2}> 0}\\
	&\hspace{150pt}\leq  \frac{1}{2}+\frac{1}{2}\Pr{\min_{\vr \in \S^{m-1}} \max_{(\vu_1,\vu_2)\in\atop(\TT_f\times \TT_g)\cap\S^{n+m-1}}\ip{\mPhi^T\vr}{\vu_1}+\|\vu_1\|_2\cdot g+\ip{\sqrt{m}\vr}{\vu_2}> 0 \Big| g \leq 0 }\\
    &\hspace{150pt}\leq  \frac{1}{2}+\frac{1}{2}\Pr{\min_{\vr \in \S^{m-1}} \max_{(\vu_1,\vu_2)\in\atop(\TT_f\times \TT_g)\cap\S^{n+m-1}}\ip{\mPhi^T\vr}{\vu_1}+\ip{\sqrt{m}\vr}{\vu_2}> 0 },
    \end{align*}
	which implies
	\begin{align}\label{p0>2q-1}
	&\Pr{\min_{\vr \in \S^{m-1}} \max_{(\vu_1,\vu_2)\in\atop(\TT_f\times \TT_g)\cap\S^{n+m-1}}\ip{\mPhi^T\vr}{\vu_1}+\ip{\sqrt{m}\vr}{\vu_2}> 0} \notag \\
	&\hspace{135pt}\geq 2\Pr{\underbrace{\min_{\vr \in \S^{m-1}} \max_{(\vu_1,\vu_2)\in\atop(\TT_f\times \TT_g)\cap\S^{n+m-1}}\ip{\vg}{\vu_1}+\|\vu_1\|_2\ip{\vh}{\vr}+\ip{\sqrt{m}\vr}{\vu_2}}_{:= \mathscr{E}_2}> 0}-1.
	\end{align}
	Moreover, $\mathscr{E}_2$ can be bounded
	\begin{equation*}
	\begin{split}
	\mathscr{E}_2	&=\min_{\vr\in\S^{m-1}}\max_{(\vu_1,\vu_2)\in\atop \left(\TT_f\times\TT_g\right)\cap\S^{n+m-1}}\ip{\vg}{\vu_1}+\ip{\vr}{\|\vu_1\|_2\vh+\sqrt{m}\vu_2}\\
	&\geq\max_{(\vu_1,\vu_2)\in\atop \left(\TT_f\times\TT_g\right)\cap\S^{n+m-1}}\min_{\vr\in\S^{m-1}}\ip{\vg}{\vu_1}+\ip{\vr}{\|\vu_1\|_2\vh+\sqrt{m}\vu_2}\\
	&=\max_{(\vu_1,\vu_2)\in\atop \left(\TT_f\times\TT_g\right)\cap\S^{n+m-1}}\ip{\vg}{\vu_1}-\bigg\|\|\vu_1\|_2\vh+\sqrt{m}\vu_2\bigg\|_2\\
	&=\max_{t\in[0,1]}\max_{\vu_1'\in\TT_f\cap\S^{n-1} \atop \vu_2'\in\TT_g\cap\S^{m-1}}t\ip{\vg}{\vu_1'}-\bigg\|t\vh+\sqrt{m(1-t^2)}\vu_2'\bigg\|_2.
	\end{split}
	\end{equation*}
	The inequality holds because of the  max-min inequality. In the last line, we have let $\|\vu_1\|_2=t,~\|\vu_2\|_2=\sqrt{1-t^2}$, $\vu_1'={\vu_1}/{\|\vu_1\|_2},$ and $\vu_2'={\vu_2}/{\|\vu_2\|_2}$.

    Define
	\begin{equation*}
	\begin{split}
	W(\vg,\vh,t):=\max_{\vu_1'\in\TT_f\cap\S^{n-1} \atop \vu_2'\in\TT_g\cap\S^{m-1}}\ip{\vg}{\vu_1'}-\bigg\|\vh+\sqrt{m(\frac{1}{t^2}-1)}\cdot\vu_2'\bigg\|_2.
	\end{split}
	\end{equation*}
    Let \footnote{The effective domain of an extended real-valued function $k(x):X\rightarrow \bar{\R}$ is defined as $\{x\in X | k(x) \in \R \cup \{-\infty\}\}$.}
	\begin{align*}
	t_2\in\arg\max_{t \in [0,1]} W(\vg,\vh,t).
	\end{align*}
    Clearly, $t_2 \neq 0$, since $W(\vg,\vh,t)\to -\infty$ as $t\to 0_{+}$.  Then we have
	\begin{align} \label{bound of q2}
	\Pr{\max_{t\in[0,1]}\max_{\vu_1'\in\TT_f\cap\S^{n-1} \atop \vu_2'\in\TT_g\cap\S^{m-1}}t\ip{\vg}{\vu_1'}-\bigg\|t\vh-\sqrt{m(1-t^2)}\vu_2'\bigg\|_2 > 0}
	&=\Pr{\max_{t\in [0,1]}t\cdot W(\vg,\vh,t)>0}  \notag \\
    & \geq \Pr{t_2\cdot W(\vg,\vh,t_2)>0}\notag \\
	&=\Pr{W(\vg,\vh,t_2)>0}\notag\\
    &=\Pr{\max_{t \in (0,1]} W(\vg,\vh,t)>0}.
	\end{align}

    Combining \eqref{bound for min min norms}, \eqref{p0>2q-1}, and \eqref{bound of q2}, we obtain
	\begin{equation}\label{ResultofStep2Fail}
	\begin{split}
	\Pr{\min_{\vr\in\S^{m-1}} \min_{ \vs \in (\TT_f\times \TT_g)^{\circ}} \| \vs - \mA^T\vr \|_2 > 0}&\geq 2\Pr{\max_{t\in(0,1]}W(\vg,\vh,t) > 0}-1.
	\end{split}
	\end{equation}
    Therefore, our goal reduces to establish the lower bound for $\Pr{\max_{t\in(0,1]}W(\vg,\vh,t) > 0}$.

    \textbf{Step 2: Establish the lower bound for $\Pr{\max_{t\in(0,1]}W(\vg,\vh,t) > 0}$.} In this step, we use the Gaussian concentration inequality to establish the lower bound for $\Pr{\max_{t\in(0,1]}W(\vg,\vh,t) > 0}$.

    Similar to the success case, $\max_{t\in(0,1]}W(\vg,\vh,t)$ can be rewritten as
	\begin{equation*}
	\begin{split}
	\max_{t\in(0,1]} W(\vg,\vh,t)&=\max_{t\in(0,1]}\max_{\vu_1'\in\TT_f\cap\S^{n-1} \atop \vu_2'\in\TT_g\cap\S^{m-1}}\ip{\vg}{\vu_1'}-\bigg\|\vh+\sqrt{m(\frac{1}{t^2}-1)}\cdot\vu_2'\bigg\|_2\\
	&=\max_{\vu_1'\in\TT_f\cap\S^{n-1}}\max_{t'\geq 0 \atop \vu_2'\in\TT_g\cap\S^{m-1}}\ip{\vg}{\vu_1'}-\bigg\|\vh+t'\cdot\vu_2'\bigg\|_2\\
    &=\max_{\vx\in\TT_g\atop \vu_1'\in\TT_f\cap\S^{n-1}}\ip{\vg}{\vu_1'}-\|\vh+\vx\|_2.\\
    \end{split}
	\end{equation*}
    It then follows from Lemma \ref{Lipschitz function} that $\max_{t\in(0,1]} W(\vg,\vh,t)$  is a $\sqrt{2}$-Lipschitz function. Moreover, its expectation can be bounded from below:
	\begin{equation}\label{bound of W2}
	\begin{split}
	\E \max_{t\in(0,1]} W(\vg,\vh,t) &= \E \left(\max_{\vx\in\TT_g\atop \vu_1'\in\TT_f\cap\S^{n-1}}\ip{\vg}{\vu_1'}-\|\vh+\vx\|_2 \right) \\
    &=\omega(\TT_f\cap\S^{n-1})-\E\dist(\vh,\TT_g)\\
	&\geq\omega(\TT_f\cap\S^{n-1})-\sqrt{\E\dist^2(\vh,\TT_g)}\\
	&=\omega(\TT_f\cap\S^{n-1})-\sqrt{m-\E\dist^2(\vh,\TT_g^\circ)}\\
	&=\omega(\TT_f\cap\S^{n-1})-\sqrt{m-\E\bigg(\max_{\vx\in\TT_g\cap\B_2^n}\ip{\vh}{\vx}\bigg)^2}\\
	&\geq\omega(\TT_f\cap\S^{n-1})-\sqrt{m-\omega^2(\TT_g\cap\S^{m-1})}\\
	&\geq\sqrt{\omega^2(\TT_g\cap\S^{m-1})+\omega^2(\TT_f\cap\S^{n-1})}-\sqrt{m}\\
	&\geq\epsilon.
	\end{split}
	\end{equation}
	The second line holds because $-\vh$ and $\vh$ have the same distribution. The first inequality is due to Jensen's inequality. The fourth line has used  Moreau’s decomposition theorem (Fact \ref{Moreau}). The next two lines follow from  Facts \ref{omega & dist} and  \ref{relation_w_d}, respectively. The last two lines holds because of the measurement condition \eqref{NumberofMeasurements_c2}.
	
	Now applying the Gaussian concentration inequality (Fact \ref{gaussian concentration}) yields
	\begin{align*}
	&\Pr{ \max_{t\in(0,1]}W(\vg,\vh,t) - \E \max_{t\in(0,1]}W(\vg,\vh,t) \leq -\epsilon} \leq \exp\left(\frac{-\epsilon^2}{4}\right),
	\end{align*}
	which in turn implies that
	\begin{equation}\label{pr of W2>0}
	\begin{split}
	\Pr{ \max_{t\in(0,1]}W(\vg,\vh,t) > 0}	&\geq\Pr{ \max_{t\in(0,1]}W(\vg,\vh,t) >\E \max_{t\in(0,1]}W(\vg,\vh,t)-\epsilon}\\
	&\geq 1-\exp\left(\frac{-\epsilon^2}{4}\right).
	\end{split}
	\end{equation}

    \textbf{Step 3: Complete the proof.} Putting \eqref{ResultofStep2Fail}  and \eqref{pr of W2>0} together, we have
	\begin{equation*}
	\begin{split}
	\Pr{\min_{\vr\in\S^{m-1}} \min_{ \vs \in (\TT_f\times \TT_g)^{\circ}} \| \vs - \mA^T\vr \|_2 > 0}&\geq 2\Pr{\max_{t\in(0,1]}W(\vg,\vh,t) > 0}-1\\
	&\geq 1-2\exp\left(\frac{-\epsilon^2}{4}\right).
	\end{split}
	\end{equation*}
	
	Thus we have established that when $\sqrt{m} \leq\sqrt{\omega^2\left(\TT_f\cap\S^{n-1}\right)+\omega^2\left(\TT_g\cap\S^{m-1}\right)} -\epsilon$, the constrained procedures \eqref{constrained_procedure1} and \eqref{constrained_procedure2} fail with probability at least $1-2\exp\left(-\epsilon^2/4\right)$. This completes the proof.
	
\end{proof}

\section{Proofs of Lemma \ref{lem: success and failure} and Theorem \ref{them: Phase transition}}
In this appendix, we prove the phase transition result of the penalized recovery procedure. Some auxiliary lemma and facts used in the proofs are included in Appendix \ref{auxiliaryresults}.

\subsection{Proof of Lemma \ref{lem: success and failure}}
\begin{proof}
	The penalized recovery procedure \eqref{penilized_procedure} can be reformulated as the following unconstrained form
	\begin{align*}
	\min_{\vx} ~f(\vx)+\lambda\cdot g\left(\frac{1}{\sqrt{m}}(\vy-\mPhi\vx)\right).
	\end{align*}
	Define $F(\vx)= f(\vx) + \lambda\cdot g\left(\frac{1}{\sqrt{m}}(\vy-\mPhi\vx)\right)$. Clearly, $F(\vx)$ is a proper convex function. It follows from \cite[Theorems 23.8 and 23.9]{rockafellar1970convex} that the subdifferential of $F$ at $\vx^\star$ is given by
	\begin{align*}
	\partial F(\vx^\star)=\partial f(\vx^\star)-\frac{\lambda}{\sqrt{m}}\mPhi^T\cdot\partial g(\vv^\star).
	\end{align*}
	Moreover, $F(\vx)$ attains its minimum at $\vx^\star$ if and only if $\vzero\in\partial F(\vx^\star)$ \cite[Theorems 27.1]{rockafellar1970convex}. Therefore, if
	\begin{align}\label{sucesscondition}
	\vzero\in \mPhi^T\cdot\partial g(\vv^\star)-\frac{\sqrt{m}}{\lambda}\partial f(\vx^\star),
	\end{align}
	then the penalized problem \eqref{penilized_procedure} succeeds. If $\vzero\notin\partial F(\vx^\star)$, i.e.,
	\begin{align*}
	\min_{\va\in\partial f(\vx^\star) , \vb\in\partial g(\vv^\star)} \|\mPhi^T\vb-\frac{\sqrt{m}}{\lambda}\va\|_2>0.
	\end{align*}
	then the penalized problem \eqref{penilized_procedure} fails.
	
	Clearly, \eqref{sucesscondition} holds if $\vzero \in \mM(\TT_J\cap\S^{n+m-1})$. Since $f$ and $g$ are proper convex functions, then $\partial f(\vx^\star)$ and $\partial g(\vv^\star)$ are nonempty, closed convex sets, and hence $\TT_J\cap\S^{n+m-1}$ is nonempty, closed, and spherically convex. A direct application of the polarity principle (Fact \ref{polarity argument}) yields the desired sufficient condition \eqref{polar_s_con}.
	
\end{proof}

\subsection{Proof of Theorem \ref{them: Phase transition}}

\begin{proof}
	\textbf{Success case:} By Lemma \ref{lem: success and failure}, the penalized problem \eqref{penilized_procedure} succeeds if
	\begin{align*}
	\min_{\vr\in\S^{n-1}}\min_{\vs\in\TT_J^\circ} \|\vs-\mM^T\vr\|_2>0,
	\end{align*}
	where $\mM=[-\frac{\sqrt{m}}{\lambda}\mI_n, \mPhi^T]$. So it is sufficient to show that if the number of measurements satisfies \eqref{NumberofMeasurements1}, then the above inequality
    holds with high probability. The proof is also divided into three steps.

   \textbf{Step 1:  Problem reduction.} In this step, we similarly use Gordon's Lemma to convert the probability of the targeted event to another one which can be handled easily.

    Let $\vs = [\vs_1^T, \vs_2^T]^T$ and $\vu = [\vu_1^T, \vu_2^T]^T$. Note first that for any $\vr\in\S^{n-1}$, we have
	\begin{equation*}
	\begin{split}
	\min_{\vs\in\TT_J^\circ} \|\vs-\mM^T\vr\|_2&=\min_{\vs\in\TT_J^\circ}\left\|\begin{bmatrix} \vs_1+\frac{\sqrt{m}}{\lambda}\vr \\ \vs_2-\mPhi\vr \end{bmatrix}\right\|_2\\
	&=\min_{\vs\in\TT_J^\circ}\max_{\vu\in\S^{n+m-1}}\ip{\vs_1+\frac{\sqrt{m}}{\lambda}\vr}{-\vu_1}+\ip{\vs_2-\mPhi\vr}{-\vu_2}\\
	&=\min_{\vs\in\TT_J^\circ}\max_{\vu\in\S^{n+m-1}}\ip{\mPhi\vr}{\vu_2}-\ip{\frac{\sqrt{m}}{\lambda}\vr}{\vu_1}-\ip{\vs}{\vu}\\
	&\geq\max_{\vu\in\S^{n+m-1}}\left[\ip{\mPhi\vr}{\vu_2}-\ip{\frac{\sqrt{m}}{\lambda}\vr}{\vu_1}-\max_{\vs\in\TT_J^\circ}\ip{\vs}{\vu}\right]\\
	&=\max_{\vu\in\TT_J\cap\S^{n+m-1}}\ip{\mPhi\vr}{\vu_2}-\ip{\frac{\sqrt{m}}{\lambda}\vr}{\vu_1}.
	\end{split}
	\end{equation*}
	The inequality is due to the max-min inequality. The last step has used the fact that $\max_{\vs\in\TT_J^\circ}\ip{\vs}{\vu}=0$ when $\vu\in\TT_J$, otherwise it equals $\infty$. Thus we obtain
	\begin{align}\label{Lowbound1}
	&\min_{\vr\in\S^{n-1}}\min_{\vs\in\TT_J^\circ} \|\vs-\mM^T\vr\|_2\quad\geq\quad\min_{\vr\in\S^{n-1}}\max_{\vu\in\TT_J\cap\S^{n+m-1}}\ip{\mPhi\vr}{\vu_2}-\ip{\frac{\sqrt{m}}{\lambda}\vr}{\vu_1}.
	\end{align}
	We then use Gordon's Lemma to establish a lower bound for the probability of the targeted event. To this end, for any $\vr \in \S^{n-1}$ and $\vu \in \TT_J \cap \S^{n+m-1}$, define the following two Gaussian processes
	$$
	X_{\vr,\vu}:=\ip{\mPhi\vr}{\vu_2}+\|\vu_2\|_2\cdot g~~~
	$$
	and
	$$
	Y_{\vr,\vu}:=\ip{\vh}{\vu_2}+\|\vu_2\|_2\ip{\vg}{\vr}
	$$
	where $g \sim \NN(0, 1)$, $\vh \sim \NN(\vzero, \mI_m)$, and $\vg \sim \NN(\vzero, \mI_n)$ are independent of each other. Direct calculations show that these two defined Gaussian processes satisfy the conditions in Gordon's Lemma:
	\begin{align*}
	\E X_{\vr,\vu}^2 &= 2\| \vu_2\|_2^2 = \E Y_{\vr,\vu}^2,\\
	\E [X_{\vr,\vu} X_{\vr',\vu'}]-
	\E [Y_{\vr,\vu} Y_{\vr',\vu'}]&= \ip{\vr}{\vr'}\ip{\vu_2}{\vu_2'}+\|\vu_2\|_2\|\vu_2'\|_2-\ip{\vu_2}{\vu_2'}-\|\vu_2\|_2\|\vu_2'\|_2\ip{\vr}{\vr'}\\
	&=\left(1-\ip{\vr}{\vr'}\right)\left(\|\vu_2\|_2\|\vu_2'\|_2-\ip{\vu_2}{\vu_2'}\right)\\
	&\geq 0.
	\end{align*}
	In the last line, the equality holds when $\vr=\vr'$. It follows from Gordon's lemma (Fact \ref{Grodon lem}) that (by setting $\tau_{\vr,\vu}=\frac{\sqrt{m}}{\lambda}\vu_1^T\vr+0_+$)
	\begin{align*}
	\Pr{\min_{\vr\in\S^{n-1}}\max_{\vu\in\TT_J\cap\S^{n+m-1}}Y_{\vr,\vu}\geq \tau_{\vr,\vu}}&=\Pr{\min_{\vr\in\S^{n-1}}\max_{\vu\in\TT_J\cap\S^{n+m-1}}\ip{\vh}{\vu_2}+\|\vu_2\|_2\ip{\vg}{\vr}-\ip{\frac{\sqrt{m}}{\lambda}\vr}{\vu_1}> 0}\\
	&\leq \Pr{\min_{\vr\in\S^{n-1}}\max_{\vu\in\TT_J\cap\S^{n+m-1}}X_{\vr,\vu}\geq \tau_{\vr,\vu}}\\
	&= \Pr{\min_{\vr\in\S^{n-1}}\max_{\vu\in\TT_J\cap\S^{n+m-1}}\ip{\mPhi\vr}{\vu_2}+\|\vu_2\|_2\cdot g-\ip{\frac{\sqrt{m}}{\lambda}\vr}{\vu_1}> 0}\\
	&\leq \frac{1}{2} + \frac{1}{2} \Pr{\min_{\vr\in\S^{n-1}}\max_{\vu\in\TT_J\cap\S^{n+m-1}}\ip{\mPhi\vr}{\vu_2}+\|\vu_2\|_2\cdot g-\ip{\frac{\sqrt{m}}{\lambda}\vr}{\vu_1}> 0 ~ \big| ~g \leq 0}\\
    &\leq \frac{1}{2} + \frac{1}{2} \Pr{\min_{\vr\in\S^{n-1}}\max_{\vu\in\TT_J\cap\S^{n+m-1}}\ip{\mPhi\vr}{\vu_2}-\ip{\frac{\sqrt{m}}{\lambda}\vr}{\vu_1}> 0},
    \end{align*}
	which implies
	\begin{align}\label{pr of U>0}
	&\Pr{\min_{\vr\in\S^{n-1}}\max_{\vu\in\TT_J\cap\S^{n+m-1}}\ip{\mPhi\vr}{\vu_2}-\ip{\frac{\sqrt{m}}{\lambda}\vr}{\vu_1}> 0} \notag \\
	&\hspace{135pt}\geq 2\Pr{\underbrace{\min_{\vr\in\S^{n-1}}\max_{\vu\in\TT_J\cap\S^{n+m-1}}\ip{\vh}{\vu_2}+\|\vu_2\|_2\ip{\vg}{\vr}-\ip{\frac{\sqrt{m}}{\lambda}\vr}{\vu_1}}_{:=~\mathscr{E}_3}> 0}-1.
	\end{align}
	Moreover, $\mathscr{E}_3$ can be bounded from below as follows
	\begin{equation*}
	\begin{split}
	\mathscr{E}_3&=\min_{\vr\in\S^{n-1}}\max_{\vu\in\TT_J\cap\S^{n+m-1}}\ip{\vh}{\vu_2}+\ip{\vr}{\|\vu_2\|_2\vg-\frac{\sqrt{m}}{\lambda}\vu_1}\\
	&\geq\max_{\vu\in\TT_J\cap\S^{n+m-1}} \min_{\vr\in\S^{n-1}}\ip{\vh}{\vu_2}+\ip{\vr}{\|\vu_2\|_2\vg-\frac{\sqrt{m}}{\lambda}\vu_1}\\
	&=\max_{\vu\in\TT_J\cap\S^{n+m-1}}\ip{\vh}{\vu_2}-\left\|\|\vu_2\|_2\vg-\frac{\sqrt{m}}{\lambda}\vu_1\right\|_2\\
	&=\max_{\va\in\partial f(\vx^\star) \atop \vb\in\partial g(\vv^\star)} \underbrace{\frac{\|\vb\|_2}{\sqrt{\|\va\|_2^2+\|\vb\|_2^2}}}_{:= c(\va,\vb)}  \cdot\left[\ip{\vh}{\frac{\vb}{\|\vb\|_2}}-\left\|\vg-\frac{\sqrt{m}}{\lambda\|\vb\|_2}\va\right\|_2\right]\\
	\end{split}
	\end{equation*}
The first inequality holds because of the max-min inequality. In the last line, recall that the joint cone is defined as $\TT_J=\{t\cdot(\va,\vb):t\geq 0,~\va\in\partial f(\vx^\star),~\vb\in\partial g(\vv^\star)\}$, so we have let $\vu_1=\frac{\va}{\sqrt{\|\va\|_2^2+\|\vb\|_2^2}}$ and $\vu_2=\frac{\vb}{\sqrt{\|\va\|_2^2+\|\vb\|_2^2}}$ for $\va\in\partial f(\vx^\star),~\vb\in\partial g(\vv^\star)$.

Since $\partial f(\vx^\star)$ and $\partial g(\vv^\star)$ are nonempty and closed, we choose $(\va_0,\vb_0)$ such that
$$
(\va_0,\vb_0)\in\arg\max_{\va\in\partial f(\vx^\star) \atop \vb\in\partial g(\vv^\star)}\ip{\vh}{\frac{\vb}{\|\vb\|_2}}-\left\|\vg-\frac{\sqrt{m}}{\lambda\|\vb\|_2}\va\right\|_2,
$$
which leads to
	\begin{equation*}
	\begin{split}
	\mathscr{E}_3&\geq c(\va_0,\vb_0)\cdot\left[\ip{\vh}{\frac{\vb_0}{\|\vb_0\|_2}}-\left\|\vg-\frac{\sqrt{m}}{\lambda\|\vb_0\|_2}\va_0\right\|_2\right]\\
	&= c(\va_0,\vb_0)\cdot\max_{\va\in\partial f(\vx^\star) \atop \vb\in\partial g(\vv^\star)}\ip{\vh}{\frac{\vb}{\|\vb\|_2}}-\left\|\vg-\frac{\sqrt{m}}{\lambda\|\vb\|_2}\va\right\|_2\\
	&= c(\va_0,\vb_0)\cdot\max_{\alpha\leq t\leq\beta}\max_{\va\in\partial f(\vx^\star) \atop \vb\in\partial g(\vv^\star)\cap t\S^{m-1}}\ip{\vh}{\frac{\vb}{t}}-\left\|\vg-\frac{\sqrt{m}}{\lambda t}\va\right\|_2\\
	\end{split}
	\end{equation*}
In the last line, we have let $\|\vb\|_2=t$, $\alpha=\min_{\vb\in\partial g(\vv^\star)}\|\vb\|_2$, and $\beta=\max_{\vb\in\partial g(\vv^\star)}\|\vb\|_2$.

Define
\begin{equation*}
  L(\vg,\vh,t) := \max_{\va\in\partial f(\vx^\star) \atop \vb\in\partial g(\vv^\star)\cap t\S^{m-1}}\ip{\vh}{\frac{\vb}{t}}-\left\|\vg-\frac{\sqrt{m}}{\lambda t}\va\right\|_2,
\end{equation*}
and choose $t_3$ such that
\begin{align*}
t_3\in\arg\min_{\alpha\leq t\leq \beta} \E&\left[2\cdot\dist\left(\vg,\frac{\sqrt{m}}{\lambda t}\partial f(\vx^\star)\right)+\dist^2\left(\vh,\frac{1}{t}\partial g(\vv^\star)\cap \S^{m-1}\right)-1\right].
\end{align*}
Then we have
	\begin{equation}\label{boundofproblemreduction}
	\begin{split}
	\mathscr{E}_3&\geq c(\va_0,\vb_0)\cdot\max_{\alpha\leq t\leq\beta} L(\vg,\vh,t)\\
                 &\geq c(\va_0,\vb_0)\cdot L(\vg,\vh,t_3).
	\end{split}
	\end{equation}
Combining \eqref{Lowbound1}, \eqref{pr of U>0}, and  \eqref{boundofproblemreduction} yields
	\begin{equation}\label{Lowbound2}
	\begin{split}
	\Pr{ \min_{\vr\in\S^{n-1}}\min_{\vs\in\TT_J^\circ} \|\vs-\mM^T\vr\|_2>0}&\geq 2\Pr{c(\va_0,\vb_0)\cdot L(\vg,\vh,t_3) > 0}-1\\
	&= 2\Pr{ L(\vg,\vh,t_3) > 0}-1,
	\end{split}
	\end{equation}
where the last line holds because $0 \notin \partial g(\vv^\star)$, and hence $c(\va_0,\vb_0) >0$. Thus it suffices to establish the lower bound for $\Pr{ L(\vg,\vh,t_3) > 0}$.

\textbf{Step 2: Establish the lower bound for $\Pr{ L(\vg,\vh,t_3) > 0}$.} In this step, we apply the Gaussian concentration inequality to establish the lower bound for $\Pr{ L(\vg,\vh,t_3) > 0}$.


It follows from Lemma \ref{Lipschitz function} that the function $L(\vg,\vh,t_3)$ is a $\sqrt{2}$-Lipschitz function. Its expectation can be bounded from below as follows:
	\begin{equation*}
	\begin{split}
	\E L(\vg,\vh,t_3)
	&=\E\max_{\va\in\partial f(\vx^\star) \atop \vb\in\partial g(\vv^\star)\cap t_3\S^{m-1}}\ip{\vh}{\frac{\vb}{t_3}}-\left\|\vg-\frac{\sqrt{m}}{\lambda t_3}\va\right\|_2\\
	&=\E \max_{\va\in\partial f(\vx^\star) \atop \vb\in\partial g(\vv^\star)\cap t_3\S^{m-1}}\frac{1+\|\vh\|_2^2-\|\vh-\frac{\vb}{t_3}\|_2^2}{2}-\left\|\vg-\frac{\sqrt{m}}{\lambda t_3}\va\right\|_2\\
	&=\frac{m}{2}-\E\min_{\va\in\partial f(\vx^\star) \atop \vb\in\partial g(\vv^\star)\cap t_3\S^{m-1}}\left(\left\|\vg-\frac{\sqrt{m}}{\lambda t_3}\va\right\|_2+\frac{1}{2}\|\vh-\frac{\vb}{t_3}\|_2^2-\frac{1}{2}\right)\\
    & = \frac{m}{2}-\E \left[  \dist\left(\vg,\frac{\sqrt{m}}{\lambda t_3}\partial f(\vx^\star)\right) + \frac{1}{2}\cdot \dist^2\left(\vh,\frac{1}{t_3}\partial g(\vv^\star)\cap \S^{m-1}\right) - \frac{1}{2}  \right] \\
	& = \frac{m}{2}- \min_{\alpha\leq t\leq \beta} \E \left[  \dist\left(\vg,\frac{\sqrt{m}}{\lambda t}\partial f(\vx^\star)\right) + \frac{1}{2}\cdot \dist^2\left(\vh,\frac{1}{t}\partial g(\vv^\star)\cap \S^{m-1}\right) - \frac{1}{2}  \right] \\
    &\geq \epsilon/2.
	\end{split}
	\end{equation*}
	The last line is due to the  measurement condition \eqref{NumberofMeasurements1}.
	
	Now using the Gaussian concentration inequality (Fact \ref{gaussian concentration}) yields
	\begin{equation*}
	\Pr{ L(\vg,\vh,t_3) - \E L(\vg,\vh,t_3) \leq -\epsilon/2} \leq \exp\left(\frac{-\epsilon^2}{16}\right).
	\end{equation*}
	which in turn implies
	\begin{equation}\label{pr of L>0}
	\begin{split}
	\Pr{ L(\vg,\vh,t_3) > 0}&\geq\Pr{ L(\vg,\vh,t_3) > \E L(\vg,\vh,t_3)-\epsilon/2}\\
	&\geq 1-\exp\left(\frac{-\epsilon^2}{16}\right).
	\end{split}
	\end{equation}

\textbf{Step 3: Complete the proof.}

    Combining \eqref{Lowbound2} and \eqref{pr of L>0}, we have
	\begin{equation*}
	\begin{split}
	\Pr{ \min_{\vr\in\S^{n-1}}\min_{\vs\in\TT_J^\circ} \|\vs-\mM^T\vr\|_2>0} & \geq  2\Pr{ L(\vg,\vh,t_3) > 0}-1\\
	&\geq 1-2\exp\left(\frac{-\epsilon^2}{16}\right).
	\end{split}
	\end{equation*}
	This means that when $m \geq \mathscr{C}_p(\lambda)+\epsilon$, the penalized problem \eqref{penilized_procedure} succeeds with probability at least $1-2\exp\left(\frac{-\epsilon^2}{16}\right)$.

	\textbf{Failure case:} According to Lemma \ref{lem: success and failure}, the penalized problem \eqref{penilized_procedure} fails if
	\begin{align*}
	\min_{\va\in\partial f(\vx^\star)\atop \vb\in\partial g(\vv^\star)} \|\mPhi^T\vb-\frac{\sqrt{m}}{\lambda}\va\|_2>0.
	\end{align*}
    So it is enough to show that if the number of measurements satisfies \eqref{NumberofMeasurements2}, then the above inequality
    holds with high probability. For clarity, the proof is similarly divided into three steps.

    \textbf{Step 1: Problem reduction.} In this step, we employ Gordon's Lemma to convert the probability of the targeted event to another one which is easy to handle.

	Note that
	\begin{equation}\label{bound of failure cond}
	\min_{\va\in\partial f(\vx^\star) \atop \vb\in\partial g(\vv^\star)} \|\mPhi^T\vb-\frac{\sqrt{m}}{\lambda}\va\|_2=\min_{\va\in\partial f(\vx^\star) \atop \vb\in\partial g(\vv^\star)}\max_{\vu\in\S^{n-1}}\ip{\mPhi^T\vb}{\vu}-\ip{\frac{\sqrt{m}}{\lambda}\va}{\vu}.
	\end{equation}
	We then use Gordon's Lemma to establish a lower bound for the probability of the targeted event. To this end, for any $(\va,\vb) \in \partial f(\vx^{\star}) \times \partial g(\vv^{\star})$ and $\vu\in \S^{n-1}$,  define the following two Gaussian processes
	$$
	X_{(\va,\vb),\vu}:=\ip{\mPhi^T\vb}{\vu}+ \|\vb\|_2\cdot g
	$$
	and
	$$
	Y_{(\va,\vb),\vu}:=\|\vb\|_2\ip{\vg}{\vu}+\ip{\vh}{\vb},
	$$
	here $g \sim \NN(0, 1)$, $\vg \sim \NN(\vzero, \mI_n)$, and $\vh \sim \NN(\vzero, \mI_m)$ are independent of each other. It is not hard to check that these two defined processes satisfy the conditions in Gordon's Lemma:
	\begin{align*}
	\E X_{(\va,\vb),\vu}^2 &= 2\|\vb\|_2^2= \E Y_{(\va,\vb),\vu}^2, \\
	\E [X_{(\va,\vb),\vu} X_{(\va',\vb'),\vu'}]-
	\E [Y_{(\va,\vb),\vu} Y_{(\va',\vb'),\vu'}]&= \ip{\vu}{\vu'}\ip{\vb}{\vb'}+\|\vb\|_2\|\vb'\|_2-\ip{\vb}{\vb'}-\|\vb\|_2\|\vb'\|_2\ip{\vu}{\vu'}\\
	&=\left(1-\ip{\vu}{\vu'}\right)\left(\|\vb\|_2\|\vb'\|_2-\ip{\vb}{\vb'}\right)\\
	&\geq 0.
	\end{align*}
	In the last line, the equality holds when $\vb=\vb'$. It follows from Gordon's Lemma (Fact \ref{Grodon lem}) that (by setting $\tau_{(\va,\vb),\vu}=\frac{\sqrt{m}}{\lambda}\va^T\vu+0_+$)
	\begin{align*}
	\Pr{\min_{\va\in\partial f(\vx^\star) \atop \vb\in\partial g(\vv^\star)}\max_{\vu\in\S^{n-1}}Y_{(\va,\vb),\vu}\geq \tau_{(\va,\vb),\vu}}&=\Pr{\min_{\va\in\partial f(\vx^\star) \atop \vb\in\partial g(\vv^\star)}\max_{\vu\in\S^{n-1}}\|\vb\|_2\ip{\vg}{\vu}+\ip{\vh}{\vb}-\ip{\frac{\sqrt{m}}{\lambda}\va}{\vu}> 0}\\
	&\leq\Pr{\min_{\va\in\partial f(\vx^\star) \atop \vb\in\partial g(\vv^\star)}\max_{\vu\in\S^{n-1}}X_{(\va,\vb),\vu}\geq \tau_{(\va,\vb),\vu}} \\
	&= \Pr{\min_{\va\in\partial f(\vx^\star) \atop \vb\in\partial g(\vv^\star)}\max_{\vu\in\S^{n-1}}\ip{\mPhi^T\vb}{\vu}+\|\vb\|_2\cdot g-\ip{\frac{\sqrt{m}}{\lambda}\va}{\vu}> 0}\\
	&\leq \frac{1}{2}+ \frac{1}{2}\Pr{\min_{\va\in\partial f(\vx^\star) \atop \vb\in\partial g(\vv^\star)}\max_{\vu\in\S^{n-1}}\ip{\mPhi^T\vb}{\vu}+\|\vb\|_2\cdot g-\ip{\frac{\sqrt{m}}{\lambda}\va}{\vu}> 0 \Big | g \leq 0}\\
    &\leq \frac{1}{2}+ \frac{1}{2}\Pr{\min_{\va\in\partial f(\vx^\star) \atop \vb\in\partial g(\vv^\star)}\max_{\vu\in\S^{n-1}}\ip{\mPhi^T\vb}{\vu} -\ip{\frac{\sqrt{m}}{\lambda}\va}{\vu}> 0},
\end{align*}
	which implies
	\begin{align}\label{pr of W>0}
	&\Pr{\min_{\va\in\partial f(\vx^\star) \atop \vb\in\partial g(\vv^\star)}\max_{\vu\in\S^{n-1}}\ip{\mPhi^T\vb}{\vu}-\ip{\frac{\sqrt{m}}{\lambda}\va}{\vu}> 0} \notag \\
	&\hspace{135pt}\geq 2\Pr{\underbrace{\min_{\va\in\partial f(\vx^\star) \atop \vb\in\partial g(\vv^\star)}\max_{\vu\in\S^{n-1}}\|\vb\|_2\ip{\vg}{\vu}+\ip{\vh}{\vb}-\ip{\frac{\sqrt{m}}{\lambda}\va}{\vu}}_{:=\mathscr{E}_4}> 0}-1.
	\end{align}
	Moreover,  $\mathscr{E}_4$ can be bounded from below as follows:
	\begin{equation*}
	\begin{split}
	\mathscr{E}_4 &=\min_{\va\in\partial f(\vx^\star) \atop \vb\in\partial g(\vv^\star)}\max_{\vu\in\S^{n-1}}\ip{\vu}{\|\vb\|_2\vg-\frac{\sqrt{m}}{\lambda}\va}+\ip{\vh}{\vb}\\
	&=\min_{\va\in\partial f(\vx^\star) \atop \vb\in\partial g(\vv^\star)}\left\|\|\vb\|_2\vg-\frac{\sqrt{m}}{\lambda}\va\right\|_2+\ip{\vh}{\vb}\\
	&=\left\|\|\vb_1\|_2\vg-\frac{\sqrt{m}}{\lambda}\va_1\right\|_2+\ip{\vh}{\vb_1}\\
	&\geq \|\vb_1\|_2\cdot\min_{\va\in\partial f(\vx^\star) \atop \vb\in\partial g(\vv^\star)}\left\|\vg-\frac{\sqrt{m}}{\lambda\|\vb\|_2}\va\right\|_2+\ip{\vh}{\frac{\vb}{\|\vb\|_2}}\\
	&= \|\vb_1\|_2\cdot\min_{\alpha\leq t\leq\beta}\min_{\va\in\partial f(\vx^\star) \atop \vb\in\partial g(\vv^\star)\cap t\S^{m-1}}\left\|\vg-\frac{\sqrt{m}}{\lambda t}\va\right\|_2+\ip{\vh}{\frac{\vb}{t}}\\
	\end{split}
	\end{equation*}
	In the third line, we have chosen $(\va_1,\vb_1)$ such that
	$$
	(\va_1,\vb_1)\in\arg\min_{\va\in\partial f(\vx^\star) \atop \vb\in\partial g(\vv^\star)}\left\|\|\vb\|_2\vg-\frac{\sqrt{m}}{\lambda}\va\right\|_2+\ip{\vh}{\vb}.
	$$
	In the last line, we have let $\|\vb\|_2=t$, $\alpha=\min_{\vb\in\partial g(\vv^\star)}\|\vb\|_2$, and $\beta=\max_{\vb\in\partial g(\vv^\star)}\|\vb\|_2$.

    Define
    \begin{equation*}
      V(\vg,\vh,t) = \min_{\va\in\partial f(\vx^\star) \atop \vb\in\partial g(\vv^\star)\cap t\S^{m-1}}\left\|\vg-\frac{\sqrt{m}}{\lambda t}\va\right\|_2+\ip{\vh}{\frac{\vb}{t}},
    \end{equation*}
    and choose $t_4$ such that
	\begin{align*}
	t_4\in\arg\min_{\alpha\leq t\leq \beta}V(\vg,\vh,t).
	\end{align*}
    Then we have
    \begin{equation}\label{bound of W}
	\begin{split}
	\mathscr{E}_4 & \geq \|\vb_1\|_2\cdot\min_{\alpha\leq t\leq \beta}V(\vg,\vh,t)\\
                  &=\|\vb_1\|_2\cdot V(\vg,\vh,t_4).
	\end{split}
	\end{equation}

    Combining  \eqref{bound of failure cond}, \eqref{pr of W>0}, and \eqref{bound of W} yields
    \begin{equation}\label{boundofstep12}
	\begin{split}
	\Pr{\min_{\va\in\partial f(\vx^\star) \atop\vb\in\partial g(\vv^\star)} \|\mPhi^T\vb-\frac{\sqrt{m}}{\lambda}\va\|_2>0} &\geq 2\Pr{\|\vb_1\|_2\cdot V(\vg,\vh,t_4) > 0}-1\\
	&= 2\Pr{V(\vg,\vh,t_4) > 0}-1,
	\end{split}
	\end{equation}
    where we have used the fact that $0 \notin \partial g(\vv^\star)$ and hence $\|\vb_1\|_2 >0$. Thus it is enough to establish the lower bound for $\Pr{V(\vg,\vh,t_4) > 0}$.

    \textbf{Step 2: Establish the lower bound for $\Pr{V(\vg,\vh,t_4) > 0}$.} In this step, we use the Gausian concentration inequaltiy to establish the lower bound for $\Pr{V(\vg,\vh,t_4) > 0}$.

    By Lemma \ref{Lipschitz function}, the function $V(\vg,\vh,t_4)$ is a $\sqrt{2}$-Lipschitz function. Its expectation can be bounded:
	\begin{equation}\label{bound of V}
	\begin{split}
	\E V(\vg,\vh,t_4)&=\E\min_{\va\in\partial f(\vx^\star) \atop \vb\in\partial g(\vv^\star)\cap t_4\S^{m-1}}\left\|\vg-\frac{\sqrt{m}}{\lambda t_4}\va\right\|_2+\ip{\vh}{\frac{\vb}{t_4}}\\
	&=\E\min_{\va\in\partial f(\vx^\star) \atop \vb\in\partial g(\vv^\star)\cap t_4\S^{m-1}}\left\|\vg-\frac{\sqrt{m}}{\lambda t_4}\va\right\|_2+\ip{-\vh}{\frac{\vb}{t_4}}\\
	&=\E \min_{\va\in\partial f(\vx^\star) \atop \vb\in\partial g(\vv^\star)\cap t_4\S^{m-1}}\left\|\vg-\frac{\sqrt{m}}{\lambda t_4}\va\right\|_2+\frac{\|\vh-\frac{\vb}{t_4}\|_2^2-1-\|\vh\|_2^2}{2}\\
	&=\E  \left[ \min_{\va\in\partial f(\vx^\star) \atop \vb\in\partial g(\vv^\star)\cap t_4\S^{m-1}}\left(\left\|\vg-\frac{\sqrt{m}}{\lambda t_4}\va\right\|_2+\frac{1}{2}\|\vh-\frac{\vb}{t_4}\|_2^2-\frac{1}{2}\right)\right]-\frac{m}{2}\\
	&\geq \min_{\alpha\leq t\leq \beta}\E \left[ \min_{\va\in\partial f(\vx^\star) \atop \vb\in\partial g(\vv^\star)\cap t\S^{m-1}}\left(\left\|\vg-\frac{\sqrt{m}}{\lambda t}\va\right\|_2+\frac{1}{2}\|\vh-\frac{\vb}{t}\|_2^2-\frac{1}{2}\right)\right]-\frac{m}{2}\\
    & =  \min_{\alpha\leq t\leq \beta}\E \left[ \dist\left(\vg,\frac{\sqrt{m}}{\lambda t}\partial f(\vx^\star)\right) + \frac{1}{2} \dist^2\left(\vh,\frac{1}{t}\partial g(\vv^\star)\cap \S^{m-1}\right) - \frac{1}{2} \right]-\frac{m}{2} \\
	&\geq \epsilon/2.
	\end{split}
	\end{equation}
	The second line is due to the fact that $-\vh$ and $\vh$ have the same distribution. The last inequality is due to the measurement condition \eqref{NumberofMeasurements2}.
	
	Now it follows from the Gaussian concentration inequality (Fact \ref{gaussian concentration}) that
	\begin{equation*}
	\Pr{ V(\vg,\vh,t_4) - \E V(\vg,\vh,t_4) \leq -\epsilon/2} \leq \exp\left(\frac{-\epsilon^2}{16}\right),
	\end{equation*}
	which in turn implies
	\begin{equation}\label{pr of V>0}
	\begin{split}
	\Pr{ V(\vg,\vh,t_4) > 0}&\geq\Pr{ V(\vg,\vh,t_4) >\E V(\vg,\vh,t_4)-\epsilon/2}\\
	&\geq 1-\exp\left(\frac{-\epsilon^2}{16}\right).
	\end{split}
	\end{equation}

\textbf{Step 3: Complete the proof.}

Putting \eqref{boundofstep12} and \eqref{pr of V>0} together, we have
	\begin{equation*}
	\begin{split}
	\Pr{\min_{\va\in\partial f(\vx^\star) \atop\vb\in\partial g(\vv^\star)} \|\mPhi^T\vb-\frac{\sqrt{m}}{\lambda}\va\|_2>0} & \geq  2\Pr{V(\vg,\vh,t_4) > 0}-1\\
	&\geq 1-2\exp\left(\frac{-\epsilon^2}{16}\right).
	\end{split}
	\end{equation*}
	Thus we have shown that when $m \leq  \mathscr{C}_p(\lambda)-\epsilon$, the penalized problem \eqref{penilized_procedure} fails with probability at least $1-2\exp\left(\frac{-\epsilon^2}{16}\right)$. This completes the proof.

\end{proof}

\section{Proof of Theorem \ref{relation_CC_PP}}
In this appendix, we prove our last theorem which establishes the relationship between constrained and penalized recovery procedures and illustrates how to select the optimal parameter $\lambda$ for the penalized method. Some auxiliary lemma and facts used in the proof are included in Appendix \ref{auxiliaryresults}.

\begin{figure*}
	\centering
	\includegraphics[width= .6\textwidth]{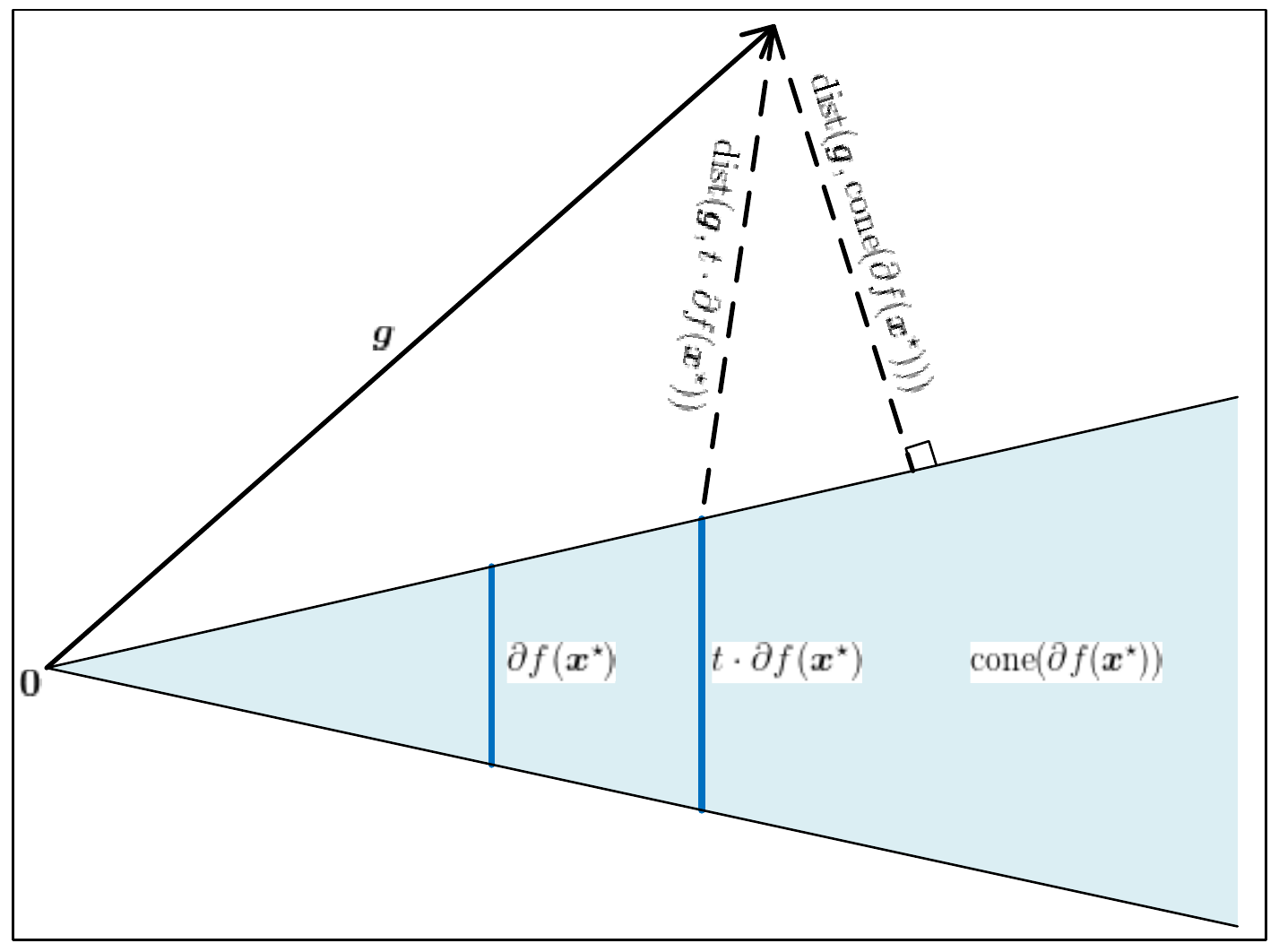}
	\caption{Illustration of the distances of a vector to the scaled subdifferential $t\cdot\partial f(\vx^\star)$ and to the cone of the subdifferential $\cone(\partial f(\vx^\star))$.}
	\label{fig:Gaussian distance}
\end{figure*}

\begin{proof}
		The core ingredient in the proof of Theorem \ref{relation_CC_PP} is the fact that the distance of a vector to the scaled subdifferential $t\cdot\partial f(\vx^\star)$ can always be bounded from below by that of this vector to the cone of the subdifferential $\cone(\partial f(\vx^\star))$ (see Fig. \ref{fig:Gaussian distance}). With this observation in mind, we have
		\begin{align} \label{relation_cp_cpl}
		\mathscr{C}_p(\lambda)&=\min_{\alpha\leq t\leq \beta} \E\left[2\cdot\dist\left(\vg,\frac{\sqrt{m}}{\lambda t}\partial f(\vx^\star)\right)+\dist^2\left(\vh,\frac{1}{t}\partial g(\vv^\star)\cap \S^{m-1}\right)-1\right] \notag \\
		&\geq\E\left[2\cdot\dist\left(\vg,\cone(\partial f(\vx^\star))\right)+\dist^2\left(\vh,\cone(\partial g(\vv^\star))\cap \S^{m-1}\right)-1\right] \notag \\
		&=\E\left[2\cdot\dist\left(\vg,\cone(\partial f(\vx^\star))\right)+\min_{\vb\in\cone(\partial g(\vv^\star))\cap \S^{m-1}}\|\vh-\vb\|_2^2-1\right] \notag\\
		&=\E\left[2\cdot\dist\left(\vg,\cone(\partial f(\vx^\star))\right)-2\max_{\vb\in\cone(\partial g(\vv^\star))\cap \S^{m-1}}\ip{\vh}{\vb}+m\right].
		\end{align}
By assumptions, $\vzero\notin\partial f(\vx^\star)$ and $\vzero\notin\partial g(\vv^\star)$, we obtain
		$$
		\cone(\partial f(\vx^\star))=\NN_f(\vx^{\star})~~~\textrm{and}~\cone(\partial g(\vv^\star))=\NN_g(\vv^{\star}).
		$$
		Substituting the above equalities into \eqref{relation_cp_cpl} and rearranging yields
		\begin{align} \label{key_obe_cp_cpl}
		\frac{\mathscr{C}_p(\lambda)-m}{2}\geq \E\left[\dist\left(\vg,\NN_f(\vx^\star)\right)-\max_{\vb\in\NN_g(\vv^\star)\cap \S^{m-1}}\ip{\vh}{\vb}\right]=\E\dist(\vg,\NN_f(\vx^\star))-\omega(\NN_g(\vv^\star)\cap \S^{m-1}).
		\end{align}
		We are now ready to establish the bounds in Theorem \ref{relation_CC_PP}. First consider the case in which $m\geq \mathscr{C}_p(\lambda)$, then \eqref{key_obe_cp_cpl} implies that
		\begin{align*}
		0&\geq\E^2\dist(\vg,\NN_f(\vx^\star))-\omega^2(\NN_g(\vv^\star)\cap \S^{m-1})\\
		&\geq \E\dist^2(\vg,\NN_f(\vx^\star))-1-\omega^2(\NN_g(\vv^\star)\cap \S^{m-1})\\
		&=\E \left(\max_{\vx\in\TT_f(\vx^\star)\cap\B_2^n}\ip{\vx}{\vg}\right)^2-1-\omega^2(\NN_g(\vv^\star)\cap \S^{m-1})\\
		&\geq\omega^2(\TT_f(\vx^\star)\cap\S^{n-1})-1-\omega^2(\NN_g(\vv^\star)\cap \S^{m-1})\\
		&\geq\omega^2(\TT_f(\vx^\star)\cap\S^{n-1})-1-[m-\omega^2(\TT_g(\vv^\star)\cap \S^{m-1})]\\
		&=\mathscr{C}_p-m-1.
		\end{align*}
		The second inequality is due to the relation \eqref{Relationship}. The third line has used Fact \ref{omega & dist} and the assumption that $\TT_f(\vx^\star)$ is closed (thus $\NN_f^\circ(\vx^\star)=\TT_f(\vx^\star)$). The next two lines follow from Facts \ref{relation_w_d} and \ref{Gs width of normal cone}, respectively. Thus we have shown that if $m\geq \mathscr{C}_p(\lambda)$, then $m\geq \mathscr{C}_p-1$.

		We next consider the case where $m\geq \mathscr{C}_p$. It is not hard to find that for $\alpha \leq  t \leq \beta $,  the two Gaussian distances in $\mathscr{C}_p(\lambda)$ have the following lower bounds:
		 $$\E\dist^2\left(\vh,\frac{1}{t}\partial g(\vv^\star)\cap \S^{m-1}\right)\geq\E\dist^2\left(\vh,\cone(\partial g(\vv^\star))\cap \S^{m-1}\right)$$
        and
        $$\E\dist\left(\vg,\frac{\sqrt{m}}{\lambda t}\partial f(\vx^\star)\right)\geq\E\dist\left(\vg,\cone(\partial f(\vx^\star))\right). $$
		Then we can choose
		$$
		t^\star=\arg\min_{\alpha\leq t\leq \beta}\eta^2\left(\frac{1}{t} \partial g(\vv^{\star})\cap\S^{m-1}\right)  ~~ \textrm{and} ~~  \lambda^\star=\arg\min_{\lambda>0}\zeta\left(\frac{\sqrt m}{\lambda t^\star} \partial f(\vx^{\star})\right)
		$$
		such that the above two Gaussian distances attain their lower bounds simultaneously. The second lower bound is achievable because ${\sqrt m}/{\lambda t^\star}$ can take any positive number if $\lambda>0$. Thus $\mathscr{C}_p(\lambda)$ attains its minimum at $\lambda^\star$, i.e.,
		\begin{align*}
		\mathscr{C}_p(\lambda^\star)&=\E\left[2\cdot\dist\left(\vg,\cone(\partial f(\vx^\star))\right)+\dist^2\left(\vh,\cone(\partial g(\vv^\star))\cap \S^{m-1}\right)-1\right]\\
		&=\E\left[2\cdot\dist\left(\vg,\NN_f(\vx^\star)\right)-2\max_{\vb\in\NN_g(\vv^\star)\cap \S^{m-1}}\ip{\vh}{\vb}+m\right].
		\end{align*}
		Further, we have the following upper bound
		\begin{align}\label{upperbound}
		\frac{\mathscr{C}_p(\lambda^\star)-m}{2}&= \E\dist(\vg,\NN_f(\vx^\star))-\omega(\NN_g(\vv^\star)\cap \S^{m-1})\\ \notag
		&\leq\sqrt{\E\dist^2(\vg,\NN_f(\vx^\star))}-\omega(\NN_g(\vv^\star)\cap \S^{m-1})\\ \notag
		&=\sqrt{\delta(\TT_f(\vx^\star))}-\omega(\NN_g(\vv^\star)\cap \S^{m-1})\\ \notag
		&\leq \sqrt{\omega^2(\TT_f(\vx^\star)\cap\S^{n-1})+1}-\sqrt{\delta(\NN_g(\vv^\star))-1}\\ \notag
		&= \sqrt{\omega^2(\TT_f(\vx^\star)\cap\S^{n-1})+1}-\sqrt{m-\delta(\TT_g(\vv^\star))-1}\\ \notag
		&\leq\sqrt{\omega^2(\TT_f(\vx^\star)\cap\S^{n-1})+1}-\sqrt{m-\omega^2(\TT_g(\vv^\star)\cap \S^{m-1})-2}\\ \notag
		&\leq\omega(\TT_f(\vx^\star)\cap\S^{n-1})+1-\sqrt{m-\omega^2(\TT_g(\vv^\star)\cap \S^{m-1})}+\sqrt{2}\\ \notag
		&\leq \sqrt{\mathscr{C}_p}-\sqrt{m}+1+\sqrt{2}\\ \notag
		&\leq 1+\sqrt{2}.
		\end{align}
		Here, the first inequality is due to Jensen's inequality. The next two lines have used Facts \ref{omega & dist} and \ref{relation_w_d}, respectively. The third equality holds because of Moreau's decomposition theorem (Fact \ref{Moreau}). The next inequality has used Fact \ref{relation_w_d} again. The lase two inequalities follow from the condition $m\geq \mathscr{C}_p$. Rearranging completes the proof.	

It's worth noting that if we impose an extra condition on $\omega(\TT_f(\vx^\star)\cap\S^{n-1})$ (e.g., $\omega(\TT_f(\vx^\star)\cap\S^{n-1})\geq 4$, which can be easily satisfied in practical applications), then we can obtain a sharper upper bound than \eqref{upperbound}:
			\begin{align*}
			\frac{\mathscr{C}_p(\lambda^\star)-m}{2}
			&\leq\sqrt{\omega^2(\TT_f(\vx^\star)\cap\S^{n-1})+1}-\sqrt{m-\omega^2(\TT_g(\vv^\star)\cap \S^{m-1})-2}\\
			&\leq\omega(\TT_f(\vx^\star)\cap\S^{n-1})+(\sqrt{17}-\sqrt{16})-\sqrt{m-\omega^2(\TT_g(\vv^\star)\cap \S^{m-1})}+(\sqrt{16}-\sqrt{14})\\
			&\leq \sqrt{\mathscr{C}_p}-\sqrt{m}+\sqrt{17}-\sqrt{14}\\
			&\leq \frac{1}{2}.
			\end{align*}
	The second inequality is due to the facts that $\sqrt{a^2+1}\leq a+(\sqrt{17}-\sqrt{16})$ for $a\geq 4$ and $\sqrt{b^2-2}\geq b-(\sqrt{16}-\sqrt{14})$ for $b = \sqrt{m-\omega^2(\TT_g(\vv^\star)\cap \S^{m-1})} \geq 4$. The lase two inequalities have used the condition $m\geq \mathscr{C}_p$. Rearranging yields $m\geq \mathscr{C}_p(\lambda^\star)-1$, which leads to a smaller gap than $5$.

\end{proof}

\section{Evaluate $\mathscr{C}_p$ and $\mathscr{C}_p(\lambda)$ for Typical Structured Signal and Corruption}
\label{cal}
In Section \ref{Main results}, we have shown that $\mathscr{C}_p$ and $\mathscr{C}_p(\lambda)$ can be accurately estimated by equations \eqref{cal_ph of con} and \eqref{cal_ph of pen}, respectively. Thus it is sufficient to evaluate two related functionals: Gaussian squared distance to a scaled subdifferential and spherical Gaussian width of a scaled subdifferential. There exist some standard recipes to calculate these two quantities in the literature, see e.g., \cite{chandrasekaran2012convex,amelunxen2014living,foygel2014corrupted}. For the completeness of this paper, we calculate these two functionals for sparse vectors and low-rank matrices in this appendix.
%

\subsection{Calculation for Sparse Vectors}
Let $\vx$ be an $s$-sparse vector in $\R^n$, and let $S$ denote its support. We use the $\ell_1$-norm to promote the structure of sparse vectors. The scaled subdifferential of $\|\vx\|_1$ is given by
$$
t\cdot\partial\|\vx\|_1=\{\vz\in\R^n:\vz_i=t\cdot\textrm{sgn}(\vx_i)~\textrm{for}~i\in S,~|\vz_i|\leq t~\textrm{for}~i\in S^c\},
$$
where $S^c$ represents the complement of $S$. Then the Gaussian squared distance to a scaled subdifferential can be calculated as
\begin{align*}
\eta^2(t\cdot\partial\|\vx\|_1)&=\E \inf_{\vz\in t\cdot\partial\|\vx\|_1}\|\vg-\vz\|_2^2=\E \inf_{\vz\in t\cdot\partial\|\vx\|_1}\sum_{i\in S}(\vg_i-\vz_i)^2+\sum_{i\in S^c}(\vg_i-\vz_i)^2\\
&=\E ~ \sum_{i\in S}(\vg_i-t\cdot \textrm{sgn}(x_i))^2+\sum_{i\in S^c}\textrm{shrink}(\vg_i,t)^2\\
&=s(1+t^2)+\frac{2(n-s)}{\sqrt{2\pi}}\left((1+t^2)\int_{t}^{\infty}e^{-x^2/2}dx-te^{-t^2/2}\right).
\end{align*}
Here $\textrm{shrink}(\vg_i,t)$ is the soft thresholding operator defined as:
\begin{equation*}
\textrm{shrink}(\vg_i,t)=\left\{
\begin{array}{ll}
\vg_i+t  & \textrm{if}~~ \vg_i< -t, \\
0      &\textrm{if}~~-t\leq\vg_i\leq t,\\
\vg_i-t &\textrm{if}~~  \vg> t.
\end{array}
\right.
\end{equation*}

Let $\vv$ be a $k$-sparse vector in $\R^m$, and let $S$ denote the support of $\vv$. The scaled spherical part of the subdifferential $\partial\|\vv\|_1$ is given by
$$
\partial\|\vv\|_1\cap t\S^{m-1}=\{\vz\in\R^m:\vz_i=\textrm{sgn}(\vv_i)~\textrm{for}~i\in S,~|\vz_i|\leq 1~\textrm{for}~i\in S^c,~\|\vz\|_2=t \}.
$$
Then the spherical Gaussian width of a scaled subdifferential is
\begin{align*}
\omega\left(\frac{1}{t}\partial\|\vv\|_1\cap\S^{m-1}\right)&=\frac{1}{t}\cdot\omega(\partial\|\vv\|_1\cap t\S^{m-1})=\frac{1}{t}\E \sup_{\vz\in \partial\|\vv\|_1\cap t\S^{m-1}}\ip{\vg}{\vz}\\
&=\frac{1}{t}\E \sup_{\vz\in \partial\|\vv\|_1\cap t\S^{m-1}}\sum_{i\in S}\vg_i\vz_i+\sum_{i\in S^c}\vg_i\vz_i\\
&=\frac{1}{t}\E \sup_{\vz\in \partial\|\vv\|_1\cap t\S^{m-1}}\sum_{i\in S^c}\vg_i\vz_i\\
&=\frac{1}{t}\sup_{\vz\in \partial\|\vv\|_1\cap t\S^{m-1}}\sum_{i\in S^c}|\vz_i|\cdot\E|\vg_i|\\
&=\frac{1}{t}\sqrt{m-k}\sqrt{t^2-k}\sqrt{\frac{2}{\pi}}=\sqrt{\frac{2}{\pi}(m-k)\left(1-\frac{k}{t^2}\right)}.
\end{align*}
The last line is due to the Cauchy-Schwarz inequality, i.e., $\sum_{i\in S^c}|\vz_i|\leq\sqrt{m-k}\cdot\sqrt{\sum_{i\in S^c}\vz_i^2}=\sqrt{m-k}\cdot\sqrt{t^2-k}$. The equality holds when $|\vz_i|=\sqrt{\frac{t^2-k}{m-k}}$ for $i\in S^c$.

\subsection{Calculation for Low-rank Matrices}
Let $\mX\in\R^{n_1\times n_2}$ be an $r$-rank matrix with $n_1\leq n_2$. We use the nuclear norm $\|\mX\|_*$, which is the sum of singular values of $\mX$, to promote the structure of low-rank matrices. Note that the nuclear norm and Gaussian distance are both unitary invariant, without loss of generality, we can assume that $\mX$ takes the form
$$
\mX=\begin{bmatrix}
\mSigma & \vzero\\
\vzero & \vzero
\end{bmatrix},
$$
where $\mSigma=\textrm{diag}(\sigma_1,\sigma_2,...,\sigma_r)$. The scaled subdifferential of $\|\mX\|_*$ is
\begin{equation*}
t\cdot\partial\|\mX\|_*=\left\{\begin{bmatrix}
t\mI_r & \vzero\\
\vzero & t\mW
\end{bmatrix}:~\|\mW\|\leq 1\right\}.
\end{equation*}
Here $\|\mW\|$ is the spectral norm, which is equal to the maximum singular value of $\mW$. Let
$$
\mG=\begin{bmatrix}
\mG_1 & \mG_1'\\
\mG_2' & \mG_2
\end{bmatrix}
$$
be a partition of Gaussian matrix $\mG\in\R^{n_1\times n_2}$ with $\mG_1\in\R^{r\times r}$ and $\mG_2\in\R^{(n_1-r)\times (n_2-r)}$. Then the Gaussian squared distance to a scaled subdifferential is
\begin{align*}
\eta^2(t\cdot\partial\|\mX\|_*)&=\E \inf_{\mZ\in t\cdot\partial\|\mX\|_*}\|\mG-\mZ\|_F^2=\E \left\{\left\|\begin{bmatrix}
\mG_1-t\mI_r & \mG_1'\\
\mG_2' & \vzero
\end{bmatrix}\right\|_F^2+\inf_{\|\mW\|\leq 1}\|\mG_2-t\mW\|_F^2\right\}\\
&=r(n_1+n_2-r+t^2)+\E\inf_{\|\mW\|\leq 1}\sum_{i=1}^{n_1-r}\left(\sigma_i(\mG_2)-t\sigma_i(\mW)\right)^2\\
&=r(n_1+n_2-r+t^2)+\E\sum_{i=1}^{n_1-r}\textrm{shrink}\left(\sigma_i(\mG_2),t\right)^2.
\end{align*}
Here $\sigma_i(\cdot)$ is the $i$-th largest singular value. The expectation term concerns the density of singular values of Gaussian matrix $\mG_2$, it seems challenging to obtain an exact formula for this term. However, there exist some asymptotic results in the literatures, see e.g. \cite{amelunxen2014living,oymak13The}. In our simulations, we use the Monte Carlo method to calculate this expectation term.

Let $\mV\in\R^{m_1\times m_2}$ be a $\rho$-rank matrix. Since the Gaussian width is also unitary invariant, we assume that $\mV$ takes the same form as $\mX$. The scaled spherical part of the subdifferential $\partial\|\mV\|_*$ is
\begin{equation*}
\partial\|\mV\|_*\cap t\S^{m_1m_2-1}=\left\{\begin{bmatrix}
\mI_\rho & \vzero\\
\vzero & \mW
\end{bmatrix}:~\|\mW\|\leq 1,~\|\mW\|_F^2=t^2-\rho\right\}.
\end{equation*}
Then spherical Gaussian width  of a scaled subdifferential is given by
\begin{align*}
\omega\left(\frac{1}{t}\partial\|\mV\|_*\cap\S^{m_1m_2-1}\right)&=\frac{1}{t}\cdot\omega(\partial\|\mV\|_*\cap t\S^{m_1m_2-1})\\
&=\frac{1}{t}\E \sup_{\mZ\in \partial\|\mV\|_*\cap t\S^{m_1m_2-1}}\ip{\mG}{\mZ}\\
&=\frac{1}{t}\E \sup_{\|\mW\|\leq 1,~\|\mW\|_F^2=t^2-\rho}\ip{\mG_2}{\mW}\\
&=\frac{1}{t}\sqrt{t^2-\rho}\cdot\E\|\mG_2\|_F=\sqrt{1-\frac{\rho}{t^2}}\cdot\mu_{(m_1-\rho)(m_2-\rho)}.
\end{align*}
The last line is due to the Cauchy-Schwarz inequality, i.e., $\ip{\mG_2}{\mW}\leq\|\mG_2\|_F\cdot\|\mW\|_F=\sqrt{t^2-\rho}\cdot\|\mG_2\|_F$. The notation $\mu_n$ denotes the expected length of an $n$-dimensional vector with independent standard normal entries.


\section{Auxiliary Lemma and Facts} \label{auxiliaryresults}
In this appendix, we present some additional auxiliary lemma and facts that are used in the proofs of our main results.

	\begin{lemma}\label{Lipschitz function}
		Let $\TT_1\subset\R^{n}$ and $\TT_2\subset\S^{m-1}$ be two subsets, $c\in\R$ is a constant. Then the functions
		$$
		F(\vg,\vh)=\min_{\va\in\TT_1\atop \vb\in\TT_2}\|\vg-c\cdot\va\|_2+\ip{\vh}{\vb}
		$$
		and
		$$
		G(\vg,\vh)=\max_{\va\in\TT_1\atop \vb\in\TT_2}\ip{\vh}{\vb}-\|\vg-c\cdot\va\|_2
		$$
		are $\sqrt{2}$-Lipschitz functions.
	\end{lemma}

	\begin{proof}
		To prove first part, it suffices to show that for any $(\vg_1,\vh_1), ~(\vg_2,\vh_2)$, we have
		\begin{align*}
		\big| F(\vg_1,\vh_1) - F(\vg_2,\vh_2) \big| \leq \sqrt{2}\sqrt{\|\vg_1-\vg_2\|_2^2+\|\vh_1-\vh_2\|_2^2}.
		\end{align*}
		To this end, let
		$$
		(\bar{\va},~\bar{\vb}) \in \arg\min_{\va \in \TT_1 \atop\vb\in\TT_2} \|\vg_2-c\cdot\va\|_2+\ip{\vh_2}{\vb}.
		$$
		Then we have
		\begin{align*}
		F(\vg_1,\vh_1) &= \min_{\va\in\TT_1\atop \vb\in\TT_2}\|\vg_1-c\cdot\va\|_2+\ip{\vh_1}{\vb}\\
		&\leq \|\vg_1-c\cdot\bar\va\|_2+\ip{\vh_1}{\bar\vb}.
		\end{align*}
		Therefore,
		\begin{align} \label{difference 1}
		F(\vg_1,\vh_1) - F(\vg_2,\vh_2)&\leq\left(\|\vg_1-c\cdot\bar\va\|_2+\ip{\vh_1}{\bar\vb}\right)-\left(\|\vg_2-c\cdot\bar\va\|_2+\ip{\vh_2}{\bar\vb}\right)\notag\\
		&= \left(\|\vg_1-c\cdot\bar\va\|_2-\|\vg_2-c\cdot\bar\va\|_2\right)+\left(\ip{\vh_1}{\bar\vb}-\ip{\vh_2}{\bar\vb}\right) \notag\\
		& \leq \|\vg_1-\vg_2\|_2+\|\vh_1-\vh_2\|_2\notag\\
		&\leq\sqrt{2}\sqrt{\|\vg_1-\vg_2\|_2^2+\|\vh_1-\vh_2\|_2^2}.
		\end{align}
		Similarly, we have
		\begin{equation} \label{difference 2}
		F(\vg_2,\vh_2) - F(\vg_1,\vh_1) \leq \sqrt{2}\sqrt{\|\vg_1-\vg_2\|_2^2+\|\vh_1-\vh_2\|_2^2}.
		\end{equation}
		Combining (\ref{difference 1}) and (\ref{difference 2}) yields the first conclusion. The second part for $G(\vg,\vh)$ can be proven similarly.
	\end{proof}

\begin{fact}[Polarity principle]\cite[Proposition 3.8]{oymak2018universality}
	\label{polarity argument}
	Let $\SS$ be a non-empty, closed, spherically convex subset of the unit sphere $\S^{n-1}$, and let $\mA:\R^{n}\to\R^m$ be a linear map. If $\cone(\SS)$ is not a subspace, then
	\begin{equation*}
	\min_{\|\vr\|_2=1}\min_{\vs\in(\cone(\SS))^\circ}\|\vs-\mA^T\vr\|_2 > 0 ~~\textrm{implies}~\vzero\in\mA(\SS),
	\end{equation*}
	where $\mA(\SS)$ is the image of $\mA$.
\end{fact}

\begin{fact}[Variance of Gaussian Lipschitz functions]\cite[Theorem 1.6.4]{bogachev1998gaussian}
	\label{Variance of Lipschitz}
	Consider a random vector $\vx\sim \NN(0,\mI_n)$ and a Lipschitz function $f:~\R^n\to\R$ with Lipschitz norm $\|f\|_{\textrm{Lip}}$ (with respect to the Euclidean metric). Then,
	\begin{equation*}
	\textrm{Var}(f(\vx))\leq \|f\|_{\textrm{Lip}}^2.
	\end{equation*}
\end{fact}

\begin{fact}[Moreau's decomposition theorem]\cite[Theorem 6.30]{bauschke2011convex}
	\label{Moreau}
	Let $\mathcal{C}$ be a nonempty closed convex cone in $\R^n$ and let $\vx\in\R^n$. Then
	\begin{align*}
	\|\vx\|_2^2=\dist^2(\vx,\mathcal{C})+\dist^2(\vx,\mathcal{C}^\circ).
	\end{align*}
\end{fact}

\begin{fact}\cite[Appendix A]{chandrasekaran2012convex}\label{omega & dist}
	Let $\mathcal{C}\subset \R^n$ be a non-empty convex cone, we have
	$$
	\sup_{\vz\in\mathcal{C}\cap\B_2^n}\ip{\vx}{\vz}=\dist(\vx,\mathcal{C}^\circ),
	$$
	where $\vx\in\R^n$.
\end{fact}

\begin{fact}\cite[Proposition 10.2]{amelunxen2014living}\label{relation_w_d}
	Let $\mathcal{C}$ be a convex cone. The Gaussian width and the statistical dimension are closely related:
	$$
	\delta(\mathcal{C})-1\leq\omega^2(\mathcal{C}\cap\S^{n-1})\leq\delta(\mathcal{C}).
	$$
\end{fact}

	\begin{fact}\cite[Lemma 3.7]{chandrasekaran2012convex}\label{Gs width of normal cone}
		Let $\mathcal{C}\subset\R^n$ be a non-empty closed, convex cone. Then we have that
		$$
		\omega^2(\mathcal{C}\cap\S^{n-1})+\omega^2(\mathcal{C}^\circ\cap\S^{n-1})\leq n.
		$$
	\end{fact}

\begin{fact}[Max–min inequality]\cite[Lemma 36.1]{rockafellar1970convex}\label{max-min}
	For any function $F:~\R^n\times\R^m\to\R$ and any $\mathcal{W}\subseteq\R^n$, $\mathcal{Z}\subseteq\R^m$, we have
	$$
	\sup_{\vz\in \mathcal{Z}}\inf_{\vw\in \mathcal{W}} F(\vw,\vz)\leq \inf_{\vw\in \mathcal{W}}\sup_{\vz\in \mathcal{Z}}F(\vw,\vz).
	$$
\end{fact}

\begin{fact}\cite[Theorem 4.3]{amelunxen2014living}\label{bound of delta}
	Let $f$ be a norm on $\R^n$, and fix a non-zero point $\vx\in\R^n$. Then
	\begin{align*}
	0\leq &\bigg(\min_{t\geq 0}\eta^2(t\cdot\partial f(\vx))\bigg)-\delta(\TT_f)\leq\frac{2\max\{\|\va\|_2:\va\in\partial f(\vx)\}}{f(\vx/\|\vx\|_2)}.
	\end{align*}	
\end{fact}

\bibliographystyle{IEEEtran}
\bibliography{IEEEabrv,myref}
\end{document}